\documentclass[12pt]{extarticle}

\usepackage[english]{babel}

\usepackage[letterpaper,top=2cm,bottom=2cm,left=3cm,right=3cm,marginparwidth=1.75cm]{geometry}

\usepackage[T1]{fontenc}
\usepackage{amsfonts,amsmath,amsthm,amssymb,mathtools}  %

\mathtoolsset{centercolon}
\usepackage{xfrac,nicefrac}
\usepackage{mathdots}
\usepackage{mleftright}  %
\let\left\mleft
\let\right\mright

\usepackage{xspace}
\xspaceaddexceptions{]\}}  %
\usepackage{regexpatch}

\usepackage{bm,bbm,dsfont}  %
\usepackage{caption}
\usepackage[normalem]{ulem}
\usepackage{enumitem}

\usepackage{graphicx}
\usepackage{float}
\usepackage{subcaption}  %
\usepackage{tcolorbox}
\usepackage{tikz}
\usetikzlibrary{decorations.pathreplacing}
\usetikzlibrary{calc}
\usetikzlibrary{positioning}
\usetikzlibrary{arrows.meta}

\usepackage[linesnumbered,boxed,ruled,vlined]{algorithm2e}
\usepackage{algpseudocode}

\usepackage{thmtools,thm-restate}
\theoremstyle{plain}

\newtheorem{theorem}{Theorem}[section]  %
\newtheorem{lemma}[theorem]{Lemma}

\theoremstyle{definition}  %
\newtheorem{definition}[theorem]{Definition}
\newtheorem{remark}[theorem]{Remark}

\usepackage[colorlinks,citecolor=blue,linkcolor=blue,urlcolor=red]{hyperref}
\usepackage[capitalise]{cleveref}
\crefname{algocf}{Algorithm}{Algorithms}
\Crefname{algocf}{Algorithm}{Algorithms}

\renewcommand{\tilde}{\widetilde}

\newcommand{\eps}{\varepsilon}

\renewcommand{\emptyset}{\varnothing}
\renewcommand{\epsilon}{\eps}

\usepackage{tabularx}
\usepackage{adjustbox}
\usepackage[numbers]{natbib}
\usepackage{amsmath, amssymb, amsthm}
\usepackage{appendix}
\newtheorem{corollary}[theorem]{Corollary}

{\newtheorem*{theorem*}{Theorem}}{}
{\newtheorem*{corollary*}{Corollary}}{}

\usepackage{hyperref}



\usepackage{tocloft}

\title{Enumerating Small Cycles}
\author{
Or Stern \\
Tel Aviv University
\and 
Or Zamir \\
Tel Aviv University
}
\date{}
\usepackage{hyperref}

\begin{document}

\maketitle

\begin{abstract}
    In a seminal result of Yuster and Zwick, they showed that for any fixed $k$, the even cycle $C_{2k}$ can be detected in an $n$-vertex graph in time $O(n^2)$.
    For $4$-cycles, a folklore algorithm extends to listing: for any $t$, we can list $t$ different $4$-cycles, if such exist, in $O(n^2+t)$ time.
    Recently, Jin, Vassilevska-Williams, and Zhou obtained similar bounds for listing $6$-cycles. 
    In this work, we generalize the above to cycles of sizes $8, 10, 12, 14,$ and $16$; we show that for all $k\leq 8$, we can list $t$ distinct $2k$-cycles in $\tilde{O}(n^2+t)$ time. In fact, our algorithm gives enumeration with pre-processing time $\tilde{O}(n^2)$ and delay $\tilde{O}(1)$.

    Additionally, for \emph{any} fixed $k$, we present an optimal enumeration (and hence also listing) algorithm for all cycles of size \emph{at most} $2k$. More generally, for any fixed $k$ and any $3\le i\le \frac{4k}{3}$, we present an algorithm with preprocessing time $\tilde{O}(n^2)$ and delay $\tilde{O}(1)$ that enumerates all cycles of sizes in the range $[i,2k]$.
\end{abstract}

\newpage
\tableofcontents
\newpage

\section{Introduction}
Listing small subgraphs is a central task in graph algorithm that also appears in many practical applications.  Already for cycles, the problem has several distinct variants.  In the \emph{detection} problem one only asks whether a cycle of the prescribed length exists; in the \emph{listing} problem, given a parameter $t$, one must output $t$ distinct copies if they exist, or else output all copies; and in the \emph{enumeration} problem one first preprocesses the graph and then outputs the copies one by one with small delay between consecutive outputs.  Algorithms for listing all simple cycles, without fixing their length, go back to the classical works of Tarjan~\cite{Tarjan1973Enumeration}, Johnson~\cite{Johnson1975Cycles}, Read and Tarjan~\cite{ReadTarjan1975Bounds}, and Mateti and Deo~\cite{MatetiDeo1976Circuits}; see also later optimal-output and practical variants~\cite{BirmaleEtAl2013Cycles,BlanusaIenneAtasu2022ParallelCycles,Grossi2016Enumeration}.  In this paper we focus on the complexity of listing and enumerating cycles of a fixed constant length.

A notable feature of cycle finding algorithms is the gap between even and odd cycles.  For every fixed $k$, Yuster and Zwick~\cite{YusterZwick1997EvenCycles} showed that an undirected $n$-vertex graph can be tested for the presence of a $C_{2k}$ in time $O(n^2)$, which is optimal in dense graphs.  
In contrast, detecting odd cycles of a prescribed length is closely connected to matrix multiplication and has no known algorithms faster than matrix multiplication time~$O(n^\omega)$~\cite{Alon1995Color,alon1997finding,DalirrooyfardVuongWilliams2021GraphPatterns}.  In sparse graphs, where the running time is measured as a function of the number of edges $m$, the classical work of Alon, Yuster, and Zwick~\cite{alon1997finding} and later improvements of Dahlgaard, Knudsen, and St{\"o}ckel~\cite{DahlgaardKnudsenStockel2017CappedWalks} give the best known detection bounds for even-cycle lengths; In this work, we focus on bounds in terms of the number of vertices~$n$ only.

The listing version is much less understood than detection.  
For triangles there is a large body of fine-grained work, including output-sensitive algorithms and conditional lower bounds~\cite{BjorcklundPaghVassilevskaWilliamsZwick2014Triangles,Patrascu2010Towards,KopelowitzPettiePorat2016Higher,VassilevskaWilliamsXu2020Monochromatic}.  For $4$-cycles, a simple folklore algorithm gives an $\widetilde{O}(n^2+t)$ bound, and recent works of Jin and Xu~\cite{JinXu2023RemovingAdditiveStructure} and Abboud, Khoury, Leibowitz, and Safier~\cite{AbboudKhouryLeibowitzSafier2023Listing4Cycles} give the sharper sparse-graph bound $\widetilde{O}(\min\{n^2,m^{4/3}\}+t)$; Here the pre-processing term corresponds to the best known time for detection.  
Recently, Jin, Vassilevska-Williams, and Zhou~\cite{Jin2024Listing} obtained an $\widetilde{O}(n^2+t)$ algorithm for listing $6$-cycles, nearly matching the $O(n^2)$ detection bound of Yuster and Zwick.  They explicitly left two questions open: whether the $C_6$ listing algorithm can be upgraded to polylogarithmic-delay enumeration after $\widetilde{O}(n^2)$ pre-processing, and whether the $\widetilde{O}(n^2+t)$ listing bound can be extended from $C_6$ to longer even cycles.
Here as well, several works study the case of sparse graphs and achieve improved bounds when the number of edges~$m$ is sufficiently small~\cite{alon1997finding,VassilevskaWilliamsWestover2025SparseC6,NakosNgoPanayi2026EvenCycles}, but these bounds are strictly worse than the $O(n^2)$ bound in a large range of edge densities. 

\subsection{Our Results}

Our first main result extends the optimal dense-graph listing bound from $C_4$ and $C_6$ to all even cycles through length $16$.

\begin{theorem*}[Even-cycle listing and enumeration]
For every fixed $2\le k\le 8$, there is a randomized algorithm which, given an $n$-vertex graph $G$ and a threshold $t$, lists $t$ distinct copies of $C_{2k}$, or all copies if fewer than $t$ exist, in time
\[
        \widetilde{O}(n^2+t)
\]
with high probability.  Moreover, for the same range of $k$, there is an enumeration algorithm with preprocessing time $\widetilde{O}(n^2)$ and delay $\widetilde{O}(1)$.
\end{theorem*}

For $k=2$ and $k=3$ the listing bounds are the known $C_4$ and $C_6$ bounds, respectively.  Our reductions upgrade these listing algorithms to enumeration with near-quadratic preprocessing and polylogarithmic delay; in particular, this resolves the delay-enumeration question for $C_6$ raised by Jin, Vassilevska Williams, and Zhou~\cite{Jin2024Listing}.  The new listing results are for $k=4,5,6,7,8$, namely $C_8,C_{10},C_{12},C_{14}$, and $C_{16}$.

Our second main result is a similar enumeration algorithm that works for \emph{any} fixed~$k$, but lists all cycles of length \emph{at most} $2k$. 
Note that this is a weaker variant than enumeration of cycles of only one specific size: exact-$C_{2k}$ enumeration (for every~$k$) would imply up-to-$C_{2k}$ enumeration, but not vice versa; as a graph may for example contain many~$C_4$'s and very few~$C_6$'s, and we could care only about listing these few~$C_6$'s.

\begin{theorem*}[Range cycle listing and enumeration]
For every fixed $k\ge 2$, there is a randomized algorithm which, given an $n$-vertex graph $G$ and a threshold $t$, lists $t$ distinct cycles of size at most~$2k$ in~$G$, or all copies if fewer than $t$ exist, in time
\[
        \widetilde{O}(n^2+t)
\]
with high probability.  Moreover, there is an enumeration algorithm with preprocessing time $\widetilde{O}(n^2)$ and delay $\widetilde{O}(1)$.

In fact, the same is true even if we want to enumerate only even cycles of size at most~$2k$, or alternatively all or all even cycles with sizes in the range~$[i,2k]$ for any~$3\leq i \leq \frac{4k}{3}$.
\end{theorem*}

The main technical object in the paper is a layered path-reporting data structure.  Given a $(k+1)$-layer graph
\[
        A_1-A_2-\cdots-A_{k+1},
\]
the data structure preprocesses the graph so that, for every pair $(u,v)\in A_1\times A_{k+1}$, it can decide in $O(1)$ time whether there is an $A_1$-to-$A_{k+1}$ path from $u$ to $v$, and can report all such paths in time linear in the number of reported paths.

The reduction from cycle listing to path reporting is completely general: for every fixed $k$, if such a data structure for $(k+1)$-layer graphs can be constructed in time $\widetilde{O}(n^2+t_{2k})$, where $t_{2k}$ is the number of $2k$-cycles in the layered graph, then $C_{2k}$ listing and enumeration follow with the same near-quadratic preprocessing bound.  The difficulty is constructing the data structure quickly.  Our technical contribution is the following theorem.

\begin{theorem*}[Path-reporting data structure]
For every fixed $2\le k\le 8$, the path-reporting data structure for $(k+1)$-layer graphs can be constructed in time
\[
        \widetilde{O}(n^2+t_{2k}),
\]
where $t_{2k}$ denotes the total number of $C_{2k}$s in the layered graph.
Furthermore, for any fixed~$k\ge 9$ the same data structure can be constructed in time
\[
        \widetilde{O}(n^2+t_4+t_6+\ldots+t_{2k}).
\]
\end{theorem*}

Our construction proceeds by recursively compressing sub-paths and charging the sizes of the resulting tables to actual $2k$-cycles in the original graph.  
While the construction works for every~$k$, a simpler analysis only bounds the construction time by the number of even cycles of size at most~$2k$ in the graph -- which is sufficient for obtaining the results on listing cycles of all sizes up to $2k$.
In the first few cases though, of constructing the data structure with up to $9$ layers, which is exactly what is needed for $k\le 8$, we are able to give a tighter analysis which shows the running time is bounded by the number of cycles of only the specific size~$2k$.  We also identify a barrier for the present technique: the natural extension needed to generalize our tighter analysis to $10$ layers is false.  This does not rule out an $\widetilde{O}(n^2+t)$ algorithm for $C_{18}$-listing, but it shows that extending our current path-reporting construction beyond $C_{16}$ with a tight analysis requires new ideas.

The same path-reporting data structure gives an asymmetric reduction for odd cycles, allowing us to list or detect $(2k-1)$-cycles in time that is proportional to the number of $2k$-cycles in the same graph.

\begin{theorem*}[Odd-cycle listing]
For every fixed $2\le k\le 8$, there is a randomized algorithm which lists $t$ copies of $C_{2k-1}$ in time
\[
        \widetilde{O}(n^2+t_{2k}(G)+t),
\]
where $t_{2k}(G)$ is the number of $2k$-cycles in the input graph.
\end{theorem*}

In particular, this gives an odd-cycle detection algorithm running in time $\widetilde{O}(n^2+t_{2k}(G))$.  This is useful when the next even-cycle count is small.  For example, in an $m$-edge graph with the ``typical'' number of $2k$-cycles, as in a random graph with the same edge density, one expects
\[
        t_{2k}(G) \approx \left(\frac{m}{n}\right)^{2k}.
\]
In this regime our detection time is $\widetilde{O}\!\bigl(n^2+(\frac{m}{n})^{2k}\bigr)$, which improves over the classical $\widetilde{O}(m^{2-1/k})$ odd-cycle detection bound~\cite{alon1997finding} throughout the range
\[
        n^{1+1/(2k-1)} \ll m \ll n^{2k^2/(2k^2-2k+1)}.
\]
Thus, although the odd-cycle bound is not uniformly better in the worst case, it gives faster detection in a nontrivial range of densities for graphs with a typical number of even cycles.

\subsection{Overview and Organization}

The starting point is the path-reporting viewpoint.  After color coding the input graph~\cite{Alon1995Color}, a $C_{2k}$ can be viewed as two internally disjoint length-$k$ paths with common endpoints.  If the color classes are arranged as
\[
        A_1,A_{2,1},\ldots,A_{k,1},A_{k+1},A_{k,2},\ldots,A_{2,2},
\]
then the two halves of the cycle are two $(k+1)$-layer graphs sharing only the endpoint layers $A_1$ and $A_{k+1}$.  Given path-reporting data structures for both halves, we iterate over all endpoint pairs $(u,v)$, report all paths from $u$ to $v$ in the first half and all paths from $u$ to $v$ in the second half, and output their Cartesian product.  The time spent on non-output work is $O(n^2)$ for the endpoint pairs, and the output work is proportional to the number of cycles.  The complete reduction and its implications appear in Section~\ref{sec:cycle_to_path_reporting} and is independent of the later restriction $k\le 8$ which comes solely from the construction of the data structure. This section also contains reductions to odd-cycle listing and range cycle listing.

Sections~\ref{sec:three_four_layers} and~\ref{sec:five_layers} develop the path-reporting data structure.  For three layers, the table simply stores all common middle vertices for each pair of endpoints -- the total size of the table entries as well as the construction time is proportional to the number of two-edge paths in the graph, which is at most an additive factor of $n^2$ larger than the number of $4$-cycles.  
For four layers, we closely follow the algorithm of Jin, Vassilevska Williams, and Zhou~\cite{Jin2024Listing} and reframe it as the construction of our path-reporting data structure: paths are divided into those with a replaceable internal vertex and those without such a replacement, which we call sparse paths.  Replaceable paths are handled by auxiliary graphs obtained by compressing two-layer intervals with many witnesses -- such a compression is fine as every two compressed edges can be expended into disjoint paths via the several witnesses guaranteed for each compression; sparse paths are stored explicitly -- which is fine as there aren't too many of them. 
At this point we already generalize the analysis of~\cite{Jin2024Listing} and show that every table constructed by the data structure has size $\widetilde{O}(n^2+t_{2k})$, for any fixed~$k$. This is in contrast to the original analysis that only bounded these sizes by $\widetilde{O}(n^2+t_{6})$ which is sufficient for listing $6$-cycles, but is insufficient when we later use this algorithm recursively to list cycles of a larger size. 

The first genuinely new case is the five-layer structure of Section~\ref{sec:five_layers},
which gives the \(C_8\) algorithm. A length-\(4\) path may fail to be sparse in several
different ways, and the data structure treats these failures separately.  Paths with a
single replaceable internal vertex, and paths where an adjacent pair can be replaced
by a sparse \(3\)-path, are handled by auxiliary graphs and recursive calls to the
previous data structures.  Paths where an adjacent pair can be replaced only through
a non-sparse \(3\)-path require a new bookkeeping device: they are coupled to paths
that are already listed, and are recovered during listing queries using an additional
table.  The remaining sparse \(4\)-paths are stored explicitly.  The size analysis again
charges table entries to genuine \(C_8\)'s in the original graph; the main extra care is
that recursive calls are now made on auxiliary graphs whose edges do not all correspond
to original graph edges.

Section~\ref{sec:general_larger_cycles} extends the construction from five layers to
an arbitrary constant number of layers.  We introduce a new notion of a \emph{mildly sparse
path}.  Sparse paths are rigid enough for \(C_8\), but not for deeper recursive auxiliary
graphs: an auxiliary edge may represent many possible original paths, and several such
edges may have to be expanded without collisions to charge the cost to actual simple cycles in the original graph.  
Mild sparsity is a recursive
relaxation that keeps the number of competing replacements for every proper subpath
under control.  This guarantees that whenever an auxiliary edge is created, it has many
internally vertex-disjoint representatives in the original graph.  With this replacement,
the same recursive strategy can be described for all constants: build shorter path
tables, contract intervals with many mildly sparse representatives, store the remaining
mildly sparse paths explicitly, and use auxiliary tables to recover paths that are
not themselves stored but differ from listed paths by a bounded mildly sparse
replacement.

The remaining work is to prove that all tables created by this recursive construction
are small. Section~\ref{sec:complexity_boundary} does this by a charging argument:
if a table is too large, then many of its entries can be charged to actual simple
\(2k\)-cycles in the original graph.  For tables in the original layered graph, this is
proved directly by fixing one endpoint, pruning low-degree vertices near the other
endpoint, and closing stored paths into cycles.  For tables inside auxiliary graphs, the
same argument reduces to the following structural question: when a dense compressed
graph connects the two ends of a short block of original layers, must there be an
actual path through that block of the exact length needed to close a \(2k\)-cycle?  
At this point, we are already able to prove the result for all cycles of size at most~$2k$, 
as there we don't mind charging the cost to any even cycle of smaller size. This is done in Section~\ref{subsec:range_result}.

The condition needed to avoid charging smaller cycles and charging only those of size \emph{exactly}~$2k$ is more difficult to prove. We call this condition \emph{the boundary layer property}.
Section~\ref{sec:complexity_boundary} shows
that proving this property for blocks of at most \(s+1\) original layers gives the full
data structure with the tighter analysis for every \(k\le 2s+2\).
We can think of this additional complexity as analogous to the gap between the
Moore bound, which implies that every \(n\)-vertex graph with
\(\Omega_k(n^{1+1/k})\) edges contains a cycle of length at most
\(2k\)~\cite{AlonHooryLinial2002}, and the substantially more technical theorem
of Bondy and Simonovits, which guarantees a cycle of length exactly
\(2k\) under the same density assumption~\cite{BondySimonovits1974}.

Section~\ref{sec:boundary_layer_proofs} proves the needed structural property for
\(s=2\) and \(s=3\).  The case \(s=2\), involving three original layers, gives
\(C_{10}\) and \(C_{12}\).  The case \(s=3\), involving four original layers, gives
\(C_{14}\) and \(C_{16}\).  In both cases, a path in the compressed graph usually
expands directly to a path of the desired length in the original graph.  When the
length has the wrong residue, the proof either finds a short shortcut correcting the
length, or shows that the absence of such shortcuts forces many edges between the
internal layers; after a simple pruning step, those edges contain the required path.

Finally, Section~\ref{sec:s4_barrier} explains why this method stops at \(C_{16}\).
The next case, \(C_{18}\), would require the same structural property for five original
layers.  We give a dense five-layer construction in which the relevant compressed graph
is large, but every path between the two end layers has length divisible by \(4\).  Thus
the exact path length needed for \(C_{18}\) cannot be forced by the current abstraction.
This does not rule out an \(\widetilde O(n^2+t)\)-time algorithm for \(C_{18}\), but it
shows that simply extending the present path-reporting framework beyond \(C_{16}\) does not result in tight bounds without introducing more tools.

\section{Preliminaries}
We consider a simple, undirected graph $G = (V, E)$, where $V$ denotes the set of vertices with $n = |V|$ and $E$ denotes the set of edges with $m = |E|$. 
\paragraph{}
    For any vertex $v \in V$, its $\mathbf{neighborhood}$ is denoted by $N(v) = \{u \in V \mid (u, v) \in E\}$, and its degree is given by $\deg(v) = |N(v)|$. For any subset of vertices $S \subseteq V$, we extend this notation to define the neighborhood of a set as $N(S) = \bigcup_{v \in S} N(v)$. Under this convention, the two-step neighborhood of a single vertex $v$ is naturally expressed as $N(N(v))$, which corresponds to the set of all vertices reachable from $v$ by a path of length exactly $2$.
\paragraph{}
  For two vertex sets $V,U$, $\mathbf {E(V,U)} $ is the set of edges between the vertex sets, and $\mathbf{e(V,U)}$ as the number of edges between the vertex sets. $E(G)$ for a graph $G=(V,E)$ is defined to be $E(V,V)$.
\paragraph{}
For a graph $G=(V,E)$ and a subset $U\subset V$, the induced subgraph $\mathbf{G\setminus U}$ is the graph created by removing the vertex set $U$ from G, as well as all the edges in $E(U,V \setminus U)$.
Similarly, for a graph $G=(V,E)$ and a subset $U\subset V$, $\mathbf{G[U]}$ is the induced subgraph created by removing the vertex set $G\setminus U$ from G, as well as all the edges in $E(U,V \setminus U)$.
\paragraph{}
A path $P$ of length $l$ is a sequence of vertices $v_0, v_1, \dots, v_l$ such that $(v_i, v_{i+1}) \in E$ for all $0 \le i < l$. The path $P$ is simple if all its vertices are pairwise distinct. A cycle of length $c$ (or a $c$-cycle) is a simple path of length $c$ with the additional closing edge $(v_l, v_0)$. An even cycle is a cycle whose length $c$ is an even integer.
For a graph $G=(V,E)$ and a constant $k$, $t_{2k}(G)$ is the number of $2k$-cycles in the graph $G$.

\paragraph{Color-Coding Framework:}
To isolate specific paths and avoid unwanted self-intersections during the algorithm, our algorithm utilizes the color-coding technique introduced by Alon, Yuster, and Zwick \cite{Alon1995Color}. Given a target cycle length $2k$, we partition the vertex set $V$ by assigning each vertex a color uniformly and independently at random from a set of $2k$ colors. A subgraph (or path) is said to be \textbf{colorful} if all its vertices are assigned distinct colors.
\paragraph{Layered Graphs:}
A \textbf{layered graph} is a graph $G=(V,E)$ whose vertex set $V$ is partitioned into $r$ disjoint layers, denoted $V = A_1 \cup A_2 \cup \dots \cup A_r$ for a constant $r$. The edge set $E$ is restricted such that edges only exist between consecutive layers; formally, if $(u,v) \in E$ with $u \in A_i$ and $v \in A_j$, then $|i - j| = 1$. We additionally define the $i$-th \textbf{layer transition} as the edge subset $E[A_i, A_{i+1}]$. Furthermore, we denote $E[A_1, A_2]$ and $E[A_{r-1}, A_r]$ specifically as the \textbf{boundary layer transitions} of the graph.
\paragraph{Number of Cycles in the Graph:}
For a graph $G=(V,E)$, the number of $2k$-cycles in the graph $G$ is denoted $t_{2k}(G)$. 

\newcommand{\SixColorGraph}{\adjustbox{valign=m, raise=-0.6mm, margin=-1mm 0 0 0}{
    \begin{tikzpicture}[scale=0.5,every vertex/.style={inner sep=1pt}]
      \vertex (A) at (0,0) {$A$};
      \vertex (D) at (3.6,0) {$D$};
      \vertex (B1) at (1.2,0.5) {$B_1$};
      \vertex (B2) at (1.2,-0.5) {$B_2$};
      \vertex (C1) at (2.4,0.5) {$C_1$};
      \vertex (C2) at (2.4,-0.5) {$C_2$};
      \draw[-] (A) -- (B1) -- (C1) -- (D);
      \draw[-] (A) -- (B2) -- (C2) -- (D);
    \end{tikzpicture}}}
     
\newcommand{\TwoPaths}{\adjustbox{valign=m, raise=-0.6mm, margin=-1mm 0 0 0}{
    \begin{tikzpicture}[scale=0.5,every vertex/.style={inner sep=1pt}]
      \vertex (a) at (0,0) {$a$};
      \vertex (d) at (3.6,0) {$d$};
      \vertex (b1) at (1.2,0.5) {$b_1$};
      \vertex (b2) at (1.2,-0.5) {$b_2$};
      \vertex (c1) at (2.4,0.5) {$c_1$};
      \vertex (c2) at (2.4,-0.5) {$c_2$};
      \draw[-] (a) -- (b1) -- (c1) -- (d);
      \draw[-] (a) -- (b2) -- (c2) -- (d);
    \end{tikzpicture}}}
\newcommand{\PathForP}{\adjustbox{valign=m, raise=-0.6mm, margin=-1mm 0 0 0}{
    \begin{tikzpicture}[scale=0.5,every vertex/.style={inner sep=1pt}]
      \vertex (a) at (0,0) {$a$};
      \vertex (b) at (1,0) {$b$};
      \vertex (d) at (3.4,0) {$d$};
      \vertex (c1) at (2.2,0.5) {$c_1$};
      \vertex (c2) at (2.2,-0.5) {$c_2$};
      \draw[-] (a) -- (b) -- (c1) -- (d);
      \draw[-] (b) -- (c2) -- (d);
    \end{tikzpicture}}}

\newcommand{\PathForQ}{\adjustbox{valign=m, raise=-0.6mm, margin=-1mm 0 0 0}{
    \begin{tikzpicture}[scale=0.5,every vertex/.style={inner sep=1pt}]
      \vertex (a) at (0,0) {$a$};
      \vertex (d) at (3.4,0) {$d$};
      \vertex (b) at (1.2,0.5) {$b_1$};
      \vertex (b') at (1.2,-0.5) {$b_2$};
      \vertex (c) at (2.4,0) {$c$};
      \draw[-] (a) -- (b) -- (c) -- (d);
      \draw[-] (a) -- (b') -- (c);
    \end{tikzpicture}}}
    \newcommand{\PathForPNotR}{\adjustbox{valign=m, raise=-0.6mm, margin=-1mm 0 0 0}{
    \begin{tikzpicture}[scale=0.5,every vertex/.style={inner sep=1pt}]
      \vertex (a) at (0,0) {$a$};
      \vertex (b) at (1,0) {$b$};
      \vertex (d) at (3.4,0) {$d$};
      \vertex (c) at (2.2,0.5) {$c_1$};
      \vertex (c') at (2.2,-0.5) {$c_2$};
      \draw[-] (a) -- (b) -- (c) -- (d);
      \draw[-] (b) -- (c') -- (d);
    \end{tikzpicture}}}

\newcommand{\PathForQNotR}{\adjustbox{valign=m, raise=-0.6mm, margin=-1mm 0 0 0}{
    \begin{tikzpicture}[scale=0.5,every vertex/.style={inner sep=1pt}]
      \vertex (a) at (0,0) {$a$};
      \vertex (d) at (3.4,0) {$d$};
      \vertex (b1) at (1.2,0.5) {$b_1$};
      \vertex (b2) at (1.2,-0.5) {$b_2$};
      \vertex (c) at (2.4,0) {$c$};
      \draw[-] (a) -- (b1) -- (c) -- (d);
      \draw[-] (a) -- (b2) -- (c);
    \end{tikzpicture}}}
\newcommand{\SixCycleForProof}{\adjustbox{valign=m, raise=-0.6mm, margin=-1mm 0 0 0}{
    \begin{tikzpicture}[scale=0.5,every vertex/.style={inner sep=1pt}]
      \vertex (a) at (0,0) {$a$};
      \vertex (b1) at (1,1) {$b_1$};
      \vertex (c1) at (2,0.5) {$c_1$};
      \vertex (b2) at (1,0) {$b_2$};
      \vertex (c2) at (2,-0.5) {$c_2$};
      \vertex (b3) at (1,-1) {$b_3$};
      \draw[-] (a) -- (b1) -- (c1);
      \draw[-] (b2) -- (c1);
      \draw[-] (b2) -- (c2);
      \draw[-] (a) -- (b3) -- (c2);
      
    \end{tikzpicture}}}

\section{Reductions to Path Reporting}\label{sec:cycle_to_path_reporting}

This section introduces our path-reporting data structure, and contains the reductions that turn the path-reporting data structure into the cycle-listing and cycle-enumeration results of the paper.  These reductions are independent of the later construction of the data structure: for every fixed $k$, once the appropriate path-reporting data structure for layered graphs can be built within the desired time bound, the corresponding cycle-listing and enumeration algorithms follow.  We first present two simple reductions between listing variants, then use color coding to pass to layered graphs, and finally introduce the data structure then reduce even and odd cycle listing to path reporting.

\subsection{Reduction from Enumeration to Threshold Listing}

\begin{theorem}
Given an algorithm $\mathcal{A}$ that lists up to $t$ distinct $2k$-cycles in a graph $G$ within $c \cdot (n^2 + t)\log^2 n$ operations, where $c > 0$ is a global constant, we can construct an enumeration algorithm for $2k$-cycles that requires $O(n^2\log^2n)$ preprocessing time and achieves a worst-case delay of $O(\log^2 n)$ between consecutive outputs.
\end{theorem}

\begin{proof}
We modify the listing algorithm $\mathcal{A}$ such that for $2k$-cycles listed by $\mathcal{A}$, we check the hash table $\mathcal{H}$ containing all previously enumerated cycles, as well as a balanced binary tree $\mathcal{T}$ containing cycles found by $\mathcal{A}$ but not yet enumerated. New cycles are inserted into the balanced binary search tree $\mathcal{T}$. Since $k$ is a fixed constant, checking both the hash table and the tree, as well as inserting into both $\mathcal{H}$ and $\mathcal{T}$, takes $O(2k \log(n^{2k})) = O(\log n)$ time per cycle. Let $\tau(t) = c(n^2 + t)\log^2 n$ denote the upper bound for the number of operations of the modified $\mathcal{A}$.

We maintain a dynamic state using an index tracker $i$, initially set to $i = 1$.

\paragraph{Preprocessing($G$):}
\begin{enumerate}
    \item Initialize an empty hash table $\mathcal{H}$ to keep enumerated cycles, and an empty binary search tree $\mathcal{T}$ to keep cycles found by a listing algorithm which have not yet been enumerated.
    \item Execute the algorithm $\mathcal{A}$ with an initial cycle amount $t_0 = n^2$, inserting all discovered cycles into $\mathcal{T}$. This requires $\tau(n^2) = 2c n^2 \log^2 n = \tilde{O}(n^2)$ time.
    \item Initialize an instance of the listing algorithm, denoted $\mathcal{A}_1$, with a cycle amount of $t_1 = 4n^2$. Pause $\mathcal{A}_1$ at its initial state.
\end{enumerate}

\paragraph{Enumerate:}
When a request for the next $2k$-cycle is received, the algorithm performs the following sequence:
\begin{enumerate}
    \item \textbf{Tree Check:} If $\mathcal{T}$ is empty, the algorithm terminates and reports that all $2k$-cycles have been enumerated. Otherwise, extract an arbitrary cycle $C \in \mathcal{T}$, insert $C$ into $\mathcal{H}$, and output $C$. This takes $O(\log n)$ operations.
    \item \textbf{Continuing listing algorithm:} If there is an active instance $\mathcal{A}_i$, advance its execution by exactly $\Delta = 20c \log^2 n$ operations. Any cycles found are filtered through both $\mathcal{H}$ and $\mathcal{T}$ (to make sure they are new) and are then placed into $\mathcal{T}$.
    \item \textbf{Transition from one listing algorithm to the next:} If the active instance $\mathcal{A}_i$ terminates during this step (meaning it has listed the $t_i = 4^i n^2$ cycles or exhausted all cycles in the graph), we execute the following transition steps:
    \begin{itemize}
        \item Stop $\mathcal{A}_i$.
        \item If $4^i n^2 < n^{2k}$ and $\mathcal{T}$ is not empty, increment the phase tracker to $i \leftarrow i + 1$.
        \item Initiate a new listing instance $\mathcal{A}_{i+1}$ configured with an increased threshold target of $t_{i+1} = 4^{i+1} n^2$, and pause it.
    \end{itemize}
\end{enumerate}

We prove by induction that the tree $\mathcal{T}$ cannot become empty before all cycles in $G$ are enumerated. 

For the first $n^2$ calls to \texttt{Enumerate}, the tree draws from the $t_0 = n^2$ cycles generated during preprocessing. Over these $n^2$ calls, instance $\mathcal{A}_1$ conducts $n^2 \cdot \Delta = 20c n^2 \log^2 n$ operations. Because the total runtime required for $\mathcal{A}_1$ is bounded by:
\[
\tau(t_1) = \tau(4n^2) = c(n^2 + 4n^2)\log^2 n = 5c n^2 \log^2 n
\]
the allocated $20c n^2 \log^2 n$ steps guarantee that $\mathcal{A}_1$ completes before the $n^2$-th enumeration occurs. Upon completion, if $t_{2k}(G) \ge 4n^2$, $\mathcal{A}_1$ has inserted $4n^2$ $2k$-cycles into $\mathcal{T}$. Accounting for the $n^2$ $2k$-cycles enumerated during this phase, $\mathcal{T}$ is guaranteed to retain at least $4n^2 - n^2 = 3n^2$ cycles at the moment $\mathcal{A}_1$ terminates.

Now the algorithm initializes $\mathcal{A}_2$ with threshold $t_2 = 16n^2$. 
\\
By the inductive hypothesis, we assume that after $\mathcal{A}_{i}$ completes, the buffer $\mathcal{T}$ now contains at least $4^{i-1}n^2$ cycles. During the phase where these cycles are being enumerated, the newly spawned instance $\mathcal{A}_{i+1}$ will conduct 
\[
4^{i-1}n^2 \cdot \Delta = 20c \cdot 4^{i-1}n^2 \log^2 n = 5c \cdot 4^i n^2 \log^2 n
\]
operations. 
The total number of operations needed for $\mathcal{A}_{i+1}$ to run to completion from its fresh initialization state is:
\[
\tau(t_{i+1}) = c(n^2 + 4^{i+1} n^2)\log^2 n \le 5c \cdot 4^i n^2 \log^2 n.
\]
Since $5c \cdot 4^i n^2 \log^2 n \ge \tau(t_{i+1})$, the instance $\mathcal{A}_{i+1}$ is guaranteed to reach completion before the tree runs out of cycles. If $t_{2k}(G) \ge 4^{i+1} n^2$, it ensures that at least $4^{i+1} n^2-\sum_{j=0}^{i-1}4^{j} n^2\ge4^i n^2$ (accounting for the at most $\sum_{j=0}^{i-1}4^{j} n^2$ cycles enumerated during past phases) cycles will be in $\mathcal{T}$ before $\mathcal{A}_{i+2}$ is initialized. If $t_{2k}(G)<4^{i+1}n^2$, the instance $\mathcal{A}_{i+1}$ enters all cycles into $\mathcal{T}$ except those already in $\mathcal{T}$ or $\mathcal{H}$, so the algorithm will continue to enumerate the rest of the cycles. This shows that once $\mathcal{T}$ is empty, all cycles have been enumerated.

The enumeration function performs $\Delta = O(\log^2 n)$ operations of a listing algorithm, plus $O(\log ^2n)$ tree maintenance (deletion/insertions) per call. This establishes a worst-case delay of $O(\log^2 n)$, completing the proof.
\end{proof}
\subsection{Listing $t$ Cycles for a Given Parameter $t$}

We first establish that the problem of listing up to $t$ distinct $2k$-cycles in $O((n^2 + t)\log^2n)$ time can be reduced to the problem of listing \emph{all} $2k$-cycles in a graph $G$ in $O((n^2 + t_{2k}(G)) \log n)$ time. Let $\mathcal{A}$ be an algorithm that lists all $2k$-cycles in any given graph $G$ within a running time bounded by $c \cdot (n^2 + t_{2k}(G)) \log n$, where $c$ is a fixed constant. Denote $\tau = c \cdot (n^2 + t) \log n$ for the given threshold $t$. This reduction adapts the framework introduced by Jin~\cite{Jin2024Listing}. 

Let the vertex set of the input graph $G=(V,E)$ be indexed as $V = \{v_1, v_2, \dots, v_n\}$, and for any $j \in [n]$, let $G_j = G[\{v_1, \dots, v_j\}]$ denote the subgraph induced by the first $j$ vertices. We execute a binary search over the index space $[1, n]$ to find an index $j$ for which listing all $2k$-cycles in $G_j$ takes at most $\tau$, but listing all $2k$-cycles $G_{j+1}$ takes more than $\tau$. Let $L = 1$, $R = n$, and $M = \lfloor (L+R)/2 \rfloor$. At each step of the binary search, we run algorithm $\mathcal{A}$ on the induced subgraph $G_M$ with a maximal time of $\tau$ operations. 

If $\mathcal{A}$ terminates normally within $\tau$ operations and $t_{2k}(G_M) \ge t$, we output the $t$ cycles and complete the algorithm. Otherwise, if $\mathcal{A}$ terminates normally within $\tau$ operations, it implies that $t_{2k}(G_M) \le t$; we then update the lower bound by setting $L = M$.  Conversely, if $\mathcal{A}$ does not terminate within $\tau$ time steps, it guarantees that $t_{2k}(G_M) > t$; we then update the upper bound by setting $R = M$. Because each execution takes at most $O((n^2 + t)\log n)$ time and the binary search space halves at each step, this search concludes in $O(\log n)$ iterations, identifying an index $i \in [n]$ such that $t_{2k}(G_i) \le t$ and $t_{2k}(G_{i+1}) > t$.

We then conduct the following based on the value of $i$:
\begin{itemize}
    \item \textbf{Case 1 ($i = n$):} Since $t_{2k}(G_n) \le t$, we simply run $\mathcal{A}$ on the entire graph $G$. The algorithm finishes within $O((n^2 + t)\log n)$ time and returns all $t_{2k}(G)$ cycles.
    
    \item \textbf{Case 2 ($i < n$):} We execute $\mathcal{A}$ on the induced subgraph $G_{i+1}$ with a modified maximal time of $\tau' = c \cdot (n^2 + 2t) \log n$, which allows the algorithm to run to completion if $t_{2k}(G_{i+1}) \le 2t$.
    \begin{itemize}
        \item If $\mathcal{A}$ terminates within $\tau'$ operations, it returns all cycles in $G_{i+1}$. Because $t_{2k}(G_{i+1}) > t$, the output contains more than $t$ distinct cycles. We return any $t$ of these $2k$-cycles.
        
        \item If $\mathcal{A}$ fails to terminate within $\tau'$, it implies that $t_{2k}(G_{i+1}) > 2t$. Since $t_{2k}(G_i) \le t$, the single vertex $v_{i+1}$ must belong in more than $2t - t = t$ distinct $2k$-cycles within $G_{i+1}$. We can therefore use the same localized cycle-listing algorithm of Jin, Williams, and Zhou~\cite{Jin2024Listing}, which lists $t$ $2k$-cycles passing through a fixed vertex $v$ in $\tilde{O}((m + {t})\log n)$ time.
    \end{itemize}
\end{itemize}
Since the binary search requires $O(\log n)$ iterations of $O((n^2 + t)\log n)$ time, and the last step requires at most $O((n^2 + t)\log n)$ time, the reduction increases the total time complexity by only a factor of $O(\log n)$, preserving the $\tilde{O}(n^2 + t)$ complexity.
\begin{remark}
The reduction from cycle enumeration to listing all cycles works the same for cycles from a range of constant sizes.
\end{remark}
\subsection{Color Coding}

Next, we establish a logarithmic reduction from listing all $2k$-cycles in a general graph $G$ to listing all colorful $2k$-cycles in a $2k$-layered graph. The method of color coding we use is originates from ~\cite{Alon1995Color}. The reduction we prove is similar to the one shown in \cite{Jin2024Listing}.
\begin{theorem}
    Assume we are given an algorithm that lists all $2k$-cycles in a $2k$-layer graph in $O(n^2+t_{2k}(G))$ time, we use it to construct an algorithm which lists all $2k$-cycles in any graph $G$ in $O((n^2+t_{2k}(G))\log n)$ time.
\end{theorem}
\begin{proof}
Given an input graph $G=(V,E)$, we repeat the following process $10 \cdot (2k)^{2k+1}\log n$ times:
We assign each vertex a color chosen uniformly at random from the set $\{1, 2, \dots, 2k\}$. Let $V_i$ denote the set of vertices assigned color $i$. We then apply the $2k$-cycle listing algorithm for layered graphs to find all colorful $2k$-cycles in the graph $G' = (V_1 \cup V_2 \cup \dots \cup V_{2k}, E')$, where $E'$ contains only those edges that connect vertices in adjacent or neighboring color classes.

In any given iteration, a specific $2k$-cycle is colorful and properly ordered with a probability of at least $\frac{1}{(2k)^{2k}}$. To ensure cycles are not reported multiple times across different iterations, we store discovered cycles in a hash table. After $10 \cdot (2k)^{2k+1}\log n$ iterations, the probability that a specific cycle has never been listed is bounded by:
\[
\left(1-\frac{1}{(2k)^{2k}}\right)^{10 \cdot (2k)^{2k+1}\log n} \le e^{-20k\log n} \le (n^{-2k})^{10}
\]
Since there are at most $n^{2k}$ total possible $2k$-cycles in the graph, a union bound shows that the probability of the algorithm failing to list any existing cycle is at most $\frac{1}{n^{18k}}$. Thus, the algorithm succeeds with high probability, specifically $1-\frac{1}{n^{18k}}$.
\end{proof}

\subsection{The Path Reporting Data Structure and Even Cycles}

We now present a simple reduction to a path reporting data structure problem. The main contribution of this paper is the data structure's introduction, construction, and complexity analysis.

\begin{definition}
    Given a $(k+1)$-layer graph $G=(V,E)$ where $V = A_1 \cup A_2 \cup \dots \cup A_{k+1}$, we define the path-listing data structure $\mathcal{D}(G)$, which supports the following functions:
\begin{itemize}
    \item $\texttt{construct}(G)$: this function is used once and constructs the required data structure for the graph G, with time complexity $O(n^2+t_{2k}(G))$. 
    \item $\texttt{exist-path}(a_1, a_{k+1})$: Returns whether a path exists between vertex $a_1 \in A_1$ and vertex $a_{k+1} \in A_{k+1}$ in $O(1)$ time.
    \item $\texttt{all-paths}(a_1, a_{k+1})$: Lists all paths between vertex $a_1 \in A_1$ and vertex $a_{k+1} \in A_{k+1}$ in $O(|P|)$ time, where $P$ is the set of paths connecting $a_1$ and $a_{k+1}$.
\end{itemize}
\end{definition}

The primary contribution of this paper is the reduction to our from cycle listing to creating the path-reporting data structure, and the construction of this data structure in $O(n^2+t_{2k}(G))$ time for $k\in[4,8]$. We note that the techniques from the work of~\cite{Jin2024Listing} are comparable to constructing a similar data structure specifically for $4$-layer graphs.
\begin{theorem} \label{reduction_to_path_reporting_structure}
Given an algorithm to construct the path listing data structure $\mathcal{D}(G)$ for a $(k+1)$-layer graph $G$ in $O(n^2+t_{2k}(G))$ time,  we construct an algorithm listing all $2k$ colorful cycles in a $2k$-layer graph $H$ in $O(n^2+t_{2k}(H))$ time.
\end{theorem}
\begin{remark}
Note that $t_{2k}(H)$ includes all $2k$-cycles in $H$, not only colorful cycles.
\end{remark}
\begin{proof}
Let $H$ be a $2k$-layer graph. We denote the ordered layers of $H$ as:
\[
(A_1, A_{2,1}, \dots, A_{k,1}, A_{k+1}, A_{k,2}, \dots, A_{2,2})
\]
This setup can be viewed as splitting the layers of $H$ into two distinct paths that each contain $k-1$ intermediate layers, bound together by two shared layers, $A_1$ and $A_{k+1}$. To locate all colorful cycles, we construct the data structure $\mathcal{D}$ twice: once for the induced subgraph $G_1 = H[A_1\cup A_{2,1}\cup \dots \cup A_{k,1}\cup A_{k+1}]$ and once for the induced subgraph $G_2 = H[A_1 \cup A_{2,2}\cup \dots \cup  A_{k,2}\cup A_{k+1}]$.

Then, for every pair of vertices $(a_1, a_{k+1}) \in A_1 \times A_{k+1}$, we execute the following verification steps:
\begin{enumerate}
    \item Query $\texttt{exist-path}(a_1, a_{k+1})$ in $\mathcal{D}(G_1)$. If it returns \texttt{false}, no colorful cycle containing this vertex pair exists in $H$.
    \item Query $\texttt{exist-path}(a_1, a_{k+1})$ in $\mathcal{D}(G_2)$. If it returns \texttt{false}, no colorful cycle containing this vertex pair exists in $H$.
\end{enumerate}

If both queries return \texttt{true}, we execute $\texttt{all-paths}(a_1, a_{k+1})$ on both data structures. For each path $P_1$ returned from $G_1$ and each path $P_2$ returned from $G_2$, we combine them to output the complete cycle $P_1 - P_2$. Let $p_1$ and $p_2$ denote the number of paths found in $G_1$ and $G_2$, respectively. The total number of cycles listed for this vertex pair is $p_1 p_2$, and the time complexity to list them is $O(p_1 p_2)$. Summing over all vertex pairs, the total execution time of the algorithm is bounded by $O(n^2+t_{2k}(H))$, as required.
\end{proof}

\subsection{Odd Cycles from Path Reporting}\label{subsec:odd_cycle_reduction}

The same path-reporting viewpoint also gives an asymmetric reduction for odd cycles.  A $(2k-1)$-cycle can be decomposed into two paths with common endpoints, one of length $k$ and the other of length $k-1$.  Thus one half uses a $(k+1)$-layer path-reporting data structure and the other uses only a $k$-layer path-reporting data structure.

\begin{theorem}\label{thm:odd_cycle_listing_reduction}
Assume that the path-reporting data structures needed for $k$- and $(k+1)$-layer graphs can be constructed in time $\widetilde O(n^2+t_{2k}(G))$.  Then, for every fixed $k$, one can list $t$ copies of $C_{2k-1}$ in an $n$-vertex graph $G$ in time
\[
        \widetilde O(n^2+t_{2k}(G)+t).
\]
In particular, for $2\le k\le 8$, our constructions give this bound.
\end{theorem}

\begin{proof}
We use color coding with $2k-1$ colors.  In one color-coding trial, arrange the relevant color classes as
\[
        A_1,A_{2,1},\ldots,A_{k,1},A_{k+1},A_{k-1,2},\ldots,A_{2,2},
\]
where the first half
\[
        A_1-A_{2,1}-\cdots-A_{k,1}-A_{k+1}
\]
has $k+1$ layers, and the second half
\[
        A_1-A_{2,2}-\cdots-A_{k-1,2}-A_{k+1}
\]
has $k$ layers.  Build the path-reporting data structure on both halves.  For every endpoint pair $(a_1,a_{k+1})$, list all paths between the pair in both halves and output their Cartesian product. This algorithm outputs all colorful $(2k-1)$-cycles and takes $\widetilde O(n^2+t_{2k}(G)+t_{2k-1}(G))$. 

The correctness is the same as in the even-cycle reduction: the two halves have disjoint internal color classes, and every properly colored $(2k-1)$-cycle decomposes uniquely into one path of length $k$ and one path of length $k-1$. The construction time is $\widetilde O(n^2+t_{2k}(G))$, because both halves are subgraphs of the color-coded graph and the data-structure construction is charged to the number of $2k$-cycles. The reporting time is linear in the number of odd cycles output. Using the same reductions from color-coding and threshold listing, we get an $\widetilde O(n^2+t_{2k}(G)+t)$ algorithm for odd cycles listing.
\end{proof}

For detection, set $t=1$.  The running time becomes $\widetilde O(n^2+t_{2k}(G))$, which is useful when the next even-cycle count is small.  In an $m$-edge graph with the typical number of $2k$-cycles, as in a random graph with the same edge density, one expects
\[
        t_{2k}(G)\approx \left(\frac mn\right)^{2k}.
\]
In this regime our odd-cycle detection time is
\[
        \widetilde O\!\left(n^2+\left(\frac mn\right)^{2k}\right).
\]
Comparing this with the classical sparse-graph bound $\widetilde O(m^{2-1/k})$ for detecting $C_{2k-1}$~\cite{alon1997finding}, our bound is smaller throughout the range
\[
        n^{1+1/(2k-1)} \ll m \ll n^{2k^2/(2k^2-2k+1)},
\]
provided the number of $2k$-cycles is of the typical order above.  Thus the odd-cycle application is not a uniform worst-case improvement, but it gives faster detection in a nontrivial density range for graphs with few enough next-length even cycles.

\subsection{Range cycle listing}
Building on our path reporting data structure, we present an algorithm which lists all cycles whose length are in a chosen large-enough range. Given a constant $k$ and any $3\le i\le \frac{4k}{3}$ we list all cycles in the graph with lengths within the range $[i,2k]$. Assuming we have an algorithm to construct the path reporting data structure, we use Theorem~\ref{reduction_to_path_reporting_structure} to perform the cycle listing algorithm for every even number in the range $[i,2k]$. We then also use the algorithm described in Theorem~\ref{thm:odd_cycle_listing_reduction} to list all odd cycles in the size range.
\begin{theorem} \label{approx_listing}
Given an algorithm $\mathcal{A}$ which creates the path-reporting data structure for an $l$ layer graph in $\widetilde{O}(n^2+t_{2\cdot \lceil \frac{i}{2}\rceil }(G)+t_{2\cdot \lceil \frac{i}{2}\rceil+2}(G)+\dots t_{2k}(G))$ for every $3\le l \le k+1$ (the cycles counted in the complexity are cycles of even size between $i$ and $2k$), we construct a cycle listing algorithm that lists all cycles whose size in the range $[i,2k]$, in time $\widetilde{O}(n^2+t_{i}(G)+t_{i+1}(G)+\dots t_{2k}(G))$.
\end{theorem}

We first describe the algorithm for listing the even cycles, then the algorithm for listing the odd cycles. We note that listing the odd cycles takes $\widetilde{O}(n^2+t_{i}(G)+t_{i+1}(G)+\dots t_{2k}(G))$, while listing the even cycles takes only $\widetilde{O}(n^2+t_{2\cdot \lceil \frac{i}{2}\rceil }(G)+t_{2\cdot \lceil \frac{i}{2}\rceil+2}(G)+\dots t_{2k}(G))$. In other words, the time for listing the odd cycles includes the number of even cycles as well as odd cycles, but the time for listing the even cycles includes only the number of even cycles. 
\begin{lemma}\label{even_range}
Using $\mathcal{A}$, we construct a listing algorithm outputting only even cycles in the range, whose time is quasi-linear in the number of cycles listed plus $\widetilde{O}(n^2)$. 
\end{lemma}
\begin{proof}

The color-coding reduction and path reporting reduction work the same if the complexity is replaced with $\widetilde{O}(n^2+t_{2\cdot \lceil \frac{i}{2}\rceil }(G)+t_{2\cdot \lceil \frac{i}{2}\rceil+2}(G)+\dots t_{2k}(G))$. So for each $\frac{i}{2}+1\le l \le k+1$, given $\mathcal{A}$, we construct a $2(l-1)$-cycle listing algorithm with the same complexity. Going over all even values of $l$ in $[i,2k]$ we get  a cycle listing algorithm that lists all even cycles whose size in the range $[i,2k]$, in time $\widetilde{O}(n^2+t_{2\cdot \lceil \frac{i}{2}\rceil }(G)+t_{2\cdot \lceil \frac{i}{2}\rceil+2}(G)+\dots t_{2k}(G))$.
\end{proof}
Additionally, for odd-cycle listing, we use the same algorithm as in Theorem~\ref{thm:odd_cycle_listing_reduction}. For listing cycles of size $i\le j<2k$ where $j$ is odd, we construct data structures of size $\frac{j-1}{2}+1$ and $\frac{j+1}{2}+1$.
\begin{remark}
    In the case of $j=3$, note that the construction of the data structure for a $2$ layer graph is trivial (the data structure is an adjacency matrix).
\end{remark}
The sizes of both data structures are between $2$ and $k+1$, so according to our assumption they are created in $\widetilde{O}(n^2+t_{2\cdot \lceil \frac{i}{2}\rceil }(G)+t_{2\cdot \lceil \frac{i}{2}\rceil+2}(G)+\dots t_{2k}(G))$. Listing all $j$-cycles in the graph, the algorithm works in $\widetilde{O}(n^2+t_{2\cdot \lceil \frac{i}{2}\rceil }(G)+t_{2\cdot \lceil \frac{i}{2}\rceil+2}(G)+\dots t_{2k}(G)+t_j(G))$. Doing this for every odd number in $[i,2k]$, as well as every even number in the range as described in the previous paragraph, we have an algorithm which lists all cycles whose size is in the range $[i,2k]$, and takes $\widetilde{O}(n^2+t_{i}(G)+t_{i+1}(G)+\dots t_{2k}(G))$.

The same reductions from enumeration and threshold listing work here, so given $\mathcal{A}$ as in the theorem we get an enumeration algorithm that returns cycles in the size range $[i,2k]$ with $\widetilde{O}(n^2)$ preprocessing and $\widetilde{O}(1)$ delay. Additionally, using Lemma~\ref{even_range}, as well as the reductions from enumeration, we get an enumeration algorithm that returns even cycles in the size range $[i,2k]$ with $\widetilde{O}(n^2)$ preprocessing and $\widetilde{O}(1)$ delay.\\
Most of the article will discuss the algorithm needed for exact cycle listing. The result needed for this theorem ($\mathcal{A}$ which creates the path-reporting data structure for an $l$ layer graph in $\widetilde{O}(n^2+t_{2\cdot \lceil \frac{i}{2}\rceil }(G)+t_{2\cdot \lceil \frac{i}{2}\rceil+2}(G)+\dots t_{2k}(G))$ for every $3\le l \le k+1$) is shown in Subsection~\ref{subsec:range_result}.

\paragraph{Roadmap to Constructing the Data Structure:}
We present a recursive framework for constructing a data structure that supports path lookup in $k$-layer graphs. We begin by detailing the data structures for 3-layer and 4-layer graphs, adapting techniques originally introduced by Jin, Vassilevska-Williams, and Zhou~\cite{Jin2024Listing} for 6-cycle listing. We also prove tight complexity bounds that allow us to advance to constructing $\mathcal{D}$ for graphs with more layers. Following this foundation, we describe the 5-layer data structure using analogous techniques and necessary new components. Finally, we extend the approach to graphs with more layers, introducing new tools, mainly necessary for the complexity analysis.

\section{Constructing $\mathcal{D}$ for Graphs with 3 or 4 Layers}\label{sec:three_four_layers}
We begin by describing the construction of $\mathcal{D}(G)$ for graphs with 3 or 4 layers, which is a reformulation of the 6-cycle listing algorithm introduced by Jin, Vassilevska-Williams, and Zhou \cite{Jin2024Listing}. Additionally, as those will be used later in our data structure for a larger number of layers, we prove tighter bounds on the size of the generated data structure than those proved by~\cite{Jin2024Listing}. Specifically, we prove that the size of the data structure is bounded by $O(n^2+t_{2k}(G))$ for any constant $k$, rather than just $O(n^2+t_{6}(G))$ as in the original analysis.

\subsection{Data Structure for Three layers} \label{alg1}

Consider a 3-layer graph $G=(V = A_1 \cup A_2 \cup A_3, E)$. We construct $\mathcal{D}(G)$ which consists of a table $T$ containing an entry for every vertex pair $(a_1, a_3) \in A_1 \times A_3$.
Each entry $T[a_1, a_3]$ stores all vertices $a_2' \in A_2$ such that $a_1 - a_2' - a_3$ forms a $2$-path in $G$. The algorithm to build $T$ is  detailed in Algorithm~\ref{alg:compute_tablesT}.

\begin{algorithm}[htbp] 
  \caption{Compute the tables $T$ (on layers $A_1, A_2, A_3$)}
  \label{alg:compute_tablesT}
  \DontPrintSemicolon
  
  \SetKwProg{Fn}{Function}{:}{}
  \For(\Comment*{Compute table $T_{1,3}(G)$}){$a_2 \in A_2$} {
    $X_1 = \emptyset, X_3 = \emptyset$\;
    \For{$a_1 \in N(a_2) \cap A_1$} {
      Insert $a_1$ into $X_1$\;
    }
    \For{$a_3 \in N(a_2) \cap A_3$} {
      Insert $a_3$ into $X_3$\;
    }
    \For{$(a_1, a_3) \in X_1 \times X_3$} {
      Insert $a_2$ into $T_{1,3}(G)[a_1, a_3]$ \Comment*{$a_2$ is a common neighbor of $a_1$ and $a_3$} \label{line:insert1} 
    }
  }
\end{algorithm}

In words, for each vertex $a_2 \in A_2$, we collect its neighbors in $A_1$ and $A_3$, and then insert $a_2$ into $T_{1,3}(G)[a_1, a_3]$ for every pair of its neighbors $(a_1, a_3)$. The running time of the algorithm is $O(n^2+s)$, where $s$ is the size of the table $T_{1,3}(G)$.

\begin{theorem} \label{size_of_2-path_table}
Fix any constant $k$, the size of $T_{1,3}(G)$ is $O(n^2 + t_{2k}(G))$.
\end{theorem}

\begin{proof}
If |$T_{1,3}(G)| \le 100kn^2$, the claim holds trivially. Otherwise, there must exist a vertex $a_1 \in A_1$ such that the number of edges between its neighborhood in $A_2$ and the step-two neighborhood in $A_3$ exceeds $100kn$. That is, $e(N(a_1) \cap A_2, N(N(a_1)) \cap A_3) > 100kn$. This follows from the fact that $\sum_{a_1 \in A_1} e(N(a_1) \cap A_2, N(N(a_1)) \cap A_3) = |T_{1,3}(G)|$.

We construct a subgraph $G_{a_1}$ containing exclusively the edges in $E(a_1, N(a_1) \cap A_2)$ and $E(N(a_1) \cap A_2, N(N(a_1)) \cap A_3)$. This subgraph contains exactly one vertex from $A_1$, and every vertex $a_2 \in A_2$ that is connected to $a_1$. We prove that there are sufficiently many $2k$-cycles in the graph $G_{a_1}$ that use the vertex $a_1$, by describing an algorithm that finds such cycles iteratively. The process will include iteratively conducting a pruning process to remove vertices with low degree, and afterwards finding a cycle. The process stops when the graph has less than $100kn$ edges, and we will show that $t_{2k}(G_{a_1})=\Omega(E(G_{a_1})-110kn)$.

\paragraph{Pruning Process:}
Since $G_{a_1}$ has more than $100kn$ edges, we conduct a pruning process on the graph resulting in a subgraph in which the minimum degree is at least~$10k$. During this process, at each step we remove a vertex with degree below $10k$, if one exists, as well as all the edges adjacent to it. We continue this process until we are left with a graph that contains no such vertices. This pruning process removes at most $10k|U|$ edges, where $U$ is the number of removed vertices. We call the new graph $G_{a_1}'$. If the graph now contains less than $100kn$ edges, we stop the algorithm. Note that $a_1$ cannot be removed during this process; if it were, at most $10k-1$ vertices from $A_2$ would remain in the graph, bounding the remaining edge count to $10kn$, which contradicts our lower bound of $100kn$. Furthermore, $G_{a_1}'$ must be connected: any isolated component lacking $a_1$ can contain at most one vertex from $A_2$, which is hence a tree that has many leaves of degree~$1$, contradicting the guarantee we have on the minimum degree.
\paragraph{Finding Cycles:}
Focusing on $G'_{a_1}[A_2\cup A_3]$, the minimum degree guarantee ensures the existence of an even simple path~$P$ of length $2(k-1)$ starting and ending within $A_2$, as such path can be greedily constructed. By combining an edge from $a_1$ to an endpoint of~$P$, traversing~$P$ to its other endpoint, and returning from it to $a_1$, we form a cycle of length $2k$. 

\paragraph{Process of Finding Cycles:}
For length $2k$, we can sequentially find these cycles and remove the part of them in $A_2\times A_3$ from the graph, erasing a constant number of edges each time. After each erasure, we prune any vertices whose degrees fall below $10k$ and repeat this process with the new graph. We continue this algorithm until the final graph contains less than $100kn$ edges. The number of edges removed by pruning is at most $10kn$, because for each vertex pruned (at any stage in the algorithm) we delete at most $10k$ edges.


\paragraph{Bounding $\mathbf{E({G_{a_1}})}$:}
This ensures that $e({G_{a_1}}) - 110kn = O(t_{2k}(G_{a_1}))$. The reason being is that for each edge that is not either one of the $10kn$ edges removed by pruning, or one of the $100kn$ edges when the algorithm stops, we must've removed that edge from the graph as part of a $2k$-cycles. Summing this outcome across every vertex in $A_1$ (where each vertex yields distinct cycles due to the uniqueness of $a_1$), we find that the total number of $2$-paths is $O(n^2 + t_{2k}(G))$.

The intuition of this proof is illustrated in Figure~\ref{pic:NALG} (note that vertex degrees are downscaled for visual clarity).

\begin{figure}
    \centering
    \includegraphics[width=0.4\linewidth]{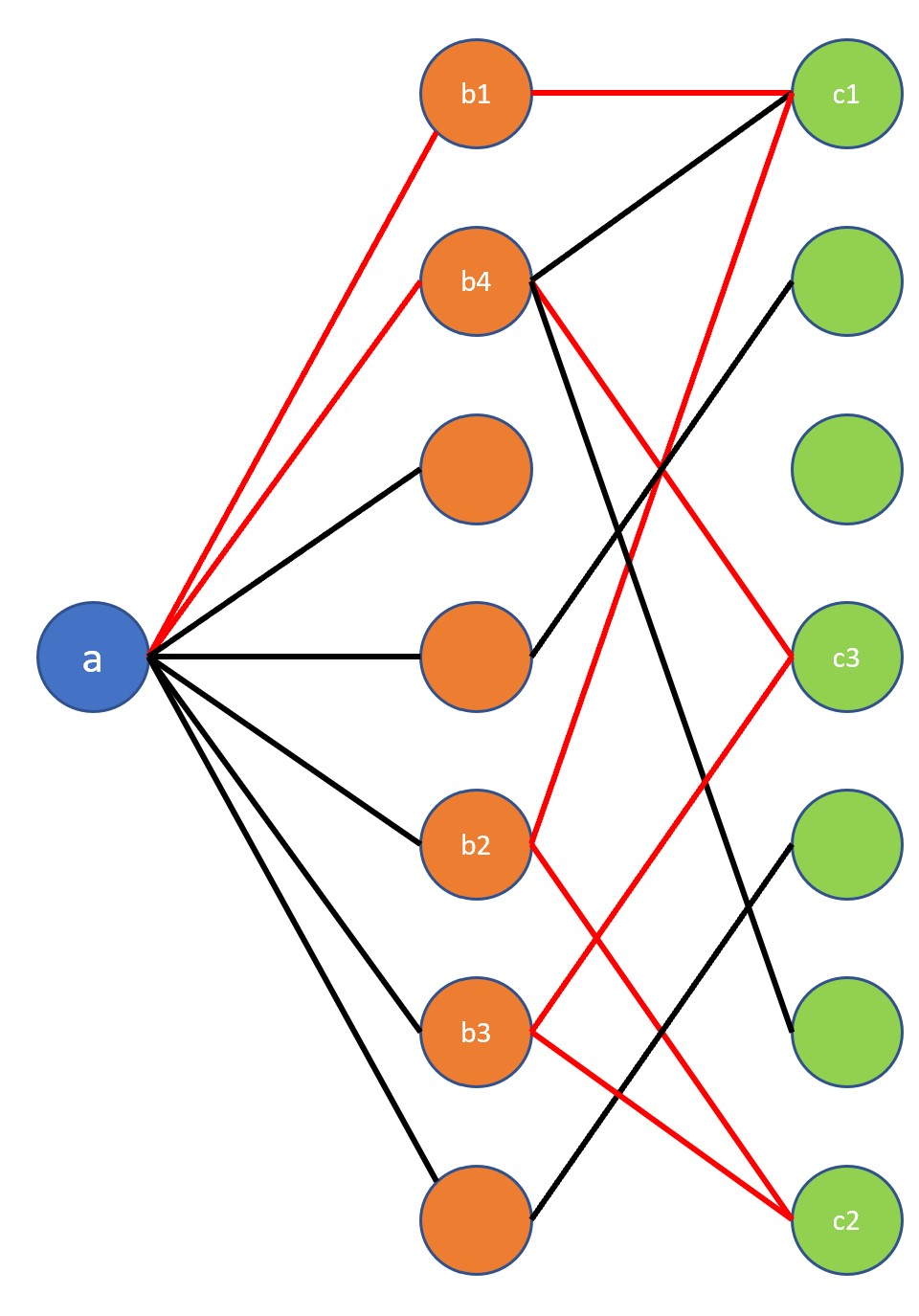}
    \caption{Proof for size of T table}
    \label{pic:NALG}
\end{figure}
\end{proof}
\paragraph{}
For an $r$-layer graph $G=(V,E)$, $V=A_1\cup \dots \cup A_r$, $t'_{2k}(G)$ is defined as the number of $2k$-cycles in $G$ that have $2$ edges in every layer transition other than $E[A_{r-1},A_r]$ and $2k-2(r-2)$ edges in the layer transition $E[A_{r-1},A_r]$. Additionally, $t''_{2k}(G)$ is defined as the number of $2k$-cycles in $G$ that have $2$ edges in every layer transition other than $E[A_1,A_2]$ and $2k-2(r-2)$ edges in the layer transition $E[A_1,A_2]$.

\begin{theorem} \label{important_remark}
Fix any constant $k$, the size of $T_{1,3}(G)$ is $O(n^2 + t'_{2k}(G))$.
\end{theorem}
\begin{proof}
The proof of Theorem~\ref{size_of_2-path_table} proves this theorem as well.
\begin{remark}  \label{other_side}
We also observe that this proof can symmetrically be conducted via selecting a graph based on vertices in $A_3$ as opposed to vertices in $A_1$, and conducting a similar proof on graphs $G_{a_3}$ for $a_3\in A_3$. When an arbitrary $2k$-cycle is isolated, we can explicitly choose whether two edges stem from $E[A_1, A_2]$ and the remainder from $E[A_2, A_3]$, or vice versa. This property becomes critical in later stages. This also proves the size of $T_{1,3}(G)$ is $O(n^2 + t''_{2k}(G))$.
\end{remark}
\end{proof}
Using the table $T$, the query $\texttt{exist-path}(a_1,a_3)$ for $a_1 \in A_1$ and $a_3 \in A_3$ takes $O(1)$ time, and the query $\texttt{all-paths}(a_1,a_3)$ for $a_1 \in A_1$ and $a_3 \in A_3$ takes $O(P)$ time, where $P$ is the number of listed paths.

\subsection{Sparse Paths}
Before moving on to the 4 layer data structure, we define the notion of sparse paths, which plays a central part in the construction of these data structures.

Given a sequence of vertex layers $(A_1, A_2, \dots, A_r)$, a colorful path $P$ from $a_i \in A_i$ to $a_j \in A_j$ is denoted \textbf{sparse} if there exists no alternative path from $a_i$ to $a_j$ that contains any internal vertices of $P$ \textbf{excluding the endpoints $a_i$ and $a_j$}. By definition, all $2$-paths are sparse, though this property does not automatically hold for paths of length 3 or greater. 
Additionally, given a layered graph $G=(V,E)$ with at least $j$ layers, we will construct tables denoted as $T_{i,j}(G)$, which are defined such that each entry $T_{i,j}(G)[a_i, a_j]$ stores all sparse paths connecting $a_i\in A_i$ and $a_j\in A_j$.

\subsection{Data Structure for Four layers}

We next describe the construction of $\mathcal{D}(G)$ for a 4-layer graph $G=(V = A_1 \cup A_2 \cup A_3 \cup A_4, E)$.
In this construction, we partition the paths into two types and construct a data structure answering queries separately for each of the two types.

We define the type of a path~$a_1 - a_2 - a_3 - a_4$ in~$G$ as follows:
\begin{enumerate}
    \item[\textbf{1}] If the path contains an internal vertex $a_i$ ($i \in \{2, 3\}$) that can be replaced by an alternative vertex $a_i'$ to yield another path, we call it a Type 1 path. 
    \item[\textbf{2}] Otherwise we call it a Type 2 path.
\end{enumerate}
For each of these types we will build a part of the data structure $\mathcal{D}(G)$, and each part will be able to answer: \\
\begin{enumerate}
    \item Does there exist a path in this type (1 or 2) from $a_1\in A_1$ to $a_4\in A_4$?
    \item List all paths from $a_1\in A_1$ to $a_4\in A_4$ in this type.
\end{enumerate}
\subsubsection{Part 1: Handling Type 1 Paths}
In this part, we present the data structure and queries designed to handle paths belonging to type 1. These are length-3 paths $a_1 - a_2 - a_3 - a_4$ between $A_1$ and $A_4$ where one intermediate vertex ($a_2$ or $a_3$) can be replaced by an alternative vertex to form a different path.

For each segment $A_i - A_{i+1} - A_{i+2}$ (where $i \in \{1, 2\}$), we first construct the table $T_{i,i+2}(G)$ as defined in the previous section. From these tables, we derive an auxiliary graph $G^{T_{i,i+2}(G)}$ as follows:
\begin{enumerate}
    \item All edges that are not adjacent to the intermediate layer $A_{i+1}$ remain unchanged.
    \item The intermediate layer $A_{i+1}$ is removed from the vertex set.
    \item For any vertex pair $(a_i, a_{i+2}) \in A_i \times A_{i+2}$, the auxiliary graph $G^{T_{i,i+2}(G)}$ contains an edge $(a_i, a_{i+2})$ if and only if $|T_{i,i+2}(G)[a_i, a_{i+2}]| > 1$. That is, an edge exists if and only if there are at least two distinct intermediate vertices in $A_{i+1}$ providing paths between $a_i$ and $a_{i+2}$ in the original graph $G$.
\end{enumerate}
Figures~\ref{pic:Original_graph_G} and Figure~\ref{pic:G'} provide an example of this transformation.

\begin{figure}
    \centering
    \includegraphics[width=0.6\linewidth]{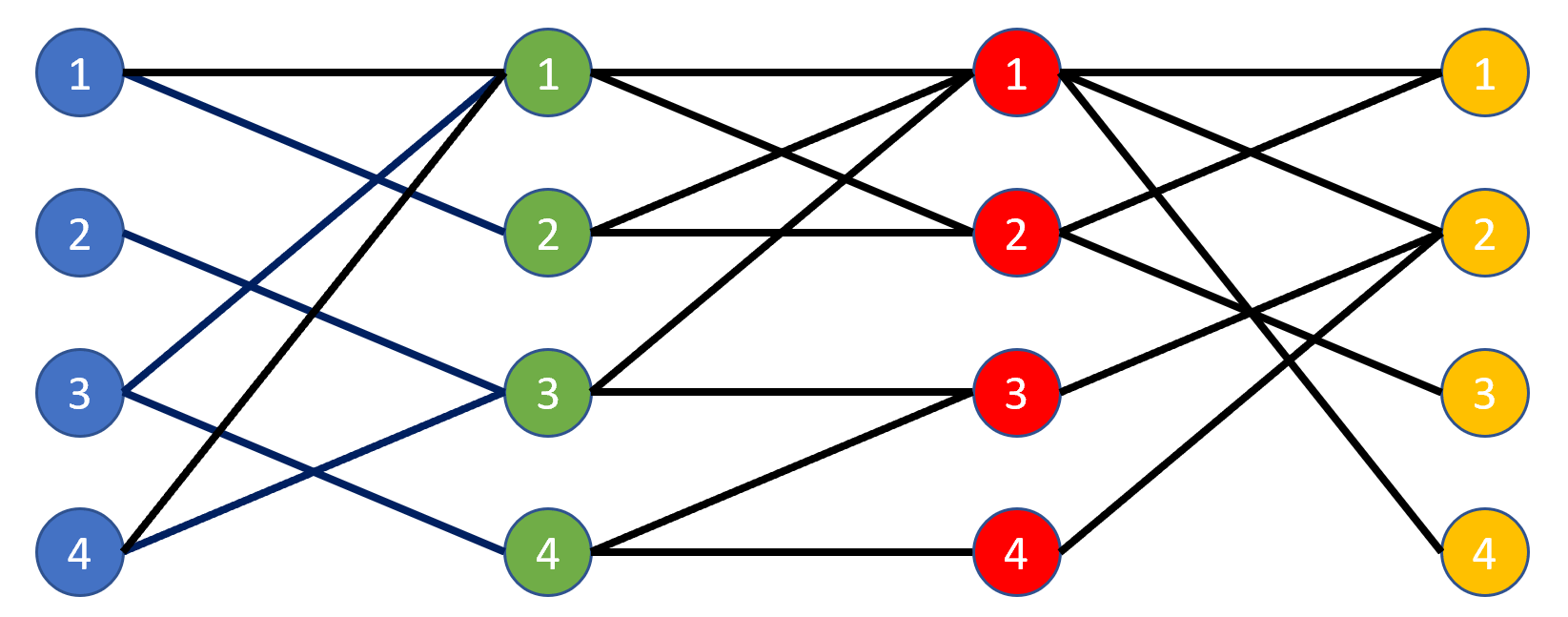}
    \caption{Original graph G}
    \label{pic:Original_graph_G}
\end{figure}

\begin{figure}
    \centering
    \includegraphics[width=0.6\linewidth]{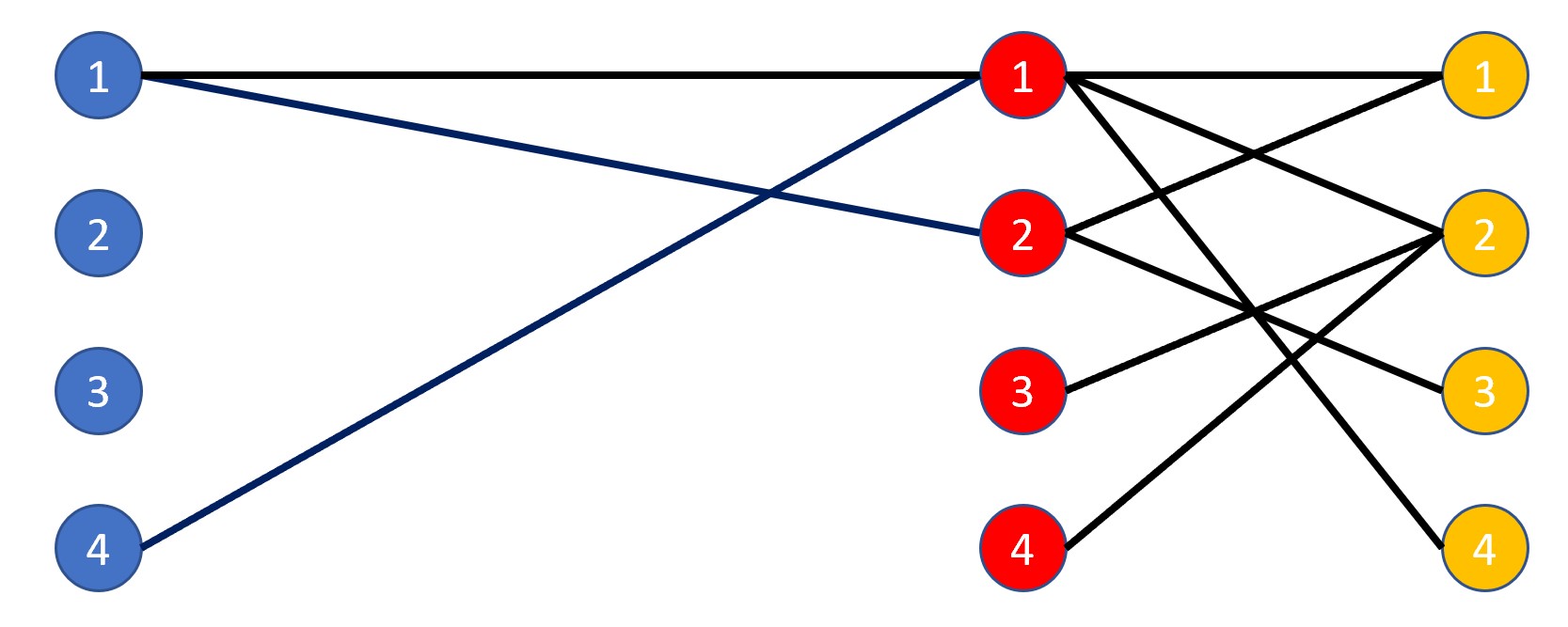}
    \caption{Graph $G^{T(G)}$ is described in the above section}
    \label{pic:G'}
\end{figure}

Notice that $G^{T_{i,i+2}(G)}$ is a graph structured over three sequential layers. We can now compute our standard 3-layer data structure on this graph. As established previously, this construction takes $O(n^2 + t_{2k}(G^{T_{i,i+2}(G)}))$ time for $k>1$. In this 3-layer graph, one pair of adjacent layers corresponds directly to original graph edges, while the other pair contains auxiliary edges representing multiple $2$-paths. 

\paragraph{Queries on $\mathcal{D}(G)$ for Type 1 Paths:}
To check whether there exists a type 1 path between $a_1 \in A_1$ and $a_4 \in A_4$, we query $\texttt{exist-path}(a_1, a_4)$ in $\mathcal{D}(G^{T_{i,i+2}(G)})$ for both $i \in \{1, 2\}$. If either query returns \texttt{true}, there is a $2$-path in that auxiliary graph, which by definition guarantees the existence of a type 1 path in $G$.

To implement $\texttt{all-paths}(a_1, a_4)$ for type 1 paths, we run $\texttt{all-paths}(a_1, a_4)$ on $\mathcal{D}(G^{T_{i,i+2}(G)})$ for both values of $i$:
\begin{itemize}
    \item For $i=1$, each retrieved path has the form $a_1 - a_3 - a_4$. We look up the table entry $T_{1,3}(G)[a_1, a_3]$ to find all intermediate vertices $a_2 \in A_2$, and list the reconstructed paths $a_1 - a_2 - a_3 - a_4$.
    \item For $i=2$, each retrieved path has the form $a_1 - a_2 - a_4$. We look up the table entry $T_{2,4}(G)[a_2, a_4]$ to find all intermediate vertices $a_3 \in A_3$, and list the reconstructed paths $a_1 - a_2 - a_3 - a_4$.
\end{itemize}

\begin{remark}
A type 1 path can be generated by the queries for both values of $i$. To prevent duplicate listings, we store the discovered paths in a hash table.
\end{remark}

\paragraph{Complexity Analysis:} 
We now bound the size of the tables of the auxiliary graphs, $T_{1,3}(G^{T_{i,i+2}(G)})$. 
Assume without loss of generality that $i=1$, meaning the auxiliary graph is $G^{T_{1,3}(G)}$, its layers are $(A_1, A_3, A_4)$ and the edges between $A_1$ and $A_3$ are the auxiliary edges. The size of the table $T_{1,3}(G^{T_{1,3}(G)})$ is $O(n^2+t'_{2k-2}(G^{T_{1,3}(G)}))$ for a fixed constant $k$, where $t'$ represents cycles which have 2 edges in the set $E[A_1,A_3]$ in the auxiliary graph, and the rest of the edges in the set $E[A_3,A_4]$. This follows from Theorem~\ref{important_remark}.
\begin{lemma}
    $t'_{2k-2}(G^{T_{1,3}(G)})\le t_{2k}(G)$ which means that $T_{1,3}(G^{T_{i,i+2}(G)})=O(n^2+t_{2k}(G))$.
\end{lemma}
\begin{proof}
Consider a $2(k-1)$-cycle in $G^{T_{1,3}(G)}$ containing two distinct auxiliary edges, $(a_1^1, a_3^1)$ and $(a_1^2, a_3^2)$. By our edge inclusion criteria, the table entry $T_{1,3}(G)[a_1^1, a_3^1]$ contains at least two distinct intermediate vertices from $A_2$, and $T_{1,3}(G)[a_1^2, a_3^2]$ similarly contains at least two distinct vertices. We can therefore choose a vertex $a_2 \in T_{1,3}(G)[a_1^1, a_3^1]$ and a distinct vertex $a_2' \in T_{1,3}(G)[a_1^2, a_3^2]$ such that $a_2 \neq a_2'$. 

Replacing the auxiliary edge $(a_1^1, a_3^1)$ with the $2$-path $a_1^1 - a_2 - a_3^1$ and the auxiliary edge $(a_1^2, a_3^2)$ with the $2$-path $a_1^2 - a_2' - a_3^2$ transforms the $2(k-1)$-cycle in $G^{T_{1,3}(G)}$ into a $2k$-cycle in the original graph $G$. Thus, the number of $2(k-1)$-cycles in the auxiliary graph which have $2$ auxiliary cycles is bounded by the total number of $2k$-cycles in $G$, ensuring that the data structure initialization time is bounded by $O(n^2 + t_{2k}(G))$.
\end{proof}

Finally, to enumerate type 1 paths for a given pair $(a_1, a_4) \in A_1 \times A_4$, the algorithm retrieves the intermediate vertices using the data structure. This step uses the tables $T_{1,3}(G^{T_{i,i+2}(G)})$. For $i=1$, we identify each vertex $a_3$ on a path from $a_1$ to $a_4$ in $G^{T_{1,3}(G)}$, and for each such $a_3$, we extract all intermediate vertices $a_2 \in T_{1,3}(G)[a_1, a_3]$ to write the full paths $a_1 - a_2 - a_3 - a_4$. Executing this lookup procedure across both values of $i$ guarantees the complete listing of all type 1 paths.

\subsection{Part 2: Handling Type 2 Paths}

We now want to also take care of $3$-paths that cannot be changed by replacing a single vertex. Notice that these are sparse $3$-paths, because for a path $a_1-a_2-a_3-a_4$, there are only 2 internal vertices, and the only way for the path not to be sparse is if the exists another path $a_1-a_2'-a_3'-a_4$ such that either $a_2=a_2'$ or $a_3=a_3'$ (they cannot both be equal). So we want to construct the table $T_{1,4}(G)$ for those paths. The algorithm to construct this table is written in pseudocode in Algorithm~\ref{alg:compute_tablesM}.

\begin{algorithm}[htbp]
  \caption{Compute the table $T_{1,4}(G)$}
  \label{alg:compute_tablesM}
  \DontPrintSemicolon
  
  \SetKwProg{Fn}{Function}{:}{}
  $L_{1,2,3} = \emptyset$\;
  \For(\Comment*{Compute helper array $L_{1,2,3}$}){$(a_1, a_3) \in A_1 \times A_3$} {
    \If{$|T_{1,3}(G)[a_1, a_3]| == 1$} {
        $a_2 = \text{the unique vertex in } T_{1,3}(G)[a_1, a_3]$\;
        Add tuple $(a_1, a_2, a_3)$ to $L_{1,2,3}$\;
    }
  }
  Sort $L_{1,2,3}$ lexicographically by the keys $(a_2, a_3)$ using Radix Sort\;
  
  \medskip
  $L_{2,3,4} = \emptyset$\;
  \For(\Comment*{Compute helper array $L_{2,3,4}$}){$(a_2, a_4) \in A_2 \times A_4$} {
    \If{$|T_{2,4}(G)[a_2, a_4]| == 1$} {
        $a_3 = \text{the unique vertex in } T_{2,4}(G)[a_2, a_4]$\;
        Add tuple $(a_2, a_3, a_4)$ to $L_{2,3,4}$\;
    }
  }
  Sort $L_{2,3,4}$ lexicographically by the keys $(a_2, a_3)$ using Radix Sort\;
  
  \medskip
  $i = 0, j = 0$\;
  \While{$i < |L_{1,2,3}|$ \textbf{and} $j < |L_{2,3,4}|$} {
      \If{$L_{1,2,3}[i].(a_2, a_3) > L_{2,3,4}[j].(a_2, a_3)$} {
         $j = j + 1$\;
      }
      \ElseIf{$L_{1,2,3}[i].(a_2, a_3) < L_{2,3,4}[j].(a_2, a_3)$} {
         $i = i + 1$ \;
      }
      \Else {
         $S_1 = \emptyset, S_4 = \emptyset$\;
         $i' = i, j' = j$\;
         \While{$i < |L_{1,2,3}|$ \textbf{and} $L_{1,2,3}[i].(a_2, a_3) == L_{1,2,3}[i'].(a_2, a_3)$} {
             Insert $L_{1,2,3}[i].a_1$ into $S_1$\;
             $i = i + 1$\;
         }
         \While{$j < |L_{2,3,4}|$ \textbf{and} $L_{2,3,4}[j].(a_2, a_3) == L_{2,3,4}[j'].(a_2, a_3)$} {
             Insert $L_{2,3,4}[j].a_4$ into $S_4$\;
             $j = j + 1$\;
         }
         \For{$(a_1', a_4') \in S_1 \times S_4$} {
             Insert $(a_2, a_3)$ into $T_{1,4}(G)[a_1', a_4']$\;
         }
      }
  }
\end{algorithm}

The algorithm constructs a list $L_{1,2,3}$ containing all $2$-paths $a_1 - a_2 - a_3 \in A_1 \times A_2 \times A_3$ such that $|T_{1,3}(G)[a_1, a_3]| = 1$. In other words, $L_{1,2,3}$ stores precisely those $2$-paths for which $a_2$ is the unique common neighbor between $a_1$ and $a_3$. Symmetrically, we construct a second list $L_{2,3,4}$ consisting of all $2$-paths $a_2 - a_3 - a_4 \in A_2 \times A_3 \times A_4$ satisfying $|T_{2,4}(G)[a_2, a_4]| = 1$. 

Since each pair of endpoints $(a_i, a_{i+2})$ contributes at most one $2$-path to its respective list, the size of both lists is bounded by $O(n^2)$. We then sort both lists lexicographically by the vertex pair $(a_2, a_3)$, treating $a_2$ as the primary key and $a_3$ as the secondary key. This sorting phase is executed efficiently in $O(n^2)$ time using radix sort.

After sorting, we perform a merge-like algorithm over $L_{1,2,3}$ and $L_{2,3,4}$ using two pointers initialized at the beginning of each list. At each step, we compare the current $(a_2, a_3)$ pairs of both lists lexicographically. If the pair in $L_{1,2,3}$ is smaller than the pair in $L_{2,3,4}$, we advance the pointer of $L_{1,2,3}$; if it is larger, we advance the pointer of $L_{2,3,4}$. When the two pairs match on a specific vertex pair $(a_2, a_3)$, we save the blocks of entries sharing this pair in both lists, by moving the pointer forward in each list and saving the paths until we reach a path with a different pair from $A_2\times A_3$. We then define the sets $S_1 = \{a_1 \mid (a_1-a_2-a_3) \in L_{1,2,3}\}$ and $S_4 = \{a_4 \mid (a_2-a_3-a_4) \in L_{2,3,4}\}$. For every pair $(a_1, a_4) \in S_1 \times S_4$, we insert the combined $3$-path $a_1 - a_2 - a_3 - a_4$ into the table entry $T_{1,4}(G)[a_1, a_4]$. After processing the matching blocks, we advance both pointers past these entries and resume the merge.

\begin{lemma}
Every colorful sparse $3$-path in $G$ is inserted into $T_{1,4}(G)$.
\end{lemma}

\begin{proof}
Let $P = a_1 - a_2 - a_3 - a_4$ be an arbitrary colorful sparse $3$-path in $G$, where $a_i \in A_i$ for each $i \in \{1,2,3,4\}$. By the definition of a sparse $3$-path, the prefix sub-path $a_1 - a_2 - a_3$ must be the unique $2$-path connecting $a_1$ and $a_3$ in $G$. Thus, $|T_{1,3}(G)[a_1, a_3]| = 1$, which guarantees that this $2$-path is included in $L_{1,2,3}$. By symmetric reasoning, the suffix sub-path $a_2 - a_3 - a_4$ is the unique $2$-path between $a_2$ and $a_4$, meaning it is present in $L_{2,3,4}$. 

During the merge phase, the two pointers must simultaneously arrive at the block corresponding to the shared vertex pair $(a_2, a_3)$. Consequently, $a_1$ will be added to $S_1$ and $a_4$ will be added to $S_4$. The algorithm then takes care of all pairs in $S_1 \times S_4$, guaranteeing that the complete path $P$ is inserted into $T_{1,4}(G)[a_1, a_4]$.
\end{proof}

Regarding the time complexity, generating $L_{1,2,3}$ and $L_{2,3,4}$ requires $O(n^2)$ time by going through the tables $T_{1,3}(G)$ and $T_{2,4}(G)$. Sorting these lists via radix sort takes time linear in their size, which is $O(n^2)$. During the merge phase, the total time spent advancing the pointers is bounded by the total number of entries, $O(n^2)$. Finally, inserting the entries of $T_{1,4}(G)$ takes time proportional to the total size of the final table, $O(|T_{1,4}(G)|)$. Therefore, the entire procedure runs in $O(n^2 + |T_{1,4}(G)|)$ time.

\begin{theorem} \label{thm:3-sparse_paths}
Fix any $2k \ge 6$, the size of $T_{1,4}(G)$ is $O(n^2 + t_{2k}(G))$. 
\end{theorem}

\begin{proof}
If $|T_{1,4}(G)| \le 100kn^2$, the claim trivially holds. Otherwise, there must exist a vertex $a_1 \in A_1$ which is a part of more than $100kn$ colorful sparse $3$-paths ($a_1 - a_2 - a_3 - a_4$). By definition, $|T_{1,4}(G)|$ represents the number of colorful sparse 4 paths. We construct subgraph $G_{a_1}$ containing $a_1$ along with every edge involved in colorful sparse $3$-paths starting in $a_1$. Formally, for each sparse path $a_1 - a_2' - a_3' - a_4'$, the graph $G_{a_1}$ contains the vertices $a_2',a_3',a_4'$, and the edges $(a_1,a_2'),(a_2',a_3'),(a_3',a_4').$ We apply a cycle-finding approach similar to the one described in Theorem~\ref{size_of_2-path_table}. However, we first establish several graph properties necessary to guarantee the algorithm's correctness.

\begin{lemma}
\label{lemma_1}
The induced subgraph $G_{a_1}\setminus A_4$ is a tree.
\end{lemma}
\begin{proof}
Suppose this assertion is false. A simple cycle $C$ in $G_{a_1}\setminus A_4$ must include a vertex from $a_3$, otherwise the cycle consists only of edges between $a_1$ and $A_2$. In this case, taking a vertex in $C$ other than $a_1$, we get a vertex in $A_2$, which can only have one edge in the cycle (to $a_1$), in contradiction. This implies the existence of a vertex $a_3 \in A_3$ connected to a pair of distinct vertices $a'_2, a''_2 \in A_2$, yielding a 4-cycle $a_1 - a'_2 - a_3 - a''_2$. Because $a_3$ is a vertex in $G_{a_1}$, it must belong to a sparse path connecting $a_1$ to some vertex $a_4 \in A_4$; let this path be $a_1 - a_2 - a_3 - a_4$. However, the vertex $a_2$ can be substituted with either $a'_2$ or $a''_2$ depending on its identity, to create either $a_1 - a_2' - a_3 - a_4$ or $a_1 - a_2'' - a_3 - a_4$, which breaks the sparsity requirement of the path, in contradiction.
\end{proof}
For each vertex $a_2 \in A_2 \cap V({G_{a_1}})$, we define  $G_{a_1-a_2}$ as a graph containing all edges participating in $2$-paths from $a_2$ to $A_4$ within $G_{a_1}$, or $G[\{a_2\} \cup (A_3 \cap N(a_2)) \cup (A_4 \cap N(N(a_2)))]$. In other words, for each path $a_1-a_2-a_3'-a_4'$, the vertices $a_3',a_4'$ are in $G_{a_1-a_2}$, and the edges $(a_2,a_3'),(a_3',a_4')$ are also in $G_{a_1-a_2}$.
\begin{lemma}
\label{lemma_2}
The graph $G_{a_1-a_2}$ is a tree.
\end{lemma}

\begin{proof}
Suppose this statement is false. This implies the existence of a cycle in $G_{a_1-a_2}$. There is a vertex $a_4\in A_4$ in the cycle, otherwise the vertices are only in $\{a_2\}\cup (A_3\cap N(A_2)$, leading to a contradiction as in Lemma~\ref{lemma_1}. So there exists a vertex $a_4 \in A_4 \cap V({G_{a_1-a_2}})$ and a pair of distinct vertices $a'_3, a''_3 \in A_3 \cap V({G_{a_1-a_2}})$ creating a cycle $a_2 - a_3' - a_4 - a_3''$.Because $(a'_3, a_4) \in E({G_{a_1}})$, a unique sparse $3$-path must exist traversing $a_1, a_3'$, and $a_4$. This path is forced to be $a_1 - a_2 - a_3' - a_4$; any alternative choice $a_2^* \in A_2$ would imply the existence of an 4-cycle $a_1 - a_2 - a_3' - a_2^*$ within $G_{a_1}$, violating Lemma~\ref{lemma_1}, because in this case $G_{a_1}\setminus A_4$ is not a tree. Thus, $a_1 - a_2 - a_3' - a_4$ is a sparse path, yet $a_1 - a_2 - a''_3 - a_4$ is another $3$-path in $G$ using the same vertices other than switching $a_3$ with $a_3'$, which means $a_1 - a_2 - a_3' - a_4$ is not sparse, in contradiction.
\end{proof}
By Lemma~\ref{lemma_1}, every sparse $3$-path in $G_{a_1}$ has a distinct last edge, which means the number of sparse $3$-paths from $a_1$ to $A_4$ is $E(G_{a_1})$. Consequently, $G_{a_1}$ must contain over $100kn$ edges.
We now prove that there are sufficiently many $2k$-cycles in the graph $G_{a_1}$ that use the vertex $a_1$, by describing an algorithm that finds such cycles iteratively. The process, similar to Theorem~\ref{size_of_2-path_table} will include iteratively conducting a pruning process to remove vertices with low degree, and afterwards finding a cycle. The process stops when the graph has less than $100kn$ edges, and we will show that $t_{2k}(G_{a_1})=\Omega(E(G_{a_1})-110kn)$.
\paragraph{Pruning Process:}
 Applying similar vertex-pruning technique used in Theorem~\ref{size_of_2-path_table}, we remove all vertices in $G[A_3\cup A_4]$ whose degree in $G[A_3\cup A_4]$ falls below $10k$ (if the edge count during this process falls below $100kn$ we stop the algorithm). Let us denote this pruned graph as $G_{a_1}'$.
\paragraph{Finding a Cycle:}
Next, we map every remaining vertex $a_3 \in A_3 \cap G_{a_1}'$ to its unique $A_2$-neighbor $\pi(a_3) \in A_2 \cap G_{a_1}'$, denoting this pair as $(a_3, \pi(a_3))$. We now construct a $2k$-cycle. Starting at the root $a_1$, we follow an arbitrary initial path segment $P=a_1 - a_2 - a_3-a_4$, where $a_2 = \pi(a_3)$. We next greedily add a path segment of length $2k-6$ to $P$ using edges in $E[A_3,A_4]$. The degree of each vertex in $G[A_3\cup A_4]$ is at least $10k$, meaning we can greedily construct a $(2k-6)$-path without intersecting $a_3$. After reaching the final vertex of the path segment, $a_4'\in A_4$, we want to add some path segment going back to $a_1$ to $P$ and close a $2k$-cycle. Looking at the neighbors of $a_4'$, at most $k-2$ of them are currently in $P$. Additionally, at most one neighbor $a_3'\in A_3$ of $a_4'$ satisfies $\pi (a_3')=a_2$, because if two different neighbors of $a_4'$  satisfy this we get a cycle in $G_{a_1-a_2}$ in contradiction to Lemma~\ref{lemma_2}. This means we must avoid a set of at most $k-1$ vertices in $A_3$, and since $\deg(a_4)\ge 10k$, we can pick some neighbor of $a_4'$, $a_3''\in A_3$, outside of that set. So we add the final path segment $a_4'-a_3''-\pi(a_3'')-a_1$ to $P$, making $P$ a simple cycle.\\
Following the finding of each $2k$-cycle, we remove the cycle's edges in $E[A_3,A_4]$, prune vertices with a degree below $10k$, and repeat the process, stopping when the number of edges in the graph is below $100kn$. The total number of pruned edges across the sequence is bounded by $10kn$, leaving $|E_{G_{a_1}}| \le 110kn + 2k\cdot t_{2k}(G_{a_1})$. Summing this for all $a_1\in A_1$ we get $|T_{1,4}(G)| \le \sum_{a_1\in A_1}E(G_{a_1})\le n\cdot 110kn+\sum_{a_1\in A_1} 2k\cdot t_{2k}(G_{a_1}))= O(n^2 + t_{2k}(G))$, as required.
\begin{theorem} \label{important_remark2}
Fix any constant $k$, the size of $T_{1,4}(G)$ is $O(n^2 + t'_{2k}(G))$.
\end{theorem}
\begin{proof}
    The proof is the same as Algorithm~\ref{alg:compute_tablesM}.
\end{proof}
\begin{remark}
Similar before, we can decide if the cycles have $2$ edges in both $E[A_1,A_2]$ and $E[A_2,A_3]$, and the rest in $E[A_3,A_4]$, or cycles  have $2$ edges in both $E[A_3,A_4]$ and $E[A_2,A_3]$, and the rest in $E[A_1,A_2]$. This also proves that the size of $T_{1,4}(G)$ is $O(n^2 + t''_{2k}(G))$.
\end{remark}
\begin{figure}[htbp]
  \centering
  \includegraphics[width=0.4\linewidth]{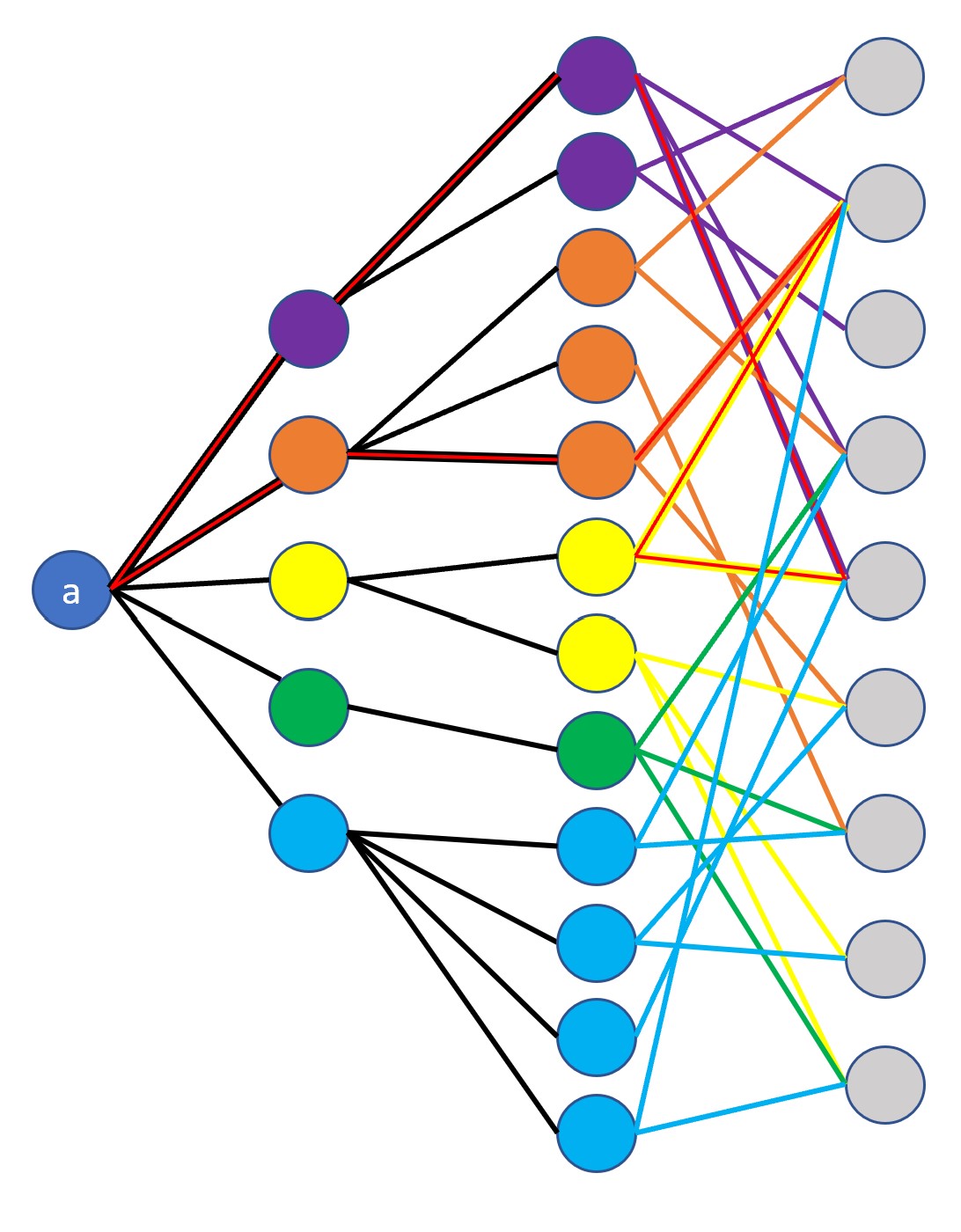}
  \caption{Proof for size of $T_{1,4}(G)$ table}
  \label{pic:RALG}
\end{figure}
\end{proof}

\subsubsection{Final Data Structure for Four layers}
With the initialization of $T_{1,4}(G)$ complete, we can extract all sparse paths for any vertex pair $(a_1, a_4) \in A_1 \times A_4$ directly from $T_{1,4}(G)[a_1, a_4]$. Moreover, listing paths via the tables $T_{1,3}({G^{T_{i,i+2}(G)}})$ for both $i=1$ or $2$ covers all non-sparse paths connecting $a_1$ and $a_4$, as seen earlier. To evaluate $\texttt{exist-path}(a_1,a_4)$ in $O(1)$ time, the algorithm simply verifies if $T_{1,3}({G^{T_{i,i+2}}})[a_1, a_4]$ for both $i=1$ or $2$ and $T_{1,4}(G)[a_1, a_4]$ are simultaneously empty. If they are, it returns false; otherwise, a path is guaranteed to exist. For $\texttt{all-paths}(a_1,a_4)$, list all type 1 paths using the algorithm for type 1 path listing, and list all type 2 paths by listing all paths in $T_{1,4}(G)$. This completes the construction of the four-layer $\mathcal{D}(G)$ structure.

\section{Data Structure for Five layers}\label{sec:five_layers}

\subsection{Auxiliary Graphs}
Before we describe the data structure for 5 layers, we first define some auxiliary graphs.
The auxiliary graph $G^{T^G_{i,i+r}}$ is defined as follows: 
\begin{itemize}
    \item All vertex layers $A_j$ such that $j\le i$ or $j\ge i+r$ remain the same. Additionally, all edge sets $E[A_j,A_{j+1}]$ for $j<i$ or $j\ge i+r$ remain the same.
    \item The vertex layers $A_j$ for $i<j<i+r$ are deleted.
    \item For each pair $(a_i,a_{i+r})$, the edge $(a_i,a_{i+r})$ is in the auxiliary graph if and only if the size of $T^G_{i,i+r}[a_i,a_{i+r}]$ is at least 2, meaning there is at least two sparse paths between $(a_i,a_{i+r})$.
\end{itemize}
 For each auxiliary graph $G'$, $\mathcal{D}(G)$ will contain $\mathcal{D}(G')$. Each auxiliary graph has less layers than the original graph, so we may construct~$\mathcal{D}(G')$ using the previously defined data structures.

 We define a \textbf{table reduction} as a single instance of contracting intermediate layers when constructing an auxiliary graph based on a given table $T$.

\subsection{Algorithm Description}
Given a 5-layer graph $G=(V,E)$, we now describe the construction of $\mathcal{D}(G)$. We split all the colorful $4$-paths in $G$ into 4 different types:\\
\begin{enumerate}
    \item[\textbf{1}] The path $a_1 - a_2 - a_3 - a_4 - a_5$ contains an internal vertex $a_i$ ($i \in \{2, 3, 4\}$) that can be replaced by an alternative vertex $a_i'$ to yield another path. 
    \item[\textbf{2}] The first condition is false (implying the length $3$ sub-paths in the path are sparse), and the path contains internal vertices $a_i, a_{i+1}$ ($i \in \{2, 3\}$) that can be replaced by $a_i', a_{i+1}'$ to form another path, such that $a_{i-1} - a_i' - a_{i+1}' - a_{i+2}$ is also a sparse path. 
    \item[\textbf{3}] The first two conditions are false, and the path contains internal vertices $a_i, a_{i+1}$ ($i \in \{2, 3\}$) that can be replaced by $a_i', a_{i+1}'$ to form another path, but the sub-path $a_{i-1} - a_i' - a_{i+1}' - a_{i+2}$ is \emph{not} a sparse path.
    \item[\textbf{4}] The path is a sparse $4$-path.

    \begin{figure}
        \centering
        \includegraphics[width=0.8\linewidth]{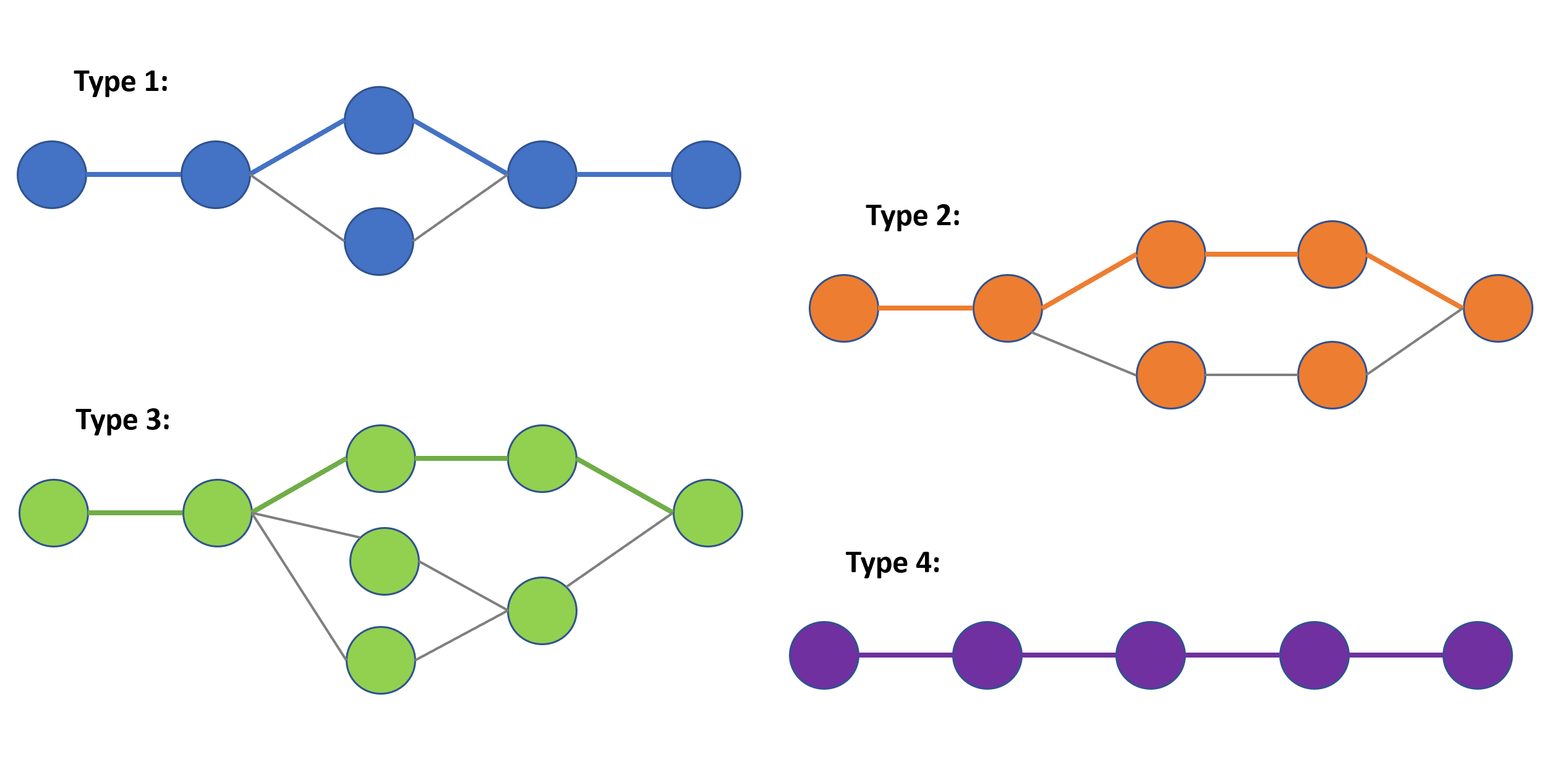}
        \caption{Illustration of path types}
        \label{fig:types_of_paths}
    \end{figure}
\end{enumerate}
\begin{theorem}
Any colorful $4$-path $P=a_1-a_2-a_3-a_4-a_5\in A_1\times A_2\times A_3\times A_4\times A_5$ belongs to one of these types.
\end{theorem}
\begin{proof}
If for any pair $(a_i,a_{i+2})$ for $i\in \{1,2,3\}$ there is another vertex $a_{i+1}'\in A_{i+1}$ such that $a_i-a_{i+1}'-a_{i+2}$ is a path in $G$, then $P$ belongs to type 1. \\
Otherwise, $a_i-a_{i+1}-a_{i+2}-a_{i+3}$ are sparse $3$-paths for $i\in \{1,2\}$. \\
If for either $i$ value there are two other vertices $a_{i+1}',a_{i+2}'$ such that $a_i-a_{i+1}'-a_{i+2}'-a_{i+3}$ is a sparse $3$-path in $G$, then $P$ is a type 2 path.\\
Otherwise, If for either $i$ value there are two other vertices $a_{i+1}',a_{i+2}'$ such that $a_i-a_{i+1}'-a_{i+2}'-a_{i+3}$ is a non-sparse $3$-path in $G$, then $P$ is a type 3 path.\\
Lastly, if that there are no other 3 paths from $a_{i}$ to $a_{i+3}$, then $P$ is a sparse $4$-path, or a type 4 path.\\
So $P$ is guaranteed to be in at least one type.
\end{proof}
We split the construction of $\mathcal{D}$ into 4 parts for each of these types.
For each of these types we will describe a construction for its corresponding part in the data structure, as well as correctness and complexity.

\subsection{Part 1: Type 1 Paths}
 We first construct a data structure to answer queries with respect to paths of type~$1$ in G.
When asked a query of $\texttt{exist-path}(a_1,a_5)$ on the graph $G$, we want to be able to check if there is a path between $a_1$ and $a_5$ that satisfies that the path $a_1 - a_2 - a_3 - a_4 - a_5$ contains an internal vertex $a_i$ ($i \in \{2, 3, 4\}$) that can be replaced by an alternative vertex $a_i'$ to yield another path. Additionally, when asked a query of $\texttt{all-paths}$, we want to list all such paths.

We begin by constructing all two-path tables $T_{i,i+2}(G)$ for $i \in \{1, 2, 3\}$. Once these tables are constructed, we build their corresponding auxiliary graphs $G^{T_{i,i+2}(G)}$. From this stage, we build the same data structure developed for the case of four layers, $\mathcal{D}(G^{T_{i,i+2}(G)})$. The complexity analysis changes and we will prove that the complexity of creating $\mathcal{D}(G^{T_{i,i+2}(G)})$ is $O(n^2+t_8(G))$. However, the construction of the data structure remains the same, and this completes part 1 of the construction.

\begin{algorithm}[htbp]
  \caption{Compute part 1 of the data structure}
  \label{alg:compute_part1}
  \DontPrintSemicolon
    \For{$i \in \{1,2,3\}$}{
      Construct table $T_{i,i+2}(G)$ \;
      Construct graph $G^{T_{i,i+2}(G)}$ \;
      Construct data structure $\mathcal{D}(G^{T_{i,i+2}(G)})$ according to Section 4 \;
    }
\end{algorithm}

\subsubsection{Queries on $\texttt{exist-path}$ and $\texttt{all-paths}$}
As established in the four-layer setting, the data structure $\mathcal{D}(G^{T_{i,i+2}(G)})$ supports both $\texttt{exist-path}$ and $\texttt{all-paths}$ queries. When invoking $\texttt{all-paths}$ on $G^{T_{i,i+2}(G)}$, every $6$-path listed by the routine corresponds to multiple candidate vertices from $A_{i+1}$. 
We now want to apply this to the $\texttt{exist-path}$ and $\texttt{all-paths}$ queries in $G$, for type 1 paths.
$\texttt{exist-path}(a_1,a_5)$: for each $i\in \{1,2,3\}$ we conduct $\texttt{exist-path}(a_1,a_5)$ for the graph $G^{T_{i,i+2}(G)}$. Return false of all of these queries return false, and true otherwise.
$\texttt{all-path}(a_1,a_5)$: for each $i\in \{1,2,3\}$ we conduct $\texttt{all-path}(a_1,a_5)$ for the graph $G^{T_{i,i+2}(G)}$. For each path listed, if $i=1$ then the paths are of the type $a_1-a_3-a_4-a_5$, and we go to $T_{1,3}(G)[a_1,a_3]$ for all possible values of $a_2'$ where $a_1-a_2'-a_3-a_4-a_5$ is a path in $G$, and list those paths. For $i=2,3$ act similarly with the correct indices. Keep the paths in a hash table to avoid listing a path twice.
\subsubsection{Complexity Analysis of Table Construction}
For each auxiliary graph $G'=G^{T_{i,i+2}}$ we build $\mathcal{D}(G')$. During the construction of $\mathcal{D}(G')$ we construct:
\begin{itemize}
    \item $T_{j,j+2}(G')$ for $j\in\{1,2\}$
    \item $T_{1,4}(G')$
    \item $\mathcal{D}(G^{T_{j,j+2}(G')})$ for $j\in\{1,2\}$. $G^{T_{j,j+2}(G')}$ is a $3$-layer graph, so $\mathcal{D}(G^{T_{j,j+2}(G')})$ includes only $T_{1,3}(G^{T_{j,j+2}(G')})$.
\end{itemize}
Let $k \ge 4$ be a fixed, even integer constant. We now prove the time complexity of constructing these tables is bounded by $O(n^2 + t_{8}(G))$ time. Most of the process is bounded by $O(n^2 + t_{2k}(G))$ for any fixed $k\ge 4$, but there is an exception we will discuss (the exception is still constructed in $O(n^2 + t_{8}(G))$ as needed).\\
 We begin the proof with a helpful lemma which shows a mapping between cycles in auxiliary graphs to cycles in the original graph. We then prove the bounds on the size of tables for auxiliary graphs $G'=G^{T_{i,i+2}(G)}$ which are $T_{j,j+2}(G')$ for $j\in\{1,2\}$, and $T_{1,4}(G')$. We then go on to prove the bounds on the size of tables for auxiliary graphs of the graph $G'$, $T_{1,3}(G^{T_{j,j+2}(G')})$ for $j=1$ or $2$. 
\noindent Let $G = (V, E)$ be a layered graph, and let $G'=G^{T_{i,i+2}(G)}$ be its auxiliary graph for some $i\in [1,3]$.
\begin{lemma}\label{lem:single_table_expansion}
Every $2(k-1)$-cycle in $G^{T_{i,i+2}(G)}$ containing exactly two auxiliary edges from $E(G^{T_{i,i+2}(G)})[A_i,A_{i+2}]$ maps injectively to a simple $2k$-cycle in $G$. The $2k$-cycle in $G$ has 2 edges in the layer transition $E[A_i,A_{i+1}]$, and 2 edges in the layer transition $E[A_{i+1},A_{i+2}]$.
\end{lemma}

\begin{proof}
Let $C$ be a $2(k-1)$-cycle in $G^{T_{i,i+2}(G)}$ containing exactly two auxiliary edges, $e_1 = (a'_i, a'_{i+2})$ and $e_2 = (a''_i, a''_{i+2})$. By definition, their corresponding table entries satisfy $|T_{i,i+2}(G)[a'_i, a'_{i+2}]| > 1$ and $|T_{i,i+2}(G)[a''_i, a''_{i+2}]| > 1$.

We select an intermediate vertex $a'_{i+1} \in T_{i,i+2}(G)[a'_i, a'_{i+2}]$ and a distinct intermediate vertex $a''_{i+1} \in T_{i,i+2}(G)[a''_i, a''_{i+2}]$ such that $a'_{i+1} \neq a''_{i+1}$. We construct a new cycle $C' \subseteq G$ by replacing $e_1$ with the $2$-path $a'_i - a'_{i+1} - a'_{i+2}$ and replacing $e_2$ with the $2$-path $a''_i - a''_{i+1} - a''_{i+2}$.

This substitution increases the total cycle length by exactly two edges. Because layer $A_{i+1}$ is omitted from the vertex set of $G^{T_{i,i+2}(G)}$, the vertices $a'_{i+1}$ and $a''_{i+1}$ cannot intersect any other vertex already present in $C$. Thus, $C'$ is simple $2k$-cycle in $G$, which has 2 edges in the layer transition $E[A_i,A_{i+1}]$, and 2 edges in the layer transition $E[A_{i+1},A_{i+2}]$.
\end{proof}

\paragraph{Bounding the Size of $T_{j,j+2}(G')$:}
\noindent Let $G'$ be a $4$-layer auxiliary graph such that $G' = G^{T_{i,i+2}(G)}$ for some $i$, partitioned into layers $B_1 - B_2 - B_3 - B_4$. One of the layer transitions consists of auxiliary edges (namely $E[A_i,A_{i+2}]$), and the others consist of edges from the original graph $G$. For example, If $i=3$, $B_1 - B_2 - B_3 - B_4=A_1-A_2-A_3-A_5$, and the layer transition $E[A_3,A_5]$ consists of auxiliary edges.
\begin{lemma}\label{lem:boundary_layer_symmetry}
The total number of paths stored in the table $T_{j,j+2}(G')$ over the induced subgraph $G'[B_{j} \cup B_{j+1} \cup B_{j+2}]$ for $j\in \{1,2\}$ is bounded by $O(n^2 + t_{2k}(G))$.
\end{lemma}

\begin{proof}
By Theorem~\ref{important_remark}, the number of paths stored in $T_{1,3}(G')$ satisfies $|T_{1,3}(G')| \le O(n^2 + t'_{2r}(G'[B_{j} \cup B_{j+1} \cup B_{j+2}]))$, where $2r$ denotes the length of the cycles, and $t'_{2r}(G'[B_{j} \cup B_{j+1} \cup B_{j+2}])$ counts only the cycles in which exactly 2 edges are in the first layer transition, while the remaining $2r-2$ edges are in the last transition. We also know $|T_{1,3}(G')| \le O(n^2 + t''_{2r}(G'[B_{j} \cup B_{j+1} \cup B_{j+2}]))$.\\

We bound this cycle count by analyzing the edge sets in $G'[B_1 \cup B_2 \cup B_3]$ relative to the original graph $G$:

\begin{itemize}
    \item \textbf{Case 1 (No Auxiliary Edges):} If both $E(G'[B_j \cup B_{j+1}]) \subseteq E(G)$ and $E(G'[B_{j+1} \cup B_{j+2}]) \subseteq E(G)$, then every edge in the induced subgraph is an edge in $G$. By setting $r = k$, the table size is bounded yielding $|T_{j,j+2}(G')| \le O(n^2 + t'_{2k}(G'))$. Because $G'[B_1 \cup B_2 \cup B_3]$ is a subgraph of $G$ and thus contains only original edges, any such $2k$-cycle counted by $t'_{2k}(G'[B_1 \cup B_2 \cup B_3])$ is a $2k$-cycle in the original graph $G$, establishing that $t'_{2k}(G'[B_{j} \cup B_{j+1} \cup B_{j+2}]) \le t_{2k}(G)$.
    \item \textbf{Case 2 (Auxiliary Edges Present):} Suppose one of the layer transitions consists of auxiliary edges representing $2$-paths, while the other transition consists of original edges from $E(G)$ (as shown before only one transition in $G'$ consists of auxiliary edges). We assume the first layer transition consists of auxiliary edges. Otherwise we can conduct the same proof using $t''$ instead of $t'$. By Theorem~\ref{important_remark}, the table size satisfies $|T_{1,3}(G')| \le O(n^2 + t'_{2r}(G'[B_{j} \cup B_{j+1} \cup B_{j+2}]))$. Choose $r = k-1$ to get $|T_{1,3}(G')| \le O(n^2 + t'_{2(k-1)}(G'[B_{j} \cup B_{j+1} \cup B_{j+2}]))$. We apply Lemma~\ref{lem:single_table_expansion} to the cycle, which replaces the two auxiliary edges within the cycle with corresponding $2$-paths in $G$ while leaving the remaining $2(k-2)$ original edges unmodified. This increases the total edge count by exactly two, mapping the $2(k-1)$-cycle injectively into a simple $2k$-cycle in $G$, establishing the bound $t'_{2(k-1)}(G') \le t_{2k}(G)$.
\end{itemize}
In both cases, the table size $|T_{1,3}(G')|$ is at most $O(t_{2k}(G)+n^2)$, completing the proof.
\end{proof}

\paragraph{Bounding the Size of $T_{1,4}(G')$:}
\begin{lemma}
The number of paths in the table $T_{1,4}(G')$ is bounded by $O(n^2+t_{2k}(G))$
\end{lemma}
\begin{proof}
$G'=G^{T_{i,i+2}(G)}$. If $i\ne 3$:\\
As shown by Theorem~\ref{thm:3-sparse_paths}, $T_{1,4}(G')=O(n^2+t'_{2k-2}(G'))$. We know that the boundary layer transition $E[B_3,B_4]$ consists of edges in the graph $G$ (because $i\ne 3$). This means that there are only 2 auxiliary edges in one of the other layer transitions, which means each of these cycles can be mapped into a $2k$-cycle in $G$ as in Lemma~\ref{lem:single_table_expansion}.\\
If $i=3$, we conduct the same proof using $t''_{2k-2}(G')$ and the boundary layer transition $E[B_1,B_2]$.
\end{proof}

 \paragraph{Bounding the Size of $\mathcal{D}(G^{T_{j,j+2}(G')})$ for Some $j\in \{1,2\}$:}   
\noindent Let $G$ be a $5$-layer graph. 
 Let $G_2$ be a $3$-layer auxiliary graph derived via a sequence of two table reductions (independent of the contraction order or intermediate configurations), resulting in an auxiliary graph spanning layers $A_1 - A_{j_1} - A_5$.
 \subparagraph{Boundary layer condition:}
 We say that an auxiliary graph $G'$ of $G$ satisfies the \emph{boundary layer condition} if one of its layer transitions consists of original edges from $E(G)$. For $G_2$, this means either $E(G_2[A_1 \cup A_{j_1}]) \subseteq E(G)$ or $E(G_2[A_{j_1} \cup A_5]) \subseteq E(G)$, while the other layer transition consists of auxiliary edges that are a result of 2 table reductions, who map back to paths which intersect with intermediate layers $A_{j_2}, A_{j_3} \subset V$.
\subparagraph{}
 If $A_{j_1}\ne A_3$, the boundary layer condition is correct. We first assume this is the case, and later deal with the exception ($A_{j_1}=A_3)$. $G_2$ is a $3$-layer auxiliary graph obtained by contracting two intermediate layers via table reductions. Its layers are either $A_1 - A_2 - A_5$, $A_1 - A_4 - A_5$ or $A_1 - A_3 - A_5$. First we assume the boundary layer condition is correct, so $V(G_2)=A_1 - A_2 - A_5$, or $A_1 - A_4 - A_5$. We assume for this proof that the layers are $A_1 - A_2 - A_5$, the other case is symmetrical. We will address the case $V(G_2)=A_1 - A_3 - A_5$ later.
\subparagraph{}
    Let an auxiliary edge representing $2$-paths composed of original graph edges from $E(G)$ be denoted as a \textbf{Level-1 auxiliary edge}, and let an auxiliary edge representing $2$-paths composed one one Level-1 auxiliary edge and one edge from the original graph $G$ be denoted as a \textbf{Level-2 auxiliary edge}.
\begin{theorem}\label{lem:layered_table_composition}
Any $(2k-4)$-cycle in $G_2$ containing exactly two Level-2 auxiliary edges maps injectively to a simple $2k$-cycle in the original graph $G$, implying that $t''_{2k-4}(G_2) \le t_{2k}(G)$. Consequently, the table sizes of $T_{1,3}(G_2)$ which are bounded by $O(n^2 + t''_{2k-4}(G_2))$ as shown in Remark~\ref{other_side} are bounded by $O(n^2 + t_{2k}(G))$.
\end{theorem}

\begin{proof}

In the graph $G_2$, the Level-2 auxiliary edges in $E(G_2[A_2 \cup A_5])$ are created by contracting the intermediate layers $A_3$ and $A_4$. Let $C$ be a $(2k-4)$-cycle in $G_2$ counted in $t''_{2k}(G_2)$, or in other words containing exactly two Level-2 auxiliary edges $e_1, e_2 \in E(G_2[A_2 \cup A_5])$ and exactly $2k-6$ original edges in $E(G_2[A_1 \cup A_2])$. We reconstruct the full cycle in $G$ via applications of Lemma~\ref{lem:single_table_expansion}:

\begin{enumerate}
    \item \textbf{First Path Mapping:} We apply Lemma~\ref{lem:single_table_expansion} to the cycle, replacing the Level-2 auxiliary edges $e_1$ and $e_2$ with their corresponding $2$-paths traversing an intermediate layer (either $A_3$ or $A_4$, depending on the construction of the auxiliary graph). The $2k-6$ original edges in $E(G_2[A_1 \cup A_2])$ remain unmodified. This substitution removes two Level-2 auxiliary edges and adds 2 edges from the original graph $G$, as well as 2 Level-1 auxiliary edges, resulting in a net addition of exactly 2 edges. This maps the $(2k-4)$-cycle into a $(2k-2)$-cycle containing exactly two Level-1 auxiliary edges situated within a layer transition which is not in $G_2$. The layer transition is either $A_3-A_5$, or $A_2-A_4$.
    \item \textbf{Second Path Mapping:} We apply Lemma~\ref{lem:single_table_expansion} a second time to the cycle, replacing the two Level-1 auxiliary edges with their corresponding $2$-paths traversing the remaining contracted intermediate layer. The $2k-6$ original edges remain unmodified. This substitution removes the two Level-1 auxiliary edges and adds four original graph edges from $E(G)$, resulting in a net addition of exactly two more edges.
\end{enumerate}

The entire mapping replaces two Level-2 auxiliary edges with exactly six original graph edges from $E(G)$, increasing the total cycle length by four to yield a $2k$-cycle. Because the intermediate vertex sets $A_3$ and $A_4$ are disjoint from $A_1$ and $A_2$ by the definition of the graph layers, the $2k-6$ unmodified edges cannot intersect the newly expanded paths. Furthermore, the internal vertices chosen from separate table entries are unique. Thus, the fully reconstructed cycle contains no vertex intersections outside its fundamental sequence, establishing that $(2k-4)$-cycle maps injectively to a $2k$-cycle in $G$.
\end{proof}

\paragraph{The Exception:} \label{exception}
The only auxiliary graph where the boundary layer condition fails occurs in the auxiliary graphs $G_2=G^{T_{1,3}{(G^{T_{3,5}(G)})}}$ and  $G_2=G^{T_{2,4}{(G^{T_{1,3}(G)})}}$ with layers $B_1 - B_2 - B_3$ (which corresponds to the layer sequence $A_1 - A_3 - A_5$). These two auxiliary graphs are equal, because both the layers and layer transitions are the same. During the construction of the table $T_{1,3}(G_2)$, the boundary layer condition is untrue. In this auxiliary graph, both layer transitions are made of auxiliary edges: the edges between $B_1$ and $B_2$ are determined by the table $T_{1,3}(G)$, and the edges between $B_2$ and $B_3$ are determined by the table $T_{3,5}(G)$. 
\begin{lemma}
$|T_{1,3}(G_2)| = O(n^2+t_8(G))$
\end{lemma}
\begin{proof}
Observe that if $|T_{1,3}(G_2)| > 100n^2$, we know that for every vertex pair $(b_1, b_3) \in B_1 \times B_3$, there exist at least $|T'_{1,3}[b_1, b_3]| - 1$ unique $4$-cycles that utilize exactly two edges from $E[B_1, B_2]$ and two edges from $E[B_2, B_3]$. This is because for two vertices $b_2',b_2''\in T_{1,3}(G_2)[b_1, b_3]$, $b_1-b_2'-b_3-b_2''$ is a cycle.

We next show that each 4-cycle in $G_2$ $a_1-a_3'-a_5-a_3''$ can be mapped to an 8-cycle in $G$. Going back to the original graph $G$, the two auxiliary edges between $B_1$ and $B_2$ (corresponding to $A_1$ and $A_3$), are denoted as $(a_1', a_3')$ and $(a_1', a_3'')$. $T_{1,3}(G)[a_1',a_3']\ge 2$, so we have two paths $a_1'-a_2^{1}-a_3'$, and $a_1'-a_2^{2}-a_3'$. Similarly, we have two paths for $(a_1',a_3'')$. Taking one of the paths between $(a_1',a_3'')$, $a_1'-a_2-a_3''$ either $a_2\ne a_2^1$, or $a_2\ne a_2^2$. Taking whichever one is different than $a_2$, we get paths between $(a_1',a_3')$ and $(a_1',a_3'')$ which don't intersect in $A_2$ or $A_3$. Applying the exact same path expansion to the two auxiliary edges between $A_3$ and $A_5$ yields a simple $8$-cycle in $G$, as desired.

While this specific sub-case does not explicitly bound the table size by $O(n^2 + t_{2k}(G))$ for an arbitrary $2k$, we prove $|T'_{1,3}| = O(n^2 + t_8(G))$. 
\end{proof}
We've shown that every single table created by the algorithm in part 1 is bounded by $O(n^2 + t_8(G))$, and so the construction phase for all tables in the algorithm is bounded by $O(n^2 + t_8(G))$.
\subsection{Part 2: Type 2 Paths}
When asked a query of $\texttt{exist-path}(a_1,a_5)$ on the graph $G$, we want to be able to check if there is a path between $a_1$ and $a_5$ that satisfies that the first condition is unmet (implying the length $3$ sub-paths are sparse), and the path contains internal vertices $a_i, a_{i+1}$ ($i \in \{2, 3\}$) that can be replaced by $a_i', a_{i+1}'$ to form another path, such that $a_{i-1} - a_i' - a_{i+1}' - a_{i+2}$ is also a sparse path.  Additionally, when asked a query of \texttt{all-paths}, we want to list all such paths.

For this purpose, we construct the sparse $3$-path tables $T_{i,i+3}(G)$ within the original graph. Following the proofs from the 4-layer construction, the total size of these tables is bounded by $O(n^2 + t_{2k}(G))$. We define the auxiliary graph $G^{T_{i,i+3}(G)}$ for $i \in \{1, 2\}$ analogously to $T_{i,i+2}$(G): the intermediate layers $A_{i+1}$ and $A_{i+2}$ are removed, and for each vertex pair $(a_i, a_{i+3}) \in A_i \times A_{i+3}$, an edge $(a_i, a_{i+3})$ exists in $E(G^{T_{i,i+3}})$ if and only if $|T_{i,i+3}(G)[a_i, a_{i+3}]| > 1$. 
\begin{algorithm}[htbp]
  \caption{Compute part 2 of the data structure}
  \label{alg:compute_part2}
  \DontPrintSemicolon
    \For{$i \in \{1,2\}$}{
      Construct table $T_{i,i+3}(G)$ \;
      Create graph $G^{T_{i,i+3}(G)}$ \;
      Create data structure $\mathcal{D}(G^{T_{i,i+3}(G)})$ according to Section 4 \;
    }
\end{algorithm}
Setting $i=1$, the layers simplifies to $A_1 - A_4 - A_5 = C_1 - C_2 - C_3$ (the case for $i=2$ is perfectly symmetrical). For each of these graphs we construct $\mathcal{D}(G^{T_{i,i+3}(G)})$, which completes our construction for part 2.
\subsubsection{Data Structure Queries for Part 2}

To check whether a $4$-path exists between $a_1 \in A_1$ and $a_5 \in A_5$, the algorithm queries the data structures $\mathcal{D}(G^{T_{i,i+3}(G)})$ across all indices $i$ to verify the existence of a $2$-path between $a_1$ and $a_5$ in the auxiliary graphs. If such a $2$-path is found in at least one auxiliary graph $G^{T_{i,i+3}(G)}$, the query $\texttt{exist-path}(a_1, a_5)$ returns \texttt{true}.

For $\texttt{all-paths}(a_1, a_5)$, the algorithm lists all $2$-paths between $a_1$ and $a_5$ across all auxiliary graphs of type $G^{T_{i,i+3}(G)}$. For each retrieved $2$-path—for example, $a_1 - a_4 - a_5$ (and symmetrically for $a_1 - a_2 - a_5$), we retrieve the paths from the table entry $T_{i,i+3}(G)[a_1, a_4]$. For each path $P = a_1 - a_2' - a_3' - a_4$ found in the table, we append $a_5$ to construct and list $a_1 - a_2' - a_3' - a_4 - a_5$ as a $4$-path in the original graph $G$.

\subsubsection{Complexity Analysis}
During the construction of $\mathcal{D}(G^{T_{i,i+3}(G)})$, we create the table $T_{1,3}(G^{T_{i,i+3}(G)})$. We must prove the time it takes to create this table is $O(n^2+t_{8}(G))$.

The algorithm to create the table $T_{1,3}(G^{T_{i,i+3}(G)})$ takes $O(n^2+|T_{1,3}(G^{T_{i,i+3}(G)})|)$, as shown in Algorithm~\ref{alg1}. We next prove $|T_{1,3}(G^{T_{i,i+3}(G)})|=O(n^2+t_{2k}(G))$ for any constant $k\ge 4$. The cases $i=1$ and $i=2$ are symmetrical, so we conduct the proof for $i=1$, so $V(G^{T_{i,i+3}(G)})=C_1-C_2-C_3=A_1-A_4-A_5$. According to Theorem~\ref{important_remark}, $|T_{1,3}(G^{T_{i,i+3}(G)})|=O(n^2+t'_{2k-4}(G^{T_{i,i+3}(G)}))$, where $t'$ means only cycles which have 2 edges in $E[C_1,C_2]$, and $2k-6$ edges in $E[C_2,C_3]$. We now prove $t'_{2k-4}(G^{T_{i,i+3}(G)})<t_{2k}(G)$.

\begin{lemma}
Let $P^* = a_1 - a_2^* - a_3^* - a_4'$ and $P = a_1 - a_2 - a_3 - a_4''$ be sparse $3$-paths in $G$. If $a_3 = a_3^*$, then $a_2 = a_2^*$.
\end{lemma}
\begin{proof}
For the sake of contradiction, assume $a_2 \neq a_2^*$. Then, the sequence $a_1 - a_2 - a_3 - a_4'$ forms a path in $G$, which contradicts the assumed sparsity of $P^*$.
\end{proof}

\begin{theorem}
$t'_{2k-4}(G^{T_{i,i+3}(G)})<t_{2k}(G)$.
\end{theorem}
\begin{proof}
Take a $2k-4$-cycle $C$ in $G^{T_{i,i+3}(G)}$ which has two edges in $E[C_1,C_2]$, call the edges $(a_1,a_4'),(a_1,a_4'')$. The entry $T[a_1, a_4']$ contains two sparse paths: $P^* = a_1 - a_2^* - a_3^* - a_4'$ and $P^{**} = a_1 - a_2^{**} - a_3^{**} - a_4'$. By the sparsity property, we know $a_2^* \neq a_2^{**}$ and $a_3^* \neq a_3^{**}$. The entry $T[a_1, a_4'']$ likewise contains at least two sparse paths; let us select one such path $P = a_1 - a_2 - a_3 - a_4''$. 

Suppose $P$ intersects with $P^*$ at an internal vertex (excluding $a_1$). By our lemma, if $a_3 = a_3^*$, then $a_2 = a_2^*$. Thus, any intersection implies $a_2 = a_2^*$. By expanding this logic, if $P$ were to intersect with both $P^*$ and $P^{**}$, it would force $a_2^* = a_2^{**}$, which is false. Therefore, $P$ can intersect with at most one of the paths. Without loss of generality, assume $P$ does not intersect with $P^*$. It follows that taking the paths $P,P^*$ in $G$, along with the edges of $C$ in $E[A_4,A_5]$, creates a $2k$-cycle in $G$. This is because we added 4 vertices to the cycle, and no intersections. So each $2k-4$-cycle $C$ in $G^{T_{i,i+3}(G)}$ which has two edges in $E[C_1,C_2]$ maps to a $2k$-cycle in $G$.
\end{proof}

\subsection{Part 3:  Type 3 paths and the table $R$}
Up to this point, our data structure handles all paths satisfying either of the following properties:
\begin{enumerate}
    \item The path $a_1 - a_2 - a_3 - a_4 - a_5$ contains an internal vertex $a_i$ ($i \in \{2, 3, 4\}$) that can be replaced by an alternative vertex $a_i'$ to yield another path. This case is covered during part 1.
    \item Assuming the first condition is unmet (implying the length $3$ sub-paths are sparse), the path contains internal vertices $a_i, a_{i+1}$ ($i \in \{2, 3\}$) that can be replaced by $a_i', a_{i+1}'$ to form another path, such that $a_{i-1} - a_i' - a_{i+1}' - a_{i+2}$ is also a sparse path. This case is covered during part 2.
\end{enumerate}

This leaves exactly two remaining scenarios to address:
\begin{enumerate}
    \item[A.] The path contains internal vertices $a_i, a_{i+1}$ ($i \in \{2, 3\}$) that can be replaced by $a_i', a_{i+1}'$ to form another path, but the sub-path $a_{i-1} - a_i' - a_{i+1}' - a_{i+2}$ is \emph{not} a sparse path.
    \item[B.] The entire path forms a sparse $4$-path.
\end{enumerate}

For scenario A, notice that the alternative path using $a_i'$ and $a_{i+1}'$ is a Type 1 path, since the non-sparsity of $a_{i-1} - a_i' - a_{i+1}' - a_{i+2}$ satisfies the condition for type 1 paths. This guarantees that for every such path containing a sparse $3$-path between $a_i$ and $a_{i+3}$, there exists an alternative path from $a_1$ to $a_5$ utilizing $a_i$ and $a_{i+3}$ that is a type 1 path. For these paths we will construct the table $R$ which will be used to list the type 3 paths in a different way than the other types, which uses the listing of the type 1 paths (which we've already described).

\subsubsection{Construction of $R$ Table}
All such sparse $3$-paths can be efficiently discovered in $O(n^2)$ time by scanning $T_{i,i+3}(G)$ for pairs $(a_i,a_{i+3})\in A_i\times A_{i+3}$ where $|T_{i,i+3}(G)[a_i,a_{i+3}]| = 1$, and then checking the tables of $G^{T_{i,i+2}(G)}$ to ensure that an alternative, non-sparse path connects $a_i$ to $a_{i+3}$. We enter these paths into the table $R_{i,i+3}$, whose size is bounded by $O(n^2)$.
\begin{algorithm}[htbp]
  \caption{Compute part 3 of the data structure}
  \label{alg:compute_part3}
  \DontPrintSemicolon
  
  \For{$i \in \{1,2\}$}{
    \For{$(a_i,a_{i+3})\in A_i\times A_{i+3}$}{
      \If{$[|T_{i,i+3}(G)[a_i,a_{i+3}]|=1]$ \textbf{and} \\ $[T_{i,i+2}(G^{T_{i,i+2}(G)})[a_i,a_{i+3}]|>0$ \textbf{or} $|T_{i,i+2}(G^{T_{i+1,i+3}(G)})[a_i,a_{i+3}]|>0]$}{
        $a_i-a_{i+1}-a_{i+2}-a_{i+3}$ is the entry in $T_{i,i+3}(G)[a_i,a_{i+3}]$\;
        $R_{i,i+3}(G)[a_i,a_{i+3}] \leftarrow (a_i-a_{i+1}-a_{i+2}-a_{i+3})$\;
      }
    }
  }
\end{algorithm}
\\
\begin{lemma}
For every Type 3 path $P=a_1-a_2-a_3-a_4-a_5$, either $R_{1,4}[a_1,a_4]=a_1-a_2-a_3-a_4$, or $R_{2,5}[a_2,a_5]=a_2-a_3-a_4-a_5$
\end{lemma}
\begin{proof}
We know that $P$ contains internal vertices $a_i,a_{i+1}$ that can be replaced by $a_i',a_{i+1}'$ to form another path, and $a_{i-1}-a_i'-a_{i+1}'-a_{i+2}$ is not sparse. Assume $i=2$, the case $i=3$ is proven the same way. The path $a_1-a_2'-a_3'-a_4$ is not sparse. When creating the table $R_{1,4}$ we traverse the table $T_{1,4}(G)$, checking every entry $(a_1',a_4')\in A_1\times A_4$. when reaching $(a_1,a_4)$, we have $|T_{1,4}[a_1,a_4]|=1$, because otherwise $P$ is a type 2 path. Additionally, since $a_1-a_2'-a_3'-a_4$ is not sparse, that either means $(a_1,a_3')\in E(G^{T_{1,3}(G)})$ which means $a_1-a_3'-a_4$ is a $2$-path in $G^{T_{1,3}(G)}$, or $(a_2',a_4)\in E(G^{T_{2,4}(G)})$ which means $a_1-a_2'-a_4$ is a $2$-path in $G^{T_{2,4}(G)}$. In the first case we have $[T_{1,3}(G^{T_{1,3}(G)})[a_1,a_{4}]|>0$, and in the second we have $[T_{1,3}(G^{T_{2,4}(G)})[a_1,a_{4}]|>0$. In the algorithm, one of these being true, along with $|T_{1,4}[a_1,a_4]|=1$ is the requirement for putting $a_1-a_2-a_3-a_4$ in $R_{1,4}[a_1,a_4]$, as needed.
\end{proof}

\subsubsection{Data Structure Queries for Part 3}
During an $\texttt{exist-path}$ query, we can ignore type 3 paths; if a  path from $a_1$ to $a_5$ is a type 3 path, there must exist an alternative path from type 1, so the data structure will return yes. \\

During an $\texttt{all-paths}$ query, when listing type 1 paths $a_1 - a_2 - a_3 - a_4 - a_5$, we check if the sub-path from $a_i$ to $a_{i+3}$ can be substituted using entries from $R[a_i, a_{i+3}]$.
\begin{lemma}
If we list every type 1 path, we also list every type 3 path.
\end{lemma}
\begin{proof}
Let $P=a_1-a_2-a_3-a_4-a_5$ be a type 3 path. This means that for either $i=2$ or $i=3$, there exist vertices $a_i',a_{i+1}'$ such that $a_{i-1}-a_i'-a_{i+1}'-a_{i+2}$ is not sparse. Assume $i=2$, the proof is the same for $i=3$. $a_1-a_2'-a_3'-a_4-a_5$ is a type 1 path, and so when it is listed, we check $R_{1,4}[a_1,a_4]$ to find $a_1-a_2-a_3-a_4$, and list the path $P$, as needed. 
\end{proof}
\subsection{Part 4: Type 4 Paths}

To construct the table $T_{1,5}(G)$ containing all sparse $4$-paths, we adapt the sorting and merging framework of Algorithm~\ref{alg:compute_tablesM}.
\\
We construct two lists, $L_{1,2,3,4}$ and $L_{2,3,4,5}$, of size $O(n^2)$. The first list, $L_{1,2,3,4}$, contains all sparse $3$-paths $a_1-a_2-a_3-a_4$ satisfying $|T_{1,4}(G)[a_1, a_4]| = 1$ that do not share endpoints with any non-sparse path (in other words, $|T_{1,3}(G^{T_{i,i+2}(G)})[a_1, a_4]|=0$ for $i\in {1,2}$); $L_{2,3,4,5}$ is defined symmetrically for layer endpoints $a_2$ and $a_5$. 

We then radix-sort both lists lexicographically using the vertex tuple $(a_2, a_3, a_4)$, as the sorting key, similarly to Algorithm~\ref{alg:compute_tablesM}. We then conduct a similar pointer algorithm to Algorithm~\ref{alg:compute_tablesM}. We initialize two pointers at the beginning of the lists $L_{1,2,3,4}$ and $L_{2,3,4,5}$. If the current entry of $L_{1,2,3,4}$ has a tuple $(a_2,a_3,a_4)$ smaller (using the radix comparison) then the one in $L_{2,3,4,5}$, we move its pointer forward. If the tuple is bigger, we move the pointer of $L_{2,3,4,5}$ forward. If they are equal, we find all paths in both lists that contain that tuple, extracts the sets $S_1 = \{a_1 \mid (a_1-a_2-a_3-a_4) \in L_{1,2,3,4}\}$ and $S_5 = \{a_5 \mid (a_2-a_3-a_4-a_5) \in L_{2,3,4,5}\}$, and enters the combinations into $T_{1,5}(G)[a_1, a_5]$.

\begin{lemma}
Every colorful sparse $4$-path in $G$ is correctly inserted into $T_{1,5}(G)$ in $O(n^2 + |T_{1,5}(G)|)$ time.
\end{lemma}
\begin{proof}
Let $P = a_1 - a_2 - a_3 - a_4 - a_5$ be a sparse $4$-path. By definition, its prefix $a_1-a_2-a_3-a_4$ and suffix $a_2-a_3-a_4-a_5$ are sparse $3$-paths, which means they are included in $L_{1,2,3,4}$ and $L_{2,3,4,5}$, respectively. Because they share the key $(a_2, a_3, a_4)$, the lexicographical merge guarantees they are included in the sets $S_1$ and $S_5$ for the key$(a_2,a_3,a_4)$. 
\end{proof}
The $O(n^2 + |T_{1,5}(G)|)$ runtime analysis follows similarly to Algorithm~\ref{alg:compute_tablesM}. Inserting paths into the lists $L_{1,2,3,4}$ and $L_{2,3,4,5}$ takes $O(n^2)$, sorting the lists takes $O(n^2)$ using radix sort, moving the pointers forward takes $O(n^2)$ because the list sizes are at most $O(n^2)$, and insertions into $T_{1,5}(G)$ take $O(T_{1,5}(G))$.
$T_{1,5}(G)=O(n^2+t_8(G))$ because each two entries of the same index in the table are disjoint due to sparsity, and create an $8$-cycle. So for each pair $(a_1,a_5)$, the number of cycles created by the paths in $T_{1,5}[a_1,a_5]$ is at least $T_{1,5}[a_1,a_5]-1$. From here, $|T_{1,5}(G)|=\sum _{(a_1,a_5)\in A_1\times A_5}|T_{1,5}[a_1,a_5]|=n^2+\sum _{(a_1,a_5)\in A_1\times A_5}(|T_{1,5}[a_1,a_5]|-1)\le n^2+t_{2k}(G)$

\subsection{Final Analysis and Complexity}
With all tables constructed, for each pair $(a_1, a_5) \in A_1 \times A_5$, we can now execute $\texttt{exist-path}(a_1, a_5)$ with $O(1)$ complexity, and $\texttt{all-paths}(a_1,a_5)$ with $O(P)$ complexity where $P$ is the number of listed paths. \\
$\texttt{exist-path}(a_1, a_5):$ Check if there exists a path in types $1, 2$, or $4$, and return $\texttt{yes/no}$ accordingly. There is no need to check type 3, because if there exists a path of type 3 between $a_1$ and $a_5$, there is also a path of type $1$.\\
$\texttt{all-paths}(a_1,a_5):$ This function will list the all paths from types $1,2,$ and $4$. Note that the algorithm used to list type 2 paths can also list type 1 paths, so we keep the listed paths in a hash table so as to not list twice. We now wish to use the tables $R$ to list paths of type 3. For each path listed of type 1, $a_1-a_2-a_3-a_4-a_5$, we will check $R_{1,4}(G)[a_1,a_4]$ in order to check if there is another path from $a_1$ to $a_4$, $a_1-a_2'-a_3'-a_4$, from type 3 that we need to list. We will similarly check $R_{2,5}(G)[a_2,a_5]$. For each type 1 path this only takes a constant amount of time, so the complexity is not changed. As shown in the analysis of type 3 paths, checking the tables $R$ for all type 1 paths will list all type 3 paths, as needed.

The complexity is $O(P)$, where $P$ is the number of paths listed during the query. This completes the construction of the path listing data structure in time $O(n^2+t_8(G))$, as needed for the earlier reductions.

\section{Generalization to $2k$-Cycle Enumeration for Larger Constants}\label{sec:general_larger_cycles}
Having established an algorithm for listing $8$-cycles, we now generalize these techniques to list larger even cycles, specifically targeting $10,12,14,$ and $16$-cycle listing. While the path listing data structure can be extended to handle graphs with 6,7,8 or 9 layers, a direct adaptation of the algorithm encounters a problem within certain auxiliary graphs.

For example, in a 6-layer graph, we cannot prove the algorithm creates the table $T_{1,3}(H)$ for $H=G^{T_{1,3}{(G^{T_{3,5}(G)})}}$ in $O(n^2+t_{10}(G))$ time. The main cause of this is that both boundary layer transitions consist of auxiliary edges, so we cannot create simple paths in the boundary layer transition as we did in Theorem~\ref{thm:3-sparse_paths}. \\
According to Theorem~\ref{size_of_2-path_table}, the size of $T_{1,3}(H)$ for $H=G^{T_{1,3}{(G^{T_{3,5}(G)})}}$ is $O(n^2+t'_{2r}(H))$ and also $O(n^2+t''_{2r}(H))$ for a fixed constant $r\ge 2$. Every edge in $H$ represents (at least two) $2$-paths in the original graph $G$, so a $2r$-cycle in $H$ is a $4r-$cycle in $G$. We run into two problems here. The first is that this does not imply $t_{2r}(H)\le t_{4r+2}(G)$, because our mapping techniques map the auxiliary edges to paths in $T_{1,3}(G)$ or $T_{3,5}(G)$, both of which contain $2$-paths, so the length of the path will be divisible by 4. The second is that a cycle in $H$ can contain 3 or more edges in a boundary layer transition, for example $E[A_3,A_5]$. Call 3 of those edges $(a_3^1,a_5^1),(a_3^2,a_5^2),(a_3^3,a_5^3)$. It could be that $T_{3,5}(G)[a_3^1,a_5^1]=T_{3,5}(G)[a_3^2,a_5^2]=T_{3,5}(G)[a_3^3,a_5^3]=\{a_4',a_4''\}$. In this case, there is no way to map these edges to paths in the original graph $G$ without a collision. The other boundary layer is also made of auxiliary edges, so we run into a similar problem.
\\To resolve this, we introduce the generalized concept of \textbf{mildly sparse paths}. After introducing this concept, we will define an algorithm that constructs tables and auxiliary graphs as before, but based on mildly sparse paths instead of sparse paths.\\
This section will focus mainly on adapting the same methods used in Section~\ref{sec:five_layers} in order to build the data structure for graphs with more layers, while incorporating mildly sparse paths. \\
In the next section we will use traits of mildly sparse paths, as well as a new concept called the Boundary Layer Property, in order to prove tight complexity bounds for graphs with up to $9$ layers.

\subsection{Mildly Sparse Paths}
The definition of a mildly sparse path is recursive. Any path of length 2 is mildly sparse. A $3$-path $a_i - a_{i+1} - a_{i+2} - a_{i+3}$ is denoted mildly sparse if both the number of alternative $2$-paths connecting $a_i$ to $a_{i+2}$ and the number of alternative $2$-paths connecting $a_{i+1}$ to $a_{i+3}$ is strictly less than a fixed constant $2c$, for $c = 10k$. Here $k$ is the number of layer transitions in the original graph $G$, when creating the data structure for auxiliary graphs with fewer layers, we use the same constant $c$.
For each set of mildly sparse paths from $A_i$ to $A_{i+m}$, we will construct $T_{i,i+m}(G)$ similarly to previous sections. For $(a_i,a_{i+m})\in A_i\times A_{i+m}$, $T_{i,i+m}(G)[a_i,a_{i+m}]$ will have all the mildly sparse paths between $a_i$ and $a_{i+m}$. \\
We now define the auxiliary graph $G^{T_{i,i+m}(G)}$ similarly to our previous definition. The layers between $A_i$ and $A_{i+m}$ are contracted. Additionally, we define a function $\lambda(m)$ to specify the minimal number of mildly sparse paths required between $a_i$ and $a_{i+m}$ for $(a_i, a_{i+m})$ to be an edge in $G^{T_{i,i+m}(G)}$. This function is defined recursively as:
$$\lambda(m) = \begin{cases} 1 & m = 1 \\  c \cdot m \sum_{j=1}^{m-1} \lambda(j)\lambda(m-j) & m \ge 2 \end{cases}$$
\begin{remark}
    $\lambda(2)=2c$.
\end{remark}
The reason $\lambda$ is defined in this way can be seen in Theorem~\ref{bounding_intersections} and Lemma~\ref{c_vertex-disjoint_paths}, where we use this definition in order to prove important qualities of mildly sparse paths.\\
The definition for $(m+1)$-mildly sparse path will then depend heavily on the definition of an $m$-mildly sparse path. An $(m+1)$-path $a_1 - a_2 - \dots - a_{m+2}$ is defined as mildly sparse if it satisfies the following conditions:
\begin{enumerate}
    \item The sub-path from $a_1$ to $a_{m+1}$ is a mildly sparse $m$-path.
    \item The sub-path from $a_2$ to $a_{m+2}$ is a mildly sparse $m$-path.
    \item There are no non-mildly sparse paths from $a_1$ to $a_{m+1}$ or from $a_2$ to $a_{m+2}$.
    \item The vertex pairs $(a_1, a_{m+1})$ and $(a_2, a_{m+2})$ are not edges in $G^{T_{1,m+1}(G)}$ or $G^{T_{2,m+2}(G)}$ respectively. 
\end{enumerate}
The last condition is equivalent to the condition that the number of $m$-mildly sparse paths between $(a_1, a_{m+1})$ is strictly less than $\lambda(m)$, and the number of $m$-mildly sparse paths between $(a_2, a_{m+2})$ is strictly less than $\lambda(m)$.
Additionally, we will also construct the tables $R_{i,i+m}(G)$ defined as follows:\\
A table index $R_{i,i+m}(G)[a_i,a_{i+m}]$ will hold all $m$-mildly sparse paths from $a_i$ to $a_{i+m}$ such that $|T_{i,i+m}(G)[a_i,a_{i+m}]|<\lambda(m)$, but there is a non mildly sparse path between $a_i$ and $a_{i+m}$.\\
Before moving on to extend the data structure using mildly sparse paths, we prove some helpful lemmas on mildly sparse paths. Afterwards, we will use mildly sparse paths, as well as the tables and auxiliary graphs based on mildly sparse paths defined above in order to construct the data structure $\mathcal{D}(G)$ for $k+1$-layer graphs $G$.

\begin{lemma} \label{basic_mild_sparsity}
    Given $i>0, 0\le j<l\le m$, if a path $P=a_{i+j}-a_{i+j+1}-\dots-a_{i+l}$ is not mildly sparse, then the path $a_i-a_{i+1}-\dots -a_{i+j}-\dots- a_{i+l}-\dots -a_{i+m}$ which contains $P$ is also not mildly sparse.
\end{lemma}
\begin{proof}
    We know that if $a_{i+j}-a_{i+j+1}-\dots-a_{i+l}-a_{i+l+1}$ is mildly sparse, then $P$ is also mildly sparse, according to condition 1 of mild sparsity. This means $a_{i+j}-a_{i+j+1}-\dots-a_{i+l}-a_{i+l+1}$ is not mildly sparse. We can do the same proof now with $P'=a_{i+j}-a_{i+j+1}-\dots-a_{i+l}-a_{i+l+1}$ and show that $a_{i+j}-a_{i+j+1}-\dots-a_{i+l+1}-a_{i+l+2}$ is not mildly sparse, and so on until we get that $a_{i+j}-a_{i+j+1}-\dots-a_{i+m-1}-a_{i+m}$ is not mildly sparse.
    
    Similarly $a_{i+j-1}-a_{i+j}-\dots-a_{i+l}-a_{i+m}$ is not mildly sparse, using condition 2 of mild sparsity. We can use this proof $j$ times to show that $a_i-\dots a_{i+m}$, is not mildly sparse, as needed.
\end{proof}

\begin{lemma} \label{bound_sparse_paths}
Given a mildly sparse path $P=a_i-a_{i+1}-\dots -a_{i+m}$, and indices $0\le j <l \le m$ such that $l-j<m$, there are at most $\lambda(l-j)$ colorful paths between $a_{i+j}$ and $a_{i+l}$.
\end{lemma}
\begin{proof}
All sub-paths of $P$ are mildly sparse as shown in Lemma~\ref{basic_mild_sparsity}. We know either $j\ne 0$ or $l\ne m$, assume $j\ne 0$ (the proof is similar if $l\ne m$). We know the path $a_{i+j-1}-\dots -a_{i+l}$ is sparse, which means the pair $(a_{i+j},a_{i+l})$ is not an edge in $G^{T_{i,i+m}(G)}$, so the number of colorful mildly sparse paths between $a_{i+j}$ and $a_{i+l}$ is at most $\lambda(l-j)-1$. We also know there are no colorful non mildly sparse paths between the pair, according to requirement 3 of the definition of mildly sparse. 
\end{proof}
\begin{theorem} \label{bounding_intersections}
    For a given $m$-mildly sparse path $P$ between $a_i$ and $a_{i+m}$ (assuming $m>1$), there are at most $\sum_{j=1}^{m-1} \lambda(j)\lambda(m-j)-1$ other paths between $a_i$ and $a_{i+m}$ that intersect with $P$ (in a vertex other than $a_i$ and $a_{i+m}$).
\end{theorem}
\begin{proof}
Take a vertex $a_{i+j}$ between $a_i$ and $a_{i+m}$. The amount of paths between $a_i$ and $a_{i+j}$ is at most $\lambda(j)$, according to Lemma~\ref{bound_sparse_paths}, and the number of paths between $a_{i+j}$ and $a_{i+m}$ is at most $\lambda(m-j)$. So the overall number of paths between $a_i$ and $a_{i+m}$ intersecting at $a_{i+j}$ is at most $\lambda(j)\lambda(m-j)$. Taking a sum over all vertices between $a_i$ and $a_{i+m}$, we get that the number of paths from $a_i$ to $a_{i+m}$ intersecting with $P$ is at most $\sum_{j=1}^{m-1} \lambda(j)\lambda(m-j)$, as needed.
\end{proof}
\begin{lemma} \label{c_vertex-disjoint_paths}
For any edge $(a_i, a_{i+m})$ in $G^{T_{i,i+m}(G)}$, there are $c \cdot m$ pairwise internally vertex-disjoint mildly sparse paths connecting $a_i$ to $a_{i+m}$.
\end{lemma}
\begin{proof}
Let $W$ be the set of mildly sparse paths between $a_i$ and $a_{i+m}$.
Using Theorem~\ref{bounding_intersections} we can greedily extract a set of $c \cdot m$ mutually internally vertex-disjoint mildly sparse paths connecting $a_i$ to $a_{i+m}$ by selecting an arbitrary path $P_1$ in $W$, discarding all paths intersecting it (at most $\sum_{j=1}^{m-1} \lambda(j)\lambda(m-j)-1$), selecting a path $P_2$ in $W$, discarding all paths intersecting it, and so on. There are at least $\lambda(m)=cm\sum_{j=1}^{m-1} \lambda(j)\lambda(m-j)$ mildly sparse paths between $a_i$ and $a_{i+m}$. After selecting $x$ paths in $W$ we've either added to $W$ or discarded at most $x\sum_{j=1}^{m-1} \lambda(j)\lambda(m-j)$ paths. So we can do this at least until $x=cm$, as needed.
\end{proof}
The next theorem will be used in large part to map auxiliary edges to edges in the original graph in the complexity proofs.
\begin{theorem}\label{distinct-mildly-sparse-paths}
    For any distinct pair of edges $(a_i, a_{i+m})$ and $(b_i, b_{i+m})$ in $G^{T_{i,i+m}(G)}$, there exist mutually internally vertex-disjoint mildly sparse paths $P_a$ and $P_b$ connecting their respective endpoints in $G$.
\end{theorem}
\begin{proof}
    Each edge in the auxiliary graph represents at least $\lambda(m)$ distinct mildly sparse paths. We proved in Lemma~\ref{c_vertex-disjoint_paths} that there are at least $cm$ mutually internally vertex-disjoint mildly sparse paths connecting $a_i$ to $a_{i+m}$. We take $cm$ internally vertex disjoint mildly sparse paths connecting $a_i$ to $a_{i+m}$, denoted by $\{P_1, P_2, \dots, P_{cm}\}$.
    Now, let $P_b$ be an arbitrary mildly sparse path connecting $b_i$ to $b_{i+m}$. Since $P_b$ contains exactly $m-1$ internal vertices, and each internal vertex can intersect at most one path from our internally vertex-disjoint path set, $P_b$ can intersect at most $m-1$ paths from the set $\{P_1, P_2, \dots, P_{cm}\}$ in an internal vertex. This leaves at least $(c-1)m > 1$ internally vertex-disjoint paths available for $P_a$, completing the theorem.
\end{proof}
\begin{theorem}\label{many-distinct-mildly-sparse-paths}
    For a number $x<c$ and any distinct set of edges $(a_i^1, a_{i+m}^1), (a_i^2, a_{i+m}^2), \dots (a_i^x, a_{i+m}^x)$ in $G^{T_{i,i+m}(G)}$, there exist internally vertex-disjoint mildly sparse paths $P_1,P_2,...P_x$ connecting their respective endpoints in $G$.
\end{theorem}
\begin{proof}
    Each edge in the auxiliary graph represents at least $\lambda(m)$ distinct mildly sparse paths. We proved in Lemma~\ref{c_vertex-disjoint_paths} that there are at least $cm$ internally vertex-disjoint mildly sparse paths connecting $a_i^j$ to $a_{i+m}^j$ for $1\le j \le x$. Assume we have internally vertex-disjoint mildly sparse paths for the first $j$ edges in the auxiliary graph, we prove that we can choose a mildly sparse paths for the edge $(a_i^{j+1}, a_{i+m}^{j+1})$.
    
    We take $cm$ internally vertex disjoint mildly sparse paths connecting $a_i^{j+1}$ to $a_{i+m}^{j+1}$, denoted by $\{Q_1, Q_2, \dots, Q_{cm}\}$.
    $j<c$, so the number of paths we have to avoid intersecting with is less than $c$. Each of those paths has $m-1$ internal vertices, meaning we have to avoid at most $c(m-1)<cm$ internal vertices. Each of these vertices can be in at most one path $Q_y$, which leaves at least $c$ paths that have none of these vertices in them. We choose one of those paths to be $P_{j+1}$. We can continue this process until we reach $x$, proving the theorem.
    
\end{proof}

\subsection{Types of paths}
As in the data structure for 5 layer graphs, we wish to split the paths into types and define the construction and queries of the data structure such that $\texttt{exist-path}$ will check if there exists a path from any of these types, and $\texttt{all-paths}$ will list the paths of all these types. We first define the types, then prove they are mutually exhaustive.\\
For a $(k+1)-$layer graph, we split possible paths into $2(k-2)$ different types. a path $P=a_1-a_2-\dots -a_{k+1}$ belongs to one of the following types:
\\ A path $P$ is in type $1$ if there is a sub-path of size 2 in $P$, $a_i-a_{i+1}-a_{i+2}$ such that the number of mildly sparse 2 paths (every $2$-path is mildly sparse) between $a_i$ and $a_{i+2}$ is $\ge \lambda(2)$.
\\A path $P$ belongs to type 2 if the first condition is false, and additionally, there is a sub-path of size 3 in $P$, $a_i-a_{i+1}-a_{i+2}-a_{i+3}$, such that the number of mildly sparse 3 paths between $a_i$ and $a_{i+3}$ is $\ge \lambda(3)$.\\
$P$ belongs to type 3 if the first 2 conditions are false, and additionally, there is a sub-path of size 3 in $P$, $a_i-a_{i+1}-a_{i+2}-a_{i+3}$, such that the number of mildly sparse 3 paths between $a_i$ and $a_{i+3}$ is $< \lambda(3)$, and there is also a non mildly sparse path between $a_i$ and $a_{i+3}$. 
\\Generally speaking, for some $1<j<k-2$, a path is a type $2j$ path if The first $(2j-1)$ conditions are false, and additionally, there is a sub-path of size $(j+2)$ in $P$, $a_i-a_{i+1}-\dots-a_{i+j+2}$, such that the number of mildly sparse $(j+2)-$paths between $a_i$ and $a_{i+j+2}$ is $\ge \lambda(j+2)$. 
\\A path is a type $2j+1$ path if the first $2j$ conditions are false. Additionally, there is a sub-path of size $(j+2)$ in $P$, $a_i-a_{i+1}-\dots-a_{i+j+2}$, such that the number of mildly sparse $(j+2)-$paths between $a_i$ and $a_{i+j+2}$ is $<\lambda(j+2)$, and there is also a non mildly sparse path between $a_i$ and $a_{i+j+2}$.
\\Finally, a path is type $2(k-2)$ if all previous conditions are false, so the path is a $k-$mildly sparse path (if $k=2$ then this is condition 1, not 0). 
\begin{theorem}
Every path $P$ belongs to one of these conditions.
\end{theorem}
\begin{proof}
We first prove a lemma that will be the main part of the theorem's proof.
\begin{lemma}
If a path $P=a_1-a_2-\dots -a_{k+1}$ doesn't belong to the first $2j+1$ conditions, then every sub-path of size $(j+3)$ in $P$ is mildly sparse. 
\end{lemma}
\begin{proof}
Prove this by induction on $j$. If $j=0$, $P$ not being type 1 means that no vertex pairs $(a_i,a_{i+2})$ in the path are an edge in $G^{T_{i,i+2}(G)}$ (this is direct by how we defined $G^{T_{i,i+2}(G)}$). So for every size 3 sub-path in $P$, condition 4 of mild sparsity is satisfied. Conditions 1,2,3 are trivially satisfied because all $2$-paths are mildly sparse.\\
Assume the statement is true for every number below some $j<k-2$, and take a path $P$ that doesn't belong to the first $2j+1$ types. $P$ doesn't belong to the first $2j+1$ types, meaning it also doesn't belong to the first $2j-1$ types. So according to the induction hypothesis, every $(j+2)$ sub-path of $P$ is mildly sparse. So for each $(j+3)$ sub-path, conditions 1 and 2 of mild sparsity are correct. Since condition $2j$ is false, every sub path of length $j+2$, $a_i-\dots -a_{i+j+2}$, satisfies $|T_{i,i+j+2}(G)[a_i,a_{i+j+2}]|<\lambda(j+2)$. So for all such sub-paths, there is no edge $(a_i,a_{i+j+2})$ in $G^{T_{i,i+j+2}(G)}$, which means condition 4 of mild sparsity is correct for all $(j+3)$ sub-paths. Finally, condition $2j+1$ is also false, which means that for all sub-path of length $j+2$, $a_i-\dots -a_{i+j+2}$, there is no non mildly sparse paths between $a_i$ and $a_{i+j+2}$, so condition 3 is satisfied for length $(j+3)$ sub-paths. So every $(j+3)$ sub-path in $P$ is mildly sparse, as needed.
\end{proof}
So if the path doesn't belong in any of the first $2(k-2)-1$ conditions, then it must be a mildly sparse $k-$path and belong to the last condition.
\end{proof}
In the rest of this section we will define the construction of the data structure for $k+1$-layer graphs, as well as the queries on the data structure. In the next section we will prove that for $k\le 8$, the construction works in $O(n^2+t_{2k}(G))$.
\subsection{Construction of $\mathcal{D}(G)$ for a $(k+1)-$Layer Graph}
In this subsection we show the full construction of the data structure for $(k+1)-$layer graphs in time $O(n^2+|\mathcal{D}(G)|)$. The complexity analysis bounding $|\mathcal{D}(G)|$ is done in the next section. We note that the algorithm works for any number of layers, but for $k\ge 9$ the complexity analysis becomes problematic. The algorithm will use induction heavily.\\
The construction of $\mathcal{D}(G)$ for a $3-$layer graph is simply to create the table $T_{1,3}(G)$ with Algorithm~\ref{alg:compute_tablesT}, and use the same functions for $\texttt{exist-path}$ and $\texttt{all-paths}$ described in Section~\ref{alg1}.
\\ Since the algorithm relies on recursion, when creating $\mathcal{D}$ for $k+1$ layers we assume the algorithm works for graphs with $k$ or less layers. The algorithm will use the recursion to graphs with $k$ or less layers a large constant amount of times.\\
We use this for the algorithm to construct the data structure, and the query algorithms for graphs with $k+1$ layers. In this subsection we show the construction of the data structure. Afterwards we will explain how to query the data structure, and finally we will prove complexity bounds on the construction.

The algorithm to build $\mathcal{D}(G)$ for a $k+1$-layer graph G is:
\begin{algorithm}[htbp]
  \caption{Compute part 1 of the data structure}
  \label{alg:compute_k+1-layer}
  \DontPrintSemicolon
    \For{$j \in \{1,2,\dots ,k-2\}$}{
      \For{$i \in \{1,2,\dots ,k-j\}$}{
      Create $\mathcal{D}(G[A_i\cup A_{i+1}\dots \cup A_{i+j+1}])$ (using the algorithm for $j+2$-layer graphs). \\
       Take $T_{i,i+j+1}(G)$ from $\mathcal{D}(G[A_i\cup A_{i+1}\dots \cup A_{i+j+1}])$, and create the auxiliary graphs $G^{T_{i,i+j+1}(G)}$ and the data structure $\mathcal{D}(G^{T_{i,i+j+1}(G)})$ (using the known algorithm for $(k-j+1)-$layer graphs).\\
       Take $R_{i,i+j+1}(G)$ from the created data structure $\mathcal{D}(G[A_i\cup A_{i+1}\dots \cup A_{i+j+1}])$. 
      }
    }
    Create tables $T_{1,k+1}(G)$ and $R_{1,k+1}(G)$.
\end{algorithm}

For most of the algorithm we use recursion to create data structures for graphs with $k$ or less layers. Assuming these steps work, we are left with the last part of the algorithm, to construct the tables $T_{1,k+1}(G), R_{1,k+1}(G)$.
We now describe the algorithm to build these tables. when creating the tables for step $j$, we assume all the previous steps were already conducted. The algorithm we define is similar to Algorithm~\ref{alg:compute_tablesM}. Notice that since the algorithm is used recursively on graphs with $k$ or less layers, the size of the paths in the tables we construct is $r\le k$.
\begin{theorem}
    Let $G$ be an $r+1$-layer graph. We define an algorithm which takes $O(n^2 + |T_{1,r+1}(G)|)$ time to construct the tables $T_{1,r+1}(G)$ and $R_{1,r+1}(G)$. We assume that the tables $T_{i,i+j}(G), R_{i,i+j}(G)$ have been created for all $i\le r+1-j$, $j<r$. We assume $\mathcal{D}(G^{T_{i,i+j}(G)})$ have also been created.
\end{theorem}
\begin{proof}
The algorithm constructs a list $L_{1,2,...r}$ containing all $(r-1)$-paths $a_1 - a_{2} -\dots- a_{r} \in A_1 \times A_{2}\times \dots \times A_{r}$ such that $1\le |T_{1,r}(G)[a_1, a_{r}]| < \lambda(r-1)$. In other words, $L_{1,2,...r}$ stores the $(r-1)$-paths for which $a_1 - a_{2} -\dots- a_{r}$ is a mildly sparse path between $a_1$ and $a_{r}$. Symmetrically, we construct a second list $L_{2,...r,r+1}$ containing all $(r-1)$-paths $a_{2} - a_{3} -\dots- a_{r}-a_{r+1} \in A_{2} \times A_{3}\times \dots \times A_{r}\times A_{r+1}$ such that $1\le |T_{2,r+1}(G)[a_{2}, a_{r+1}]| < \lambda(r-1)$.
These lists can contain at most $\lambda(r-1)n^2$ paths, $\lambda(r-1)$ paths for each index in $T_{1,r}$ or $T_{2,r+1}$.
\paragraph{Checking for Non Mildly Sparse Paths:}
In addition, we also want to make sure there are no other non mildly sparse paths between $a_1$ to $a_{r}$ or between $a_{2}$ and $a_{r+1}$. Let $G'$ be a $G^{T_{i,i+h}(G)}$ of the induced subgraph $G[A_1\cup A_{2}\cup \dots \cup A_{r}]$, $1\le i <i+h\le r$. If $\texttt{exist-path}(a_1,a_{r}) = true$ in the graph $G'$, then there exists a non mildly sparse path between $a_1$ and $a_{r}$ according to the definition of the auxiliary graph.
\begin{lemma} \label{R_lemma_for_mild_sparsity}
    If there exists a non mildly sparse path between $a_1$ and $a_{r}$, then there exists an auxiliary graph $G'=G^{T_{i,i+h}(G[A_1\cup A_{2}\cup \dots \cup A_{r}])}$ of the induced subgraph $G[A_1\cup A_{2}\cup \dots \cup A_{r}]$, satisfying $1\le i <i+h\le r$, such that $\texttt{exist-path}(a_1,a_{r}) = true$ in the graph $G'$
\end{lemma}
\begin{proof}
    for a non mildly sparse path $P$, call $\pi(P)$ the smallest number such that there exists some mildly sparse sub-path of $P$, $a_j-a_{j+1}-\dots-a_{j+\pi(P)}$, which satisfies either that the number of mildly sparse paths between $a_j$ and $a_{j+\pi(P)}$ is at least $\lambda(\pi(P))$, or that there is a non mildly sparse path between $a_j$ and $a_{j+\pi(P)}$. The existence of such a sub-path follows from taking the minimal non mildly sparse sub-path $a_x-\dots -a_{y}$, so either $a_x-\dots -a_{y-1}$ or $a_{x+1}-\dots -a_{y}$ satisfy the above condition.
    We take the non mildly sparse path $P=a_1-a_{2}-\dots -a_{r}$ from $a_1$ to $a_{r}$ with a minimal value of $\pi(P)$. Take the sub-path of $P$, $a_j-a_{j+1}-\dots-a_{j+\pi(P)}$ which holds the condition for $\pi$. There cannot be a non mildly sparse path $a_j-a_{j+1}'-\dots-a_{j+\pi(P)-1}'-a_{j+\pi(P)}$ between $a_j$ and $a_{j+\pi(P)}$. If there was, then taking the path $P'=a_1-\dots-a_j-a_{j+1}'-\dots -a_{j+\pi(P)}-\dots -a_{r}$, $\pi(P')<\pi(P)$, contradicting how we chose $P$. So there must be at least $\lambda(\pi(P))$ mildly sparse paths between $a_j$ and $a_{j+\pi(P)}$. So taking the auxiliary graph $G'=G^{T_{j,j+\pi(P)}(G[A_1\cup \dots \cup A_{r}])}$, $(a_j,a_{j+\pi(P)})$ is an edge in $G'$. So assuming the \texttt{exist-path} function works for graphs with $r$ or less layers, $\texttt{exist-path}(a_1,a_{r})$ in $G'$ returns true, as needed. 
\end{proof}
So in order to check whether there is a non mildly sparse path between $a_1$ and $a_{r}$, we check $\texttt{exist-path}(a_1,a_{r})$ for each auxiliary graph $G^{T_{i,i+h}(G[A_1\cup \dots \cup A_r])}$ (this takes $O(1)$ time by the induction hypothesis). The data structure for this auxiliary graph was created when creating $\mathcal{D}(G[A_1\cup \dots \cup A_r])$. For each path in $L_{1,\dots ,r}$, we check if there is a non mildly sparse path beginning and ending in the same vertices. If there is, we remove the path from $L_{1,\dots ,r}$. Do the same for $L_{2,\dots, r+1}$.

 We then sort both lists lexicographically by the vertex tuple $(a_{2}, a_{3},...,a_{r})$, treating $a_{2}$ as the primary key and $a_{3}$ as the secondary key, and so on. This sorting phase is executed efficiently in $O(n^2)$ time using radix sort, because the tuple size is constant.

After sorting, we perform a merge-like algorithm over $L_{1,\dots ,r}$ and $L_{2,\dots, r+1}$ using two pointers initialized at the beginning of each list. At each step, we compare the current $(a_{2}, a_{3},...,a_{r})$ tuples of both lists lexicographically. If the pair in $L_{1,\dots ,r}$ is smaller than the pair in $L_{2,\dots, r+1}$, we advance the pointer of $L_{1,\dots ,r}$; if it is larger, we advance the pointer of $L_{2,\dots, r+1}$. When the two pairs match on a specific vertex pair $(a_{2}, a_{3},...,a_{r})$, we save the blocks of entries sharing this pair in both lists, by moving the pointer forward in each list and saving the paths until we reach a path with a different tuple from $A_{2}\times \dots \times A_{r}$. We then define the sets $S_1 = \{a_1 \mid (a_{1}-a_{2}-\dots - a_{r}) \in L_{1,\dots ,r}\}$ and $S_{r+1} = \{a_{r+1} \mid (a_{2}-\dots -a_{r}-a_{r+1}) \in L_{2, \dots ,r+1}\}$. For every pair $(a_1, a_{r+1}) \in S_{1} \times S_{r+1}$, we insert the $r$-path $a_1 - a_{2} -\dots - a_{r} - a_{r+1}$ into the table entry $T_{1,r+1}(G)[a_1, a_{r+1}]$. After processing the matching blocks, we advance both pointers past these entries and resume the merge.
\begin{lemma}
Every colorful mildly sparse $r$-path in $G$ is inserted into $T_{1,r+1}(G)$.
\end{lemma}

\begin{proof}
Let $P = a_1 - a_{2} - \dots - a_{r+1}$ be an arbitrary colorful mildly sparse $r$-path in $G[A_1\cup \dots \cup A_{r+1}]$. the sub-path $a_1 - a_2 -\dots -a_{r}$ must be a mildly sparse $r-1$-path connecting $a_1$ and $a_{r}$ in $G$, and the pair $(a_1,a_{r})$ must satisfy $|T_{1,r}(G)[a_1, a_{r}]| < \lambda(r-1)$, and it also must satisfy that there is no non mildly sparse path between the pair. This guarantees that the $r$-path is included in $L_{1,\dots,r}$. By symmetric reasoning, the sub-path $a_2 - \dots - a_{r+1}$ is in $L_{2,\dots, r+1}$.

During the merge phase, the two pointers must simultaneously arrive at the block corresponding to the shared vertex tuple $(a_2,\dots ,a_{r})$. Consequently, $a_1$ will be added to $S_1$ and $a_{r+1}$ will be added to $S_{r+1}$ for the tuple $(a_2,\dots ,a_{r})$. The algorithm then inserts all paths $a_1-a_2\dots -a_r-a_{r+1}\in S_i \times \{a_2\}\times \dots \times \{a_{r}\}\times S_{i+m}$ into $T_{1,r+1}(G)$, guaranteeing that the path $P$ is inserted into $T_{1,r+1}(G)[a_1, a_{r+1}]$.
\end{proof}
The proof of the complexity bound is the similar to Algorithm~\ref{alg:compute_tablesM}, because we know both $L$ lists are bounded by $O(n^2)$. To create the lists $L$, sort them, and move pointers across them takes $O(n^2)$. The part of inserting the mildly sparse paths into $T_{1,r+1}(G)$ takes $O(|T_{1,r+1}(G)|)$.
\end{proof}
Afterwards, for each pair $(a_1,a_{r+1})\in A_1\times A_{r+1}$, if $T_{1,r+1}(G)[a_1,a_{r+1}]<\lambda(r)$, and there is a non mildly sparse path between $a_1$ and $a_{r+1}$, we put the paths of $T_{1,r+1}[a_1,a_{r+1}]$ in $R_{1,r+1}(G)[a_1,a_{r+1}]$.\\
According to Lemma~\ref{R_lemma_for_mild_sparsity}, we can check if there is a non mildly sparse path between $a_1$ and $a_{r+1}$ by checking $\texttt{exist-path}(a_1,a_{r+1})$ for all graphs $G^{T_{i,i+h}(G)}$, $1\le i<i+h\le r+1$ (not including the option $i=1$, $h=r$). Note that the size of $R_{1,r+1}(G)$ is bounded by $\lambda(r)n^2$
\subsection{Queries on the Data Structure}
We now describe how to conduct the queries $\texttt{exist-path}$ and $all-paths$.\\
$\texttt{exist-path}(a_1,a_{k+1})$ will check $\texttt{exist-path}(a_1,a_{k+1})$ in each auxiliary graph $G^{T_{i,j}(G)}$, as well as checking if $|T_{1,k+1}(G)[a_1,a_{k+1}]|\ge 1$. If any of those are true then we return true, and otherwise false. 
\begin{lemma}
    Assuming the algorithm works for graphs with $k$ layers or less, the query $\texttt{exist-path}$ returns the correct answer.
\end{lemma}
\begin{proof}
This process will return to the correct answer, because each path of type $2j<2(k-2)$ is represented by a path in some auxiliary graph $G^{T_{i,i+j+2}(G)}$.\\
Additionally, the existence of a path of type $2j+1$ implies that there is a path between $a_1$ and $a_{k+1}$ which is not mildly sparse, so according to Lemma~\ref{R_lemma_for_mild_sparsity}, there exist indices $i,h$ such that $\texttt{exist-path}(a_1,a_{k+1})$ is true in the graph $G^{T_{i,i+h}(G)}$, so our algorithm also returns true.
\end{proof}
$\texttt{all-paths}(a_1,a_{k+1})$ will conduct $\texttt{all-paths}(a_1,a_{k+1})$ in each auxiliary graph $G^{T_{i,j}(G)}$. For each path $a_1-\dots a_i-a_j-\dots-a_{k+1}$ in $G^{T_{i,j}(G)}$, we list every path $a_1-\dots a_i-a'_{i+1}-\dots a_{j-1}'-a_j-\dots-a_{k+1}$ such that $a_i-a'_{i+1}-\dots a_{j-1}'-a_j\in T_{i,j}(G)[a_i,a_j]$. We keep the listed paths in a hash table to make sure they are not listed twice. We also list all paths in $T_{1,k+1}(G)[a_1,a_{k+1}]$.\\
Additionally, for each of the paths listed in the algorithm above $a_1-a_2-\dots -a_{k+1}$ we check $R_{i,j}[a_i,a_j]$ for each $i,j$ to check whether there exist mildly sparse paths (at most $\lambda(j-i)$) between $a_i$ and $a_j$, for example $a_i-a_{i+1}'-\dots -a_{j-1}'-a_j$, such that $a_1-\dots -a_i-a_{i+1}'-\dots a_j-\dots -a_{k+1}$ is also a path in $G$. We list all these paths as well.
\\
\begin{lemma}
    Every path between $a_1$ and $a_{k+1}$ is listed during an $\texttt{all-paths}(a_1,a_{k+1})$ query.
\end{lemma}
\begin{proof}
A path $P=a_1-a_2-\dots - a_{k+1}$ of type $2j$ is listed when we list paths from $G^{T_{i,i+j+2}(G)}$ for some value $1\le i\le k-1-j$. This is because the path has some minimal non-mildly sparse sub-path of length $j+3$, and for some $i$, the amount of mildly sparse paths between $a_i$ and $a_{i+j+2}$ is at least $\lambda(j+2)$. So the path $a_1-\dots -a_i -a_{i+j+2}-\dots -a_{k+1}$ is in the auxiliary graph $G^{T_{i,i+j+2}(G)}$, so it is listed during an $\texttt{all-paths}$ query ($G^{T_{i,i+j+2}(G)}$ has less than $k+1$ layers so we assume inductively that the query is successful). In the algorithm we then list $a_1-\dots -a_i -P'-a_{i+j+2}-\dots -a_{k+1}$ for every $P'\in T_{i,i+j+2}(G)[a_i,a_{i+j+2}]$, and so we will list $P$ as well.\\
A path $P=a_1-a_2-\dots - a_{k+1}$ of type $2j+1$, has vertices $i,i+j+2$ such that the number of mildly sparse path between $a_i$ and $a_{i+j+2}$ is less than $\lambda(j+2)$, but there is a non-mildly sparse path between $a_i$ and $a_{i+j+2}$. According to Lemma~\ref{R_lemma_for_mild_sparsity}, There is an auxiliary graph of $G[A_i\cup A_{i+1}\cup \dots \cup A_{i+j+2}]$, $G^{T_{l,h}(G[A_i\cup A_{i+1}\cup \dots \cup A_{i+j+2}])}$ For some $i\le l <h\le i+j+2$ (either $i\ne l$ or $h\ne i+j+2$), such that there is a path between $a_i$ and $a_{i+j+2}$ in the auxiliary graph. So $a_i-a_{i+1}'-\dots -a_l'-a_h'-\dots -a_{i+j+2}$ is a path in the auxiliary graph $G^{T_{l,h}(G[A_i\cup A_{i+1}\cup \dots \cup A_{i+j+2}])}$, which means $T_{l,h}[a_l,a_h]\ge \lambda(h-l)$. Taking the auxiliary graph $G^{T_{l,h}(G)}$ of $G$, $(a_l,a_h)\in E(G^{T_{l,h}(G)})$, and so the path $a_1-\dots -a_i-\dots a_l'-a_h'-\dots -a_{i+j+2}-\dots -a_{k+1}$ is in the auxiliary graph $G^{T_{l,h}(G)}$. So when listing paths in $G^{T_{l,h}(G)}$, we get $a_1-\dots -a_i-\dots a_l'-a_h'-\dots -a_{i+j+2}-\dots -a_{k+1}$, and then taking paths in $T_{l,h}(G)[a_l,a_h]$, we list a path $P'$ from $a_1$ to $a_{k+1}$ which has $a_i$ and $a_{i+j+2}$ as internal vertices. For each of these listed paths $b_1-\dots -b_{k+1}$, as described in the algorithm, we check whether $R_{x,y}[b_x,b_y]$ has an alternative sub-path for every possible $x,y$. In this case, checking $R_{i,i+j+2}[a_i,a_{i+j+2}]$ for the path $P'$, we get $a_i-a_{i+1}\dots -a_{i+j+2}$. Putting this between $a_i$ and $a_j$ in $P'$ we list the path $P$, as needed.
\end{proof}
We've finished proving the algorithms correctness. We will now bound the time complexity of $\texttt{exist-path}$, and $\texttt{all-paths}$, and in the next section we will bound the construction time of the algorithm.
$\texttt{exist-path}$ takes $O(1)$ time because we call a constant amount of $O(1)$ function. 
\begin{lemma}
$\texttt{all-paths}$ queries take $O(P)$ time where $P$ is the number of paths listed. We assume that for $k$ layers or less, the algorithm takes $O(P)$ time. 
\end{lemma}
\begin{remark}
$k$ is a fixed constant, so we ignore complexity functions based on $k$.
\end{remark}
\begin{proof}
The first part of the algorithm conducts $\texttt{all-paths}(a_1,a_{k+1})$ in each auxiliary graph $G^{T_{i,j}(G)}$, then list every path $a_1-\dots a_i-a'_{i+1}-\dots a_{j-1}'-a_j-\dots-a_{k+1}$ such that $a_i-a'_{i+1}-\dots a_{j-1}'-a_j\in T_{i,j}(G)[a_i,a_j]$.  $\texttt{all-paths}(a_1,a_{k+1})$ in each auxiliary graph $G^{T_{i,j}(G)}$ takes time relative to the amount of paths listed (according to the induction hypothesis). Then for each of those paths, listing every path $a_1-\dots a_i-a'_{i+1}-\dots a_{j-1}'-a_j-\dots-a_{k+1}$ such that $a_i-a'_{i+1}-\dots a_{j-1}'-a_j\in T_{i,j}(G)[a_i,a_j]$ is conducted in $O(P)$ time where $P$ is the amount of paths listed. Additionally, listing paths from the table $T_{1,k+1}(G)$ also takes time relative to the amount of paths in $T_{1,k+1}(G)[a_1,a_{k+1}]$. \\
The last part is for each of those listed paths we check $R_{i,j}[a_i,a_j]$ for each $1\le i<j\le k+1$, and list additional paths if needed. $k$ is constant, and in each entry $R_{i,j}[a_i,a_j]$ the number of paths is at most $\lambda(j-i)$, so this part takes $O(1)$ time for each path listed.
\end{proof}
\section{Complexity Analysis and Reduction to the Boundary Layer Property}\label{sec:complexity_boundary}
In the first part of this section we will be bounding the size of tables $T_{i,i+m}(G)$ of the original graph $G$. Afterwards we will show that using similar techniques to Section~\ref{sec:five_layers} but with mildly sparse paths lead to a bound of $O(n^2+t_{2r}(G))$ on the 5 layer data structure for all constants $r$. During this we will introduce another method in order to account for case~\ref{exception}. \\
In the third section we define an important concept called the boundary layer property for a given constant $s$, and show that if it is true for $s$, then our algorithm runs in $O(n^2+t_{2k}(G))$, given a $k+1$-layer graph such that $k\le 2s+2$.

\subsection{Bounding Table Sizes}
In this subsection we will bound the construction of each $T$ table in $G$ by $O(n^2+t_{2k}(G))$ for a constant $k$. We will additionally prove statements about mapping paths of an auxiliary graph to the original graph, in order to later prove bounds on auxiliary graphs as well.

\begin{theorem}\label{thm:_size_of_T_tables}
    For any target cycle length $2r \ge 2m$, the size of the table $T_{i,i+m}(G)$ is bounded by $|T_{i,i+m}(G)| = O(n^2 + t'_{2r}(G[A_i\cup \dots \cup A_{i+m}]))$ (and also symmetrically by $O(n^2 + t''_{2r}(G[A_i\cup \dots \cup A_{i+m}]))$).
\end{theorem}
\begin{proof}
    Assume that $|T_{i,i+m}| > d n^2$ for $d=100r\lambda^2(m)$ (otherwise the theorem is trivial). There must exist at least one vertex $a_i \in A_i$ which is a part of more than $dn$ mildly sparse paths. Consider the subgraph $G_{a_i}$ created by the vertices, and the edges of all the mildly sparse paths between $a_i$ and $A_{i+m}$. Each edge in the set $E[A_{i+m-1}, A_{i+m}]$ belongs to at least one and at most $\lambda(m)$ mildly sparse paths between $a_i$ and $a_{i+m}$ according to Theorem~\ref{bounding_intersections}, because multiple paths that use the same edge are intersecting in an internal vertex of the path ($\lambda(m)>\sum_{j=1}^{m-1} \lambda(j)\lambda(m-j)$). This establishes that the number of edges in this set is proportional to the number of mildly sparse paths up to the constant factor $\lambda(m)$. We now prove that there are sufficiently many $2r$-cycles in the graph $G_{a_i}$ that use the vertex $a_i$, by describing an algorithm that finds such cycles iteratively. The process, similar to Theorem~\ref{thm:3-sparse_paths} will include iteratively conducting a pruning process to remove vertices with low degree, and afterwards finding a cycle. The process stops when the set $E(G_{a_i}[A_{i+m-1},A_{i+m}])$ has less than $100r\lambda(m)n$ edges, and we will show that $t_{2r}(G_{a_i})=\Omega(E(G_{a_i})-2dn)$.

    \paragraph{Pruning Process:}
    We create the subgraph $G_{a_i}' \subseteq G_{a_i}$ by iteratively pruning all vertices in $A_{i+m-1} \cup A_{i+m}$ whose degree in $G_{a_i}[A_{i+m-1}\cup A_{i+m}]$ are below $\frac{d}{10\lambda(m)}=10r\lambda(m)$.

    \paragraph{Finding Cycles in $G_{a_i}'$:}
    As noted earlier, Each edge in the set $E[A_{i+m-1}, A_{i+m}]$ belongs to at least one and at most $\lambda(m)$ mildly sparse paths. For each of these edges, we define $\kappa((a_{i+m-1},a_{i+m}))=P$, where $P$ is one of the mildly sparse paths ending with the edge $(a_{i+m-1},a_{i+m})$. We begin a cycle $C$ at $a_i$ and go to $a_{i+m}$ using the path $P$. 
    We now have some mildly sparse path $P=a_i-a_{i+1}-\dots -a_{i+m}$, and we now create a cycle similar to in Theorem~\ref{thm:3-sparse_paths}. Because the degree of each vertex in $G_{a_i}'[A_{i+m-1}\cup A_{i+m}]$ is at least $\frac{d}{10\lambda(m)}=10r\lambda(m)>10r$, we can greedily create a $2r-2m$ path in $G_{a_i}'[A_{i+m-1}\cup A_{i+m}]$ that is simple and doesn't intersect with $a_{i+m-1}$, call this path $P'$. Denote the last vertex in this path as $a_{i+m}'$, we want to take a path back to $a_i$ without intersecting $P$ or $P'$. 
    \\
    We now show that there exists such a path, by bounding the amount of neighbors of $a_{i+m}'$, $a_{i+m-1}\in N(a_{i+m})$ such that $\kappa(a_{i+m-1},a_{i+m}')$ intersects with $P$ or $P'$.
    \begin{lemma}
        The amount of edges $(a_{i+m-1}, a_{i+m}')$ such that $a_{i+m-1}\in A_{i+m-1}$ and the path $\kappa(a_{i+m-1},a_{i+m}')$ intersects with $P$ is at most $\lambda(m)$.
    \end{lemma}
    \begin{proof}
        Assume the amount of such edges is $z\ge \lambda(m)+1$. Each of these edges is mapped to a different mildly sparse path, so there are at least $\lambda(m)+1$ such mildly sparse paths between $a_i$ and $a_{i+m}'$. Denote this set of paths $\mathcal{P}$. We can take a subset of $cm$ internally disjoint paths similarly to Lemma \ref{c_vertex-disjoint_paths}. Taking one of these paths, $Q_1$, there are at most $\sum_{j=1}^{m-1} \lambda(j)\lambda(m-j)-1$ other such paths which intersect with $Q_1$ according to \ref{bounding_intersections}. We discard those paths from $\mathcal{P}$ and then take another path $Q_2$. Doing this iteratively until $Q_{cm}$, the set $\{Q_1,\dots, Q_{cm}\}$ consists of internally disjoint paths from $\mathcal{P}$. $\lambda(m)+1>cm\sum_{j=1}^{m-1} \lambda(j)\lambda(m-j)$ which means the set $\mathcal{P}$ will not be empty while taking some $Q_i$.

        Each two of these paths can only intersect $P$ at different vertices. Since there are $cm>m$ paths, and only $m-1$ possible internal vertices in $P$, one of the paths doesn't intersect with $P$, in contradiction. 
    \end{proof}
    So $a_{i+m}'$ has at most $\lambda(m)$ neighbors $a_{i+m-1}'\in N(a_{i+m}')$ such that $\kappa(a_{i+m-1}',a_{i+m}')$ intersects with $P$. $a_{i+m}'$ also has at most $2r$ neighbors in $P'$. So in choosing the path back to $a_i$ we must avoid a set of $\lambda(m)+2r$ vertices. $\deg(a_{i+m}')\ge 10r\lambda(m)>\lambda(m)+2r$, so we can choose a neighbor of $a_{i+m}'$, $a_{i+m-1}'$, not in that set. taking the path $P \circ  P'\circ \kappa(a_{i+m-1}',a_{i+m}')$, we get a simple $2r$-cycle. The $2r$-cycle has $2$ edges in each layer transition other than the last.
    
    We remove the edges in the cycle that are also in $E(A_{i+m-1},A_{i+m})$, and continue the process (re-prune, find another cycle, etc...). We stop the process when the number of edges in $E(A_{i+m-1},A_{i+m})$ falls below $100r\lambda(m)n$.

   For a vertex $a_i\in A_i$ the number of cycles in $G$ which have $a_i$, and have $2$ edges in each layer transition other than the last is $\Omega(e(G_{a_i}[A_{i+m-1},A_{i+m}])-200r\lambda(m)n)$. This is because the amount of edges in $G_{a_i}[A_{i+m-1},A_{i+m}]$ that are either pruned, or stay at the end when we stop the process, is at most $100r\lambda(m)n+ \frac{dn}{10\lambda(m)}<200r\lambda(m)n$. All other edges are used in cycles we find in $G_{a_i}$. We showed earlier that $e(G_{a_i}[A_{i+m-1},A_{i+m}])$ is proportional up to a constant to the number of mildly sparse paths in $T_{i,i+m}(G)$ which pass through $a_i$.\\
   Summing this for all vertices $a_i\in A_i$ we get $t'_{2r}(G)=\Omega((\sum_{a_i\in A_i}e(G_{a_i}[A_{i+m-1},A_{i+m}]))-200r\lambda(m)n^2)=\Omega(|T_{i,i+m}(G)|-200r\lambda(m)n^2)$, which means $|T_{i,i+m}(G)|=O(n^2+t_{2r}(G))$ as needed.
\end{proof}

\subsection{Bounding the Data Structure Size for 5-Layer Graphs}
By using Theorem~\ref{thm:_size_of_T_tables} and Theorem~\ref{distinct-mildly-sparse-paths}, the 3-, 4-, and 5-layer data structures can be created the same as in previous sections, but using mildly sparse paths instead of sparse. The time complexity proofs remain basically the same using the 2 theorems. This is because for an auxiliary graph $H$ we have that each table size $T_{x,y}(H)$ is $O(n^2+t'_{2r}(H))$. We can then use Theorem~\ref{distinct-mildly-sparse-paths} to map these cycles back to cycles in $G$ using the same proof as in Section~\ref{sec:five_layers}, and whenever we need to map auxiliary edges to paths, we use Theorem~\ref{distinct-mildly-sparse-paths}.
So the algorithm running with a complexity of $O(n^2 + t_8(G))$. While most intermediate tables are bounded by $O(n^2 + t_{2r}(G))$ for all $r \ge 4$, we had an exception to this rule as describe in Paragraph~\ref{exception}. The table $T_{1,5}$ maintains a size bounded by $O(n^2 + t'_{2r}(G))$ via Theorem~\ref{thm:_size_of_T_tables}.

\subsubsection{Bounding the Auxiliary Table $T_{1,3}(H)$ for $H=G^{T_{1,3}({G^{T_{3,5}(G)}})}$}
The auxiliary graph $H=G^{T_{1,3}({G^{T_{3,5}(G)}})}$ has 3 vertex layers denoted $B_1, B_2, B_3$ (corresponding to $A_1, A_3, A_5$ in $G$), where the edge sets $E[B_1, B_2]$ and $E[B_2, B_3]$ represent at least $c$ $2$-mildly sparse paths. We now bound the size of the table $T_{1,3}(H)$. We use a pruning and cycle finding algorithm as before, but this time the cycle finding algorithm will be different in order to find a path of any size in a boundary layer transition which consists of auxiliary edges.

\begin{theorem}\label{difficult_case}
    The size of the table satisfies $|T_{1,3}(H)| = O(n^2 + t_{2r}(G))$ for a fixed constant $r \ge 4$.
\end{theorem}
\begin{proof}
    If $|T_{1,3}(H)| \le 1000rn^2$, the statement holds trivially. Otherwise, there must exist a vertex $b_1 \in B_1$ such that the graph $H_{b_1}$ is defined to include the vertex sets $\{b_1\}\cup(N(b_1)\cap B_2)\cup (N(N(b_1))\cap B_3)$ and the edges of all the $2$-paths in $T_{1,3}(H)$ that start at $b_1$. 

    We conduct a pruning process on the graph $H_{b_1}$, which removes all vertices in $H_{b_1}\setminus\{b_1\}$ which have a degree less than $50r$. This allows us to choose some initial edge from $b_1$ to some $b_2 \in B_2$, continue the path with an even-length walk within $E[B_2, B_3]$, and return to $b_1$. 
    
    Every edge $(b_2, b_3) \in E[B_2, B_3]$ maps to at least $2c$ alternative internal vertices within $A_4$. For any walk of length bounded by $2c$, we can choose a unique, non-intersecting vertex $a_4 \in A_4$ for each edge in $H$ to create a path in $G$. However, because each edge in the auxiliary graph represents a $2$-path, an even walk in $H_{b_1}$ translates only to paths in $G$ whose lengths are multiples of $4$, which only accounts for even values of $r$.

    To extend this proof to odd values of $r$, we examine the density of $E[B_2, B_3]$. This edge set contains at least $50rn$ edges. Because each edge $(b_2, b_3)$ maps to a choice of at least $2c$ vertices in $A_4$, the pigeonhole principle guarantees the existence of two distinct edges $(b_2', b_3')$ and $(b_2'', b_3'')$ that share a common internal vertex $a_4 \in A_4$. We analyze this case by examining two mutually exclusive cases:

    \begin{figure}[htbp]
        \centering
        \includegraphics[width=0.6\linewidth]{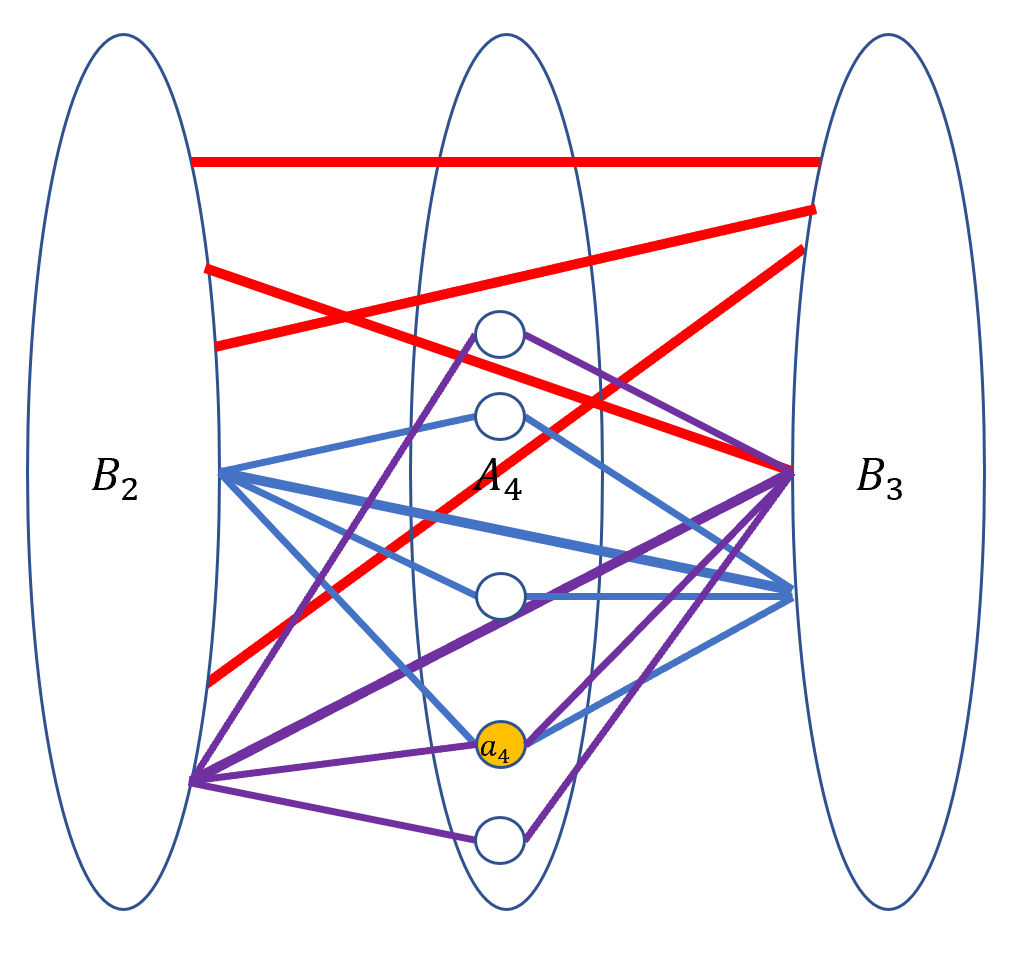}
        \caption{Parity extension utilizing a shared internal vertex $a_4$ between layers $B_2$ and $B_3$.}
        \label{pic:vertex_in_the_middle}
    \end{figure}

    \textbf{Case 1 (odd $r$, $b_2' \neq b_2''$):} The sequence $b_2' - a_4 - b_2''$ forms a $2$-path in the graph $G$. We construct a cycle by taking an edge in $H$ from $b_1$ to $b_2'$, going through the $2$-path $b_2'-a_4-b_2''$ to $b_2''$ in $G$, adding an even-length auxiliary walk from $b_2''$ to a different vertex $b_2'''$, and returning to $b_1$. By selecting distinct internal vertices within $A_4$ for the auxiliary edges and distinct vertices within $A_2$ for the edges from and to $b_1$, we obtain a $2r$-cycle for any odd constant $r > 3$. We list this cycle and delete one of its edges from $E[B_2, B_3]$ to prevent duplicate enumeration.

    \textbf{Case 2 (odd $r$, $b_2' = b_2''$):} Since the edges are distinct, it follows that $b_3' \neq b_3''$. We construct a cycle by traversing the edge $(b_1, b_2')$ to reach $b_3'$, then going across the $2$-path $b_3' - a_4 - b_3''$ in $G$, continuing along a simple path from $b_3''$ back to $B_2$, and returning to $b_1$. This yields a simple $2r$-cycle for any odd constant $r > 3$. As before, removing a single edge from $E[B_2, B_3]$ avoids double-counting (the only possible way we count a cycle with a path $b_2-a_4-b_3$ is if $(b_2,b_3)\in E[B_2, B_3]$).

    \begin{remark}
        Unlike previous pruning steps where multiple edges of the cycle are erased, here we remove only a single edge from the auxiliary edge set $E[B_2, B_3]$. This edge-removal strategy prevents listing the same cycle twice while only changing the auxiliary graph $H_{b_1}'$ (as opposed to the original graph $G$, which we cannot do).
    \end{remark}

    \begin{remark}
        As illustrated in Figure~\ref{pic:vertex_in_the_middle}, the distinct path entries represented by the auxiliary edges share a common vertex intersection at $a_4$.
    \end{remark}
    We can now conduct the pruning and cycle finding algorithm similarly to previous algorithms. Each time we prune the graph, find the cycle, prune again, etc... We stop when there are less than $1000rn$ edges in $H_{b_1}'$. All edges not erased during the pruning process (which erases at most $50rn$ edges), or remain in the final graph where the algorithm stops, are removed as part of a $2r-$cycle, which completes the proof.
\end{proof}
\begin{corollary}
The entire 5-layer data structure size is bounded by $O(n^2 + t_{2r}(G))$ for all $r \ge 4$.
\end{corollary}

\subsection{Reduction to the Boundary Layer Property}
For a $k+1$ layer graph $G$, during the algorithm to create $\mathcal{D}(G)$ using mildly sparse paths, we construct many tables for many auxiliary graphs (for each auxiliary graph we construct all tables $T_{i,j},R_{i,j}$, as well as auxiliary graphs based on the tables $T_{i,j}$). We now wish to bound the size of all the tables we've created by $O(n^2+t_{2k}(G))$. The amount of tables we create is constant, so we need to prove that the size of every table of every auxiliary graph is $O(n^2+t_{2k}(G))$.
\begin{definition}[Boundary Layer Property]\label{def:boundary_layer_prop}
We are given a $(k+1)$-layer graph $G=(V,E)$.\\
Let $1\le  s\le k' \le k$ be fixed integers. We say that the \textbf{Boundary Layer Property} holds for $s$ if, for every $(s+1)$-layer graph $J = (V, E)$, $V = A_1 \cup A_2 \cup \dots \cup A_{s+1}$, and for every auxiliary graph $J'$ of $J$, or $J$ itself, with $d+1$ layers, the following holds:

If the auxiliary graph $H=J^{T_{1,d+1}(J')}$ of the graph $J'$, and a subgraph $H'$ of $H$ satisfies that the minimal degree of $H'$ is at least $50\lambda(k)$, and the graph has at least $50\lambda(k)n$ edges, then there exists a simple path $P$ in $J$ of length $2k'-s$ with one endpoints residing in $A_1\cap H'$ and the other in $A_{s+1}\cap H'$.
\begin{remark}
The auxiliary graph $J'$ could also have only two layers, in which case we denote $H=J^{T_{1,2}(J')}=J'$.
\end{remark}
\end{definition}
For $s=1$, this means that For a 2-layer graph $J=(V,E)$, $V=A_1\cup A_2$ and a constant $k$, if there are more than $50\lambda(k)n$ edges, then there is a simple path $P \subseteq J$ of length $2k'-s$ with one endpoints residing in $A_1$ and the other in $A_{2}$. This is true as proven in Theorem~\ref{size_of_2-path_table}.\\
\begin{theorem}\label{main_thm}
For a fixed integer $s$, if the boundary layer property is true for all naturals up to and including $s$, then we can create an $O(n^2+t_{2k}(G))$ algorithm for the construction of $\mathcal{D}(G)$ for a graph $G$ with $k+1$ layers, for $1\le k \le 2s+2$.
\end{theorem}
The idea of the reduction is to conduct a pruning and cycle finding algorithm similar to before, but this time use the assumption of the boundary layer property for creating the sub-path in the last layer.
\begin{proof}
    Take any $k\le 2s+2$, and a $k+1$-layer graph $G$. In the process of the algorithm we create many auxiliary graphs, and tables for those graphs. We need to prove that the size of each of those tables is $O(n^2+t_{2k}(G))$. Let $H=B_1\cup \dots \cup B_{l+1}$ be some $l$-layer auxiliary graph of $G[A_x,A_{x+1},\dots A_y]$ for $1\le x<y \le k+1$, we prove $|T_{i,j}(H)| = O(n^2+t_{2k}(G))$ for any $1\le i<j\le l+1$.
    \begin{remark}
    In the algorithm we create auxiliary graphs for induced subgraphs as well. So we could have a process where we induce the graph $G$ to a sub-range of layers in $G$, take an auxiliary graph $G'$, induce $G'$, and so on. However, the resulting graph, in this case $H$, is an auxiliary graph of $G[A_x,A_{x+1},\dots A_y]$ for $1\le x<y \le k+1$.
    \end{remark}
     Consider the two layer transitions $E[B_i,B_{i+1}]$ and $E[B_{j-1},B_j]$. We know that either the distance between $B_i$ and $B_{i+1}$ in the original graph $G$ is at most $s+1$, or the distance between $B_{j-1}$ and $B_j$ is at most $s+1$, otherwise the distance between $B_i$ and $B_j$ is larger than $k$. We split the proof into two cases.
    \paragraph{Case 1:}
    In this case, either the distance between $B_i$ and $B_{i+1}$ in the original graph $G$ is at most $s$, or the distance between $B_{j-1}$ and $B_j$ is at most $s$. The other case is both distances being $s+1$, we take care of that later. Assume that the distance between $B_{j-1}$ and $B_j$ is at most $s$, call the distance $m$ (The case in which the distance between $B_{i}$ and $B_{i+1}$ is at most $s$ is symmetrical). Additionally, call the distance between $B_i$ and $B_{j-1}$ in $G$, $d$. If $|T_{i,j}(H)|< 1000\lambda^2(k)n^2$ then the proof is complete. Otherwise, there exists a vertex $b_i\in B_i$ such that the number of paths in $T_{i,j}(H)$ which belong to $b_i$ is at least $1000\lambda^2(k)n$. Take the graph $H_{b_i}$ which includes all these paths. Similarly to Theorem~\ref{thm:_size_of_T_tables}, all edges $e$ in $E[B_{j-1},B_j]$ have between 1 and $\lambda(k)$ paths starting at $b_i$ and ending with $e$. This means there are at least $1000\lambda(k)n$ edges in $E[B_{j-1},B_j]$. We conduct a similar pruning and cycle finding algorithm as described in Theorem~\ref{thm:_size_of_T_tables}. This time, the algorithm will return a cycle $C$ with edges in both $H$ and $G$, and then we will map $C$ to a $2k$-cycle in $G$. We will stop the process when the amount of edges in $E[B_{j-1},B_j]$ falls below $1000\lambda(k)n$. 
    
    \paragraph{Pruning Process}
    Prune the graph $H_{b_i}$ such that each vertex in $H_{b_i}[B_{j-1}\cup B_j]$ has a degree of at least $50\lambda(k)$, and call the new graph $H_{b_i}'$. 
    \paragraph{Cycle Finding via Boundary Layer Property:}
    From here, $H_{b_i}'[B_{j-1},B_j]$ is a subgraph of $H[B_{j-1},B_j]$. Going back to the original graph $G$, name $A_r=B_{j-1}$, and $A_{r+m}=B_j$ for some $m\le s$ and some $r$. This means $H[B_{j-1},B_j]$ is an auxiliary graph of $G[A_r\cup A_{r+1}\cup \dots \cup A_{r+m}]$. Using the boundary layer property for constants up to $s$, we know there is a path in $G$ between a vertex $b_{j-1}\in B_{j-1}\cap H'_{b_i}$ and a vertex $b_{j}\in B_{j}\cap H'_{b_i}$ of size $2(k-d)-m$, call it $P^{mid}$. The minimal degree is $50\lambda(k)$ which is larger than $50\lambda(k-d)$, so we can take $k'=k-d$ for the Boundary Layer Property, ensuring there is a path of length $2(k-d)-m$. From the vertex $b_{j-1}$ we can take one of the mildly sparse paths in $H$ from $b_i$ to $b_{j-1}$ (there exists one, otherwise $b_{j-1}$ wouldn't be in $H_{b_i}$), and call that path $P^{start}$.\\
    $b_j$ has a degree of at least $50\lambda(k)$. $P^{mid}$ has at most $2k$ vertices in $B_{j-1}$, so subtracting those from the neighbors of $b_j$ we still have at least $49\lambda(k)$ vertices. Each one of its edges $e$ is a part of the mildly sparse path $\kappa(e)$, as defined in Theorem~\ref{thm:_size_of_T_tables}. In Lemma~\ref{c_vertex-disjoint_paths}, we proved that for $\lambda(j-i)<\lambda(k)$ mildly sparse paths between $b_i$ and $b_j$, there are $c(j-i)$ pairwise internally disjoint mildly sparse paths $P_1,...P_{c(j-i)}$. There are only $(j-i)$ vertices where the path $P^{start}$ can intersect with one of the paths $P_z$, so according to the pigeonhole principle, there exists an index $z$ such that $P_z$ doesn't intersect with $P^{start}$. Taking $P^{start}\circ P^{mid}\circ P_z$, we get a cycle. The cycle uses partially edges in $H$ (at $P^{start}$ and $P_z$), and partially edges in the graph $G[A_r\cup \dots \cup A_{r+m}]$ (in the path $P^{mid}$). \\
    We now wish to map this cycle to a $2k$-cycle in $G$.
    \paragraph{Mapping The Cycle to a $2k-$Cycle in $G$:}
    \begin{lemma} \label{breaking_paths}
    Given a layered graph $G=(V,E), V=A_1\cup A_2\cup \dots \cup A_{k+1}$, an auxiliary graph of $G$, $H$, such that $V(H)=B_1\cup\dots \cup B_{d+1}$ and two colorful paths in $H$ between layer $B_i=A_{i'}$ and $B_{j-1}=A_{j'}$, we can map them into two colorful paths in $G$ between those layers in the original graph $G$.
    \end {lemma}
    \begin{proof}
    We prove this by induction on the number of table reductions needed to get from $G$ to $H$. For 1 table reduction, There is only one layer transition in $H$ whose edges are not edges in the original graph. Those edges can be mapped to colorful paths in $G$ using Theorem~\ref{distinct-mildly-sparse-paths}, which states that for any distinct pair of edges $(b^1_x, b^1_{y})$ and $(b^2_x, b^2_{y})$ in $H=G^{T_{x,y}(G)}$ for some $1\le x<y\le k+1$, there exist mutually internally vertex-disjoint mildly sparse paths $P_1$ and $P_2$ connecting their respective endpoints in $G$. We map the edges to those internally vertex-disjoint paths. Assuming we can do this mapping for $N$ table reductions, for $N+1$ table reductions, Let $H'$ be the auxiliary graph one table reduction before $H$. We can map the paths in $H$ to paths in $H'$ in the same way using Theorem~\ref{distinct-mildly-sparse-paths}, and then from $H'$ to $G$ using the mapping for $N$ table reductions from the induction hypothesis.
    \end{proof}
    \begin{lemma} \label{breaking_edge}
    Given a layered graph $G$, an auxiliary edge $(b_{j-1},b_j)\in B_{j-1}\times B_{j}$ in an auxiliary graph $H$, and $2k$ vertices in $G[A_{r+1}\cup \dots \cup A_{r+m}]$ ($B_{j-1}=A_r$ and $B_j=A_{r+m}$), the edge can be mapped into a colorful path in $G$ which doesn't intersect any of the given vertices.
    \end{lemma} 
    \begin{proof}
    We prove this by induction on the number of table reductions needed to get from $G$ to $H$. For 1 table reduction, we know according to Theorem~\ref{c_vertex-disjoint_paths} that there are at least $c=10k$ internally disjoint paths between $b_{j-1}$ and $b_j$ in $G$. Each of the $2k$ vertices is a part of at most one paths, which leaves a path that doesn't intersect with any of the vertices. We now assume the mapping can be done for $N$ table reductions. For $N+1$ table reductions, Let $H'$ be the auxiliary graph one table reduction before $H$. We can map the edge $(b_{j-1},b_j)$ into a colorful path in $H'$, $P$, which doesn't intersect with the set of $2k$ vertices, in the same way we did the mapping for 1 table reduction, using Theorem~\ref{c_vertex-disjoint_paths}. Now on each edge in the path $P$, we can conduct the mapping to a path in $G$ for $N$ or less table reductions based on the induction hypothesis (each time we have to avoid at most $2k$ vertices).
    \end{proof}
    Using the Lemma~\ref{breaking_paths}, we can map the path $P^{start}$, and $P_z$ without the first edge $e$ (the edge in $E[B_{j-1},B_j]$) into colorful paths in the original graph. Using Lemma~\ref{breaking_edge}, we map the edge $e$ in $E[B_{j-1},B_j]$ to a colorful path in $G$, that has no intersections with the vertices in the path $P^{mid}$. Doing this we receive a cycle of size $2k$ in $G$, as needed.
    \paragraph{Completing Case 1:}
    We remove the first edge of $P_z$, which belongs to $E[B_{j-1},B_j]$, from the graph. We then prune the remaining graph and identify another cycle, repeating this procedure until reaching a subgraph with fewer than $1000\lambda(k)n$ edges. 

 When mapped back to $G$, this edge, denoted $(u,v)$ corresponds to some path. Because $P^{start}$ and $P^{mid}$ have fixed, constant lengths, the path corresponding to the first edge of $P_z$ always occurs at fixed indices $q$ to $q'$ in any cycle starting at $b_i$.

Once the edge $(u,v)$ is removed from $E[B_{j-1},B_j]$, no subsequent iteration can output a cycle that traverses from $u$ to $v$ at position $q$. Since an undirected cycle containing $b_i$ can only be traversed in two directions starting from and ending at $b_i$, there are 2 possible paths at indices $[q,q']$ for the cycle. Consequently, any fixed cycle $C$ can be discovered and output by this algorithm at most twice. 
\begin{remark}
The reason that removing the edge $(u,v)$ doesn't alone ensure we never pick the same cycle, is because there could be a colorful path between $u$ and $v$ in $G$ that uses edges from paths mapped to by other edges in $H$.
\end{remark}
    Every edge in $E[B_{j-1},B_j]$ in the graph $H_{b_i}$ which is not either in the remaining $1000\lambda(k)n$ at the end, or the at most $50\lambda(k)n$ edges that are pruned, is part of a $2k$-cycle found during our process. So there are $\Omega(e(H_{b_i})-1050\lambda(k)n)$ $2k$-cycles in $G$ which pass through $b_i$. Summing this for all $b_i\in B_i$, we get $\Omega(|T_{i,j}(H)|-1050\lambda(k)n^2)$ $2k$-cycles in $G$, as needed.\\
    \paragraph{Case 2:}
    In this case, the distance between $B_i$ and $B_{i+1}$ is $s+1$, and the distance between $B_{j-1}$ and $B_{j}$ is $s+1$. The only way this happens is if $B_i=A_1$, $B_{i+1}=B_{j-1}=A_{s+2}$, $B_j=A_{2s+3}$. Since $H$ has 3 layers, it is easy to prove (as we've done before) that $T_{1,3}(H)=O(n^2+t_4(H))$. From here, Each of these 4-cycles can be viewed as 2 paths $P^{start}$ (first 2 edges) and $P^{end}$ (last 2 edges). These are colorful auxiliary paths, so using Lemma~\ref{breaking_paths}, We can map them into colorful paths in $G$, creating a $2k$-cycle in $G$.
    \\
    \paragraph{Final Step:}
    We've proven that for each auxiliary graph $H$ and table $T_{i,j}(H)$, $|T_{i,j}(H)|=O(n^2+t_{2k}(G))$. So the size of the data structure is also $O(n^2+t_{2k}(G))$. We've also shown that the algorithm to build $T_{i,j}(H)$ takes $O(n^2+|T_{i,j}(H)|)\le O(n^2+t_{2k}(G))$, so creating the entire data structure takes $O(n^2+t_{2k}(G))$.
\end{proof}
\subsection{Range cycle listing result}\label{subsec:range_result}
Before we go on to address the boundary layer property, we show the result needed for range cycle listing. That is, for any constant $k$ and number $3\le x \le \frac{4k}{3}$, the data structure of an $l$-layered graph can be constructed in $O(n^2+t_{\lceil\frac{x}{2}\rceil}(G)+t_{\lceil\frac{x}{2}\rceil+2}(G)+\dots+t_{2k}(G))$, for $3\le l\le k+1$. This can also be phrased as: for any constant $k$ and even number $4\le x \le \frac{4k}{3}+1$, the data structure of an $l$-layered graph can be constructed in $O(n^2+t_{x}(G)+t_{x+2}(G)+\dots+t_{2k}(G))$, for $3\le l\le k+1$.
\begin{theorem}
For any constant $k$ and even number $4\le x \le \frac{4k}{3}+1$, the data structure of an $l$-layered graph $G$ can be constructed in $O(n^2+t_{x}(G)+t_{x+2}(G)+\dots+t_{2k}(G))$, for $3\le l\le k+1$.
\end{theorem}
\begin{proof}
Similar to Theorem~\ref{main_thm}, we must show that for any auxiliary graph of an induced subgraph of $G$, $H$, any table of $H$ must be of size $O(n^2+t_{x}(G)+t_{x+2}(G)+\dots+t_{2k}(G))$. Denote $V(H)=B_1\cup \dots \cup B_d$, and $i,j$ be indices that satisfy $1\le i<j\le d$. As shown in Theorem~\ref{thm:_size_of_T_tables}, for every $r\ge j-i$, $T_{i,j}(H)=O(n^2+t'_{2r}(H[B_i\cup \dots \cup B_j]))$, and also $T_{i,j}(H)=O(n^2+t''_{2r}(H[B_i\cup \dots \cup B_j]))$. Denote $m$ as the distance between $B_i$ and $B_j$ in $G$. We know that either the distance between $B_{j-1}$ and $B_j$ in $G$ is at most $\frac{m}{2}$, or the distance between $B_{i}$ and $B_{i+1}$ in $G$ is at most $\frac{m}{2}$, because there are $m$ layer transitions in $G$ between $G_i$ and $G_j$. Assume that the distance between $B_{j-1}$ and $B_j$ in $G$ is at most $\frac{m}{2}$ (the other options works similarly using the $t''$ bound instead of the $t'$ bound). We want to take a number $r$ such that a cycle counted in $t'_{2r}(H[B_i\cup \dots \cup B_j])$ has a size between $x$ and $2k$ when mapped back to $G$. Denote $p\le \frac{m}{2}$ the distance between $B_{j-1}$ and $B_{j}$. The size of a cycle mapped to $G$ will be $2m+2yp$ for some $y\ge 0$. So we want $x\le \frac{4k}{3}+1 \le 2m+2yp \le 2k$. $2m\le 2k$, so if $2m\ge \frac{4k}{3}+1$ we take $y=0$. Otherwise, $m< \frac{2k}{3}+0.5$, so $p< \frac{k}{3}+0.25$. Since $p$ is whole, $p< \frac{k}{3}+0.25$ implies $p\le \frac{k}{3}$. When we increase $y$ by 1, the value of $2m+2yp$ increases by $2p\le \frac{2k}{3}$, so beginning at $y=0$ and increasing $y$ by 1 each time, eventually we get $\frac{4k}{3} \le 2m+2yp \le 2k$ for some $y>0$, as needed.\\
We take the value $r=(j-i)+y$, and prove that cycles counted in $t'_{2r}(H[B_i\cup \dots \cup B_j])$ can be mapped injectively to cycles in $G$ of length $2m+2yp$, which is in the allowed range.\\
Each of these cycles contains two edges in every layer transition from $B_i$ to $B_{j-1}$, and $2+2y$ edges in the last layer transition. For each of these layer transitions (including the last), we can map the edges to non-intersecting paths in $G$ using Theorem~\ref{many-distinct-mildly-sparse-paths}. \\
We conduct the mapping iteratively similarly to Theorem~\ref{main_thm}. If there was one table reduction from $G$ to $H$, then some layer transition $B_g-B_{g+1}$ in $H$ is a result of a table reduction from layers $A_{i'}-\dots- A_{j'}$ in $G$. Using Theorem~\ref{many-distinct-mildly-sparse-paths}, we map the (at most $2k$) edges in $B_g-B_{g+1}$ to non intersecting paths in $A_{i'}-\dots- A_{j'}$. If $H$ is a result of $N$ table reductions from $G$, then map the edges iteratively, using Theorem~\ref{many-distinct-mildly-sparse-paths} for each table transition, until getting back to $G$. The cycle will be of size $2m+2py$ which follows from the calculations above. So each table is constructed in $O(n^2+t_{x}(G)+t_{x+2}(G)+\dots+t_{2k}(G))$, as needed.
\end{proof}

Using the reductions from Theorem~\ref{approx_listing}, we get an algorithm for range cycle enumeration with $\tilde{O}(n^2)$ preprocessing and $\tilde{O}(1)$ delay.
\begin{corollary}
    For constants $2k$ and $3\le i\le \frac{4k}{3}$, we can build an algorithm for cycle enumeration with $\tilde{O}(n^2)$ preprocessing and $\tilde{O}(1)$ delay. The algorithm will output cycles in the size range $[i,2k]$.
\end{corollary}
In the next section, we prove the boundary layer property for small constants $2,3$, which gives a $2k$-cycle listing algorithm for $k\le 8$ according to Theorem~\ref{main_thm}. We then explain why the algorithm doesn't work for larger constants.

\section{Proving the Boundary Layer Property for $s=2,3$}\label{sec:boundary_layer_proofs}
\begin{theorem}
The Boundary Layer Property is correct for $s=2$.
\end{theorem}
\begin{proof}
The proof is very similar to the methodology used in Theorem~\ref{difficult_case}. Since $s=2$, there is only one possible auxiliary graph for a graph $J=(V,E)$, $V=A_1\cup A_2\cup A_3$, which is $H=G^{T_{1,3}(J)}$. 
$G^{T_{1,3}(J)}=G^{T_{1,2}(H)}=H$, so we must prove the statement for the auxiliary graph $H$. We claim that if for a subgraph $H'$ of $H$, all degrees in $H'$ are at least $50\lambda(k)$, then there exists a path in $J$ from $A_1$ to $A_3$ in $J$ with length $2k-2$, for a fixed $k\ge 2$. If $k$ is even, then we can take a path $P$ in $H'$ between $B_1=A_1$ and $B_2=A_3$ of length $(k-1)$, due to the minimal degree. Then, for each edge $(b_1,b_2)\in P$ there are $c=10k$ options for $a_2$, so we can choose a different $A_2$ vertex for each edge in $H'$, creating a simple path in $J$. If $k$ is odd, we use the proof from Theorem~\ref{difficult_case} which shows there are 2 edges $(b_1',b_2'),(b_1'',b_2'')$ that share an $A_2$ vertex $a_2$ ($a_2\in T_{1,3}(J)[b_1',b_2']$ and $a_2\in T_{1,3}(J)[b_1'',b_2'']$). If $b_1'\ne b_1''$ we can take the path in $J$, $b_1'-a_2-b_1''$, and from $b_1''$ continue a simple path in $H$ of length $k-2$ edges in $H$. If $b_1'= b_1''$ then go from $b_1'$ to $b_2'$ then to $a_2$, then to $b_2''$, and then continue a path of $k-3$ in $H$. Each edge in $H$ has at least $2c=20k$ different options for a vertex in $A_2$, so we pick a different $A_2$ vertex for each such edge. This proves the existence of a $2k-s$ path in $J$ between $A_1$ and $A_3$, what we wanted.
\end{proof}
\begin{corollary}
    There is an $O(n^2+t_{2k}(G))$ $2k$-cycle listing algorithm for $k=5,6$. There is also an $O(n^2+t_{2k}(G)+t_{2k-1}(G))$ $(2k-1)$-cycles listing algorithm for $k=5,6$.
\end{corollary}
\begin{theorem}
The Boundary Layer Property is correct for $s=3$.
\end{theorem}
\begin{proof}
Given $J=(V,E)$, $V=A_1\cup A_2\cup A_3\cup A_4$, the graph $H$ in the definition of the boundary layer property is either $G^{T_{1,4}(J)}$, $G^{T_{1,3}(G^{T_{1,3}(J)})}$, or $G^{T_{1,3}(G^{T_{2,4}(J)})}$. 
\begin{lemma} \label{helper_lemma_3_paths}
    For all 3 of these options, an edge in the auxiliary graph maps (injectively) to $c$ internally disjoint $3$-paths.
\end{lemma}
\begin{proof}
If the auxiliary graph is $H=G^{T_{1,4}(J)}$, using Lemma~\ref{c_vertex-disjoint_paths}, the statement is trivially true.\\
If the auxiliary graph is $H=G^{T_{1,3}(G^{T_{1,3}(J)})}$, then an edge $(a_1,a_4)\in A_1\times A_4$ has at least $c$ values in $A_3$ such that $a_1-a_3-a_4$ is a $2$-path, or in other words, $(a_1,a_4)$ maps to at least $c$ colorful $2$-paths in $G^{T_{1,3}(J)}$. Each edge $(a_1,a_3)$ ($a_3\in A_3$) in a $2$-path between $a_1$ and $a_4$, has at least $c$ options for a vertex in $A_2$. So for $c$ of these edges we can take a distinct vertex in $A_2$, by choosing for the $i$'th edge, $(a_1,a_3^i)$ a vertex in $A_2$ which has not yet been chosen. This maps the paths to internally disjoint $3$-paths, as needed.\\
The proof for the auxiliary graph $H=G^{T_{1,3}(G^{T_{2,4}(J)})}$, is the same as the proof of $H=G^{T_{1,3}(G^{T_{1,3}(J)})}$, but from the other direction.
\end{proof}
We assume the assumption of the boundary layer property, which is that each vertex in $H'\subseteq H$ has a degree of at least $50\lambda(k)$. 
We now break the proof into cases based on the value of $k$ mod $3$. Denote $J'$ as a subgraph of $J$ containing only edges in paths mapped to by edges in $H'$.
\paragraph{Case 1: $k=0$ mod $3$:}
We wish to prove that for $k\ge 3$, $k=0$ mod $3$, there exists a path of length $2k-3$ from $A_1$ to $A_4$ in $J'$. Denote $V(H)=B_1\cup B_2$. Each vertex has a degree of at least $50\lambda(k)$, so we can greedily create a $\frac{2k}{3}-1$ path in $H$. The starting vertex is in $B_1$, so the last vertex in the path is in $B_2$. We now map this into a $(2k-3)$-path in $J'$. Denote the edges in the path $e_1,e_2,e_3,...$. We iterate on the edges going forward, turning each edge into a $3$-path in $J'$. For $e_i=(b_1,b_2)$, there are at least $c=10k$ $3$-paths $b_1-a_2'-a_3'-b_2$ which don't intersect in $A_2$ or $A_3$ according to Lemma~\ref{helper_lemma_3_paths}. The amount of vertices already selected in $A_2,A_3$ is at most $2$ times the overall number of edges in $H$, which is bounded by $2\cdot (\frac{2k}{3} ) <2k$. Because the paths $b_1-a_2'-a_3'-b_2$ don't intersect in $A_2$ or $A_3$, each vertex in $A_2$ or $A_3$ already selected is in at most one of these paths. So overall at most $2k$ paths intersect with previously selected vertices, which leaves us $8k$ possible options for vertices in $A_2,A_3$, and we choose one of those options. In this way the paths we choose for each edge $e_i$ don't intersect in intermediate vertex layers, so the path is simple, as needed.
\paragraph{Case 2: $k=1$ mod $3$:}
Each edge $(b_1,b_2)$ represents at least $c$ internally disjoint $3$-paths between $b_1$ and $b_2$ in $J'$. If there is a vertex $a_2\in A_2$ such that 2 different edges in $H'$, $(b_1',b_2')$ and $(b_1'',b_2'')$ such that $b_1'\ne b_1''$ and both edges represent a path which passes through $a_2$, then we can take the path $b_1'-a_2-b_1''$, and from $b_1''$ create a $(\frac{2(k-1)}{3}-1)$-path in $H'$, $P$. Then we map $P$ into a path in $J'$ in the same way as case 1, adding $a_2$ to the already selected vertices. Adding $b_1'-a_2-b_1''$ to $P$, this will be a $(2k-3)$-path, as needed. Additionally, if there is a vertex $a_3\in A_3$ such that 2 different edges in $H$, $(b_1',b_2')$ and $(b_1'',b_2'')$ such that $b_2'\ne b_2''$ and both edges represent a path which passes through $a_3$, we can take a path from some neighbor $b_1$ of $b_2'$ in $H$, then from $b_2'$ to $a_3$, to $b_2''$. From there, continue the path with a $(\frac{2(k-1)}{3}-2)$-path to another vertex in $B_2$. We map the edges in $H$ to internally disjoint $3$-paths in $J'$ as before, adding $a_3$ to the selected vertices.\\
\begin{lemma}
    If neither of these occur, $e(H')$, which is at least $50\lambda(k)n$ is smaller than $e[A_2,A_3]$ in $J'$.
\end{lemma}
\begin{proof}
If neither of these occur, that means each vertex $a_2\in A_2$ which is in some path represented by an edge in $H'$ has at most one neighbor in $B_1$, and each vertex $a_3\in A_3$ which is in some path represented by an edge in $H'$ has at most one neighbor in $B_2$. For each edge $(b_1,b_2)$, take one of the $3$-paths in $J'$ it maps to, $b_1-a_2-a_3-b_2$. There is no other edge $(b_1',b_2')$ such that $b_1'-a_2-a_3-b_2'$ is one of the $3$-paths $(b_1',b_2')$ maps to. So we can map each edge $(b_1,b_2)$ to $c$ edges $(a_2,a_3)\in A_2\times A_3$, such that no edge $(a_2,a_3)\in A_2\times A_3$ has two edges in $H'$ that map to it. Mapping edges $(b_1,b_2)$ to $c$ edges $(a_2,a_3)\in A_2\times A_3$ such that the mapping of each two edges is different, proves that $c\cdot e(H)\le e[A_2,A_3]$
\end{proof}
Furthermore, the number of edges that are mapped to by edges in $H'$ is also at least $50\lambda(k)n$ (same proof as the above lemma). The graph $J'[A_2\cup A_3]$ only includes edges mapped to by edges in $H$. By pruning $J'[A_2\cup A_3]$ and iteratively removing all vertices with degree below $10\lambda(k)$, we get a graph with at least $40\lambda(k)n$ edges that has a minimal degree of $10\lambda(k)$. With these edges we (greedily) find a $(2k-5)$-path that starts at $A_2$ and ends in $A_3$. Taking the starting vertex $a_2$, we know it has a neighbor $b_1\in B_1$, because we only chose edges in $A_2\times A_3$ that are mapped to by edges in $H$. Similarly, $a_3$ has a neighbor $b_2\in B_2$. So taking the path from $b_1$ to $a_2$, from $a_2$ take the $(2k-5)$-path to $a_3$, then take the edge to $b_2$, The path is length $2k-3$, as needed.
\paragraph{Case 3: $k=2$ mod $3$:}
Suppose there exist vertices $a_2', a_2'' \in A_2$ and three edges  $(b_1', b_2'), (b_1'', b_2''), (b_1''', b_2''') \in H'$ with pairwise distinct first endpoints $b_1', b_1'', b_1'''$. If the paths represented by $(b_1', b_2')$ and $(b_1'', b_2'')$ both pass through $a_2'$, while those represented by $(b_1'', b_2'')$ and $(b_1''', b_2''')$ both pass through $a_2''$, then we can construct the path $b_1' - a_2' - b_1'' - a_2'' - b_1'''$. Then from $b_1'''$ create a $(\frac{2(k-2)}{3}-1)$-path in $H'$. Then we map the path in $H'$ into a path in $J'$ in the same way as case 1, adding $a_2',a_2''$ to the already selected vertices. This will be a $(2k-3)$-path, as needed.\\
Similarly, suppose there exist vertices $a_3', a_3'' \in A_3$ and three edges $(b_1', b_2'), (b_1'', b_2''), (b_1''', b_2''') \in H'$ with pairwise distinct second endpoints $b_2', b_2'', b_2'''$. If the paths represented by $(b_1', b_2')$ and $(b_1'', b_2'')$ both pass through $a_3'$, while those represented by $(b_1'', b_2'')$ and $(b_1''', b_2''')$ both pass through $a_3''$, then we can construct the path $b_2' - a_3' - b_2'' - a_3'' - b_2'''$. The edge $(b_1',b_2')$ represents at least $c$ internally disjoint paths between $b_1'$ and $b_2'$, so we can take a path from $b_1'$ to $b_2'$, $b_1'-a_2-a_3'''-b_2$, that doesn't pass through $a_3'$ or $a_3''$. So we start the $2k-3$-path with $b_1'-a_2-a_3'''-b_2'-a_3'-b_2''-a_3''-b_2'''$. Then from $b_2'''$ continue the path with a $(\frac{2(k-2)}{3}-2)$-path to another vertex in $B_2$. We map the edges in $H$ to internally disjoint $3$-paths in $J$ as before, adding $a_3',a_3''$ to the selected vertices to make a $3(\frac{2(k-2)}{3}-2)+7=2k-3$ path.
\begin{lemma}
If neither of the above conditions are true, then there are at least $50\lambda(k)n$ edges in $A_2\times A_3$ that are mapped to by edges in $H$ (using the same mapping described in case 2).
\end{lemma}
\begin{proof}
Take an edge $(b_1,b_2)\in H$, and take two $3$-paths in $J'$ which the edge maps to, $b_1-a_2'-a_3'-b_2$ and $b_1-a_2''-a_3''-b_2$. If the number of neighbors $a_2'$ has in $B_1$ is 3 or more, then the number of neighbors $a_2''$ has in $B_1$ is 1 or 2. This is because if both vertices have 3 or more neighbors in $B_1$, then we can take a neighbor of $a_2'$, $b_1'\ne b_1$, and a neighbor of $a_2''$, $b_1''\ne b_1,b_1'$, which creates the path $b_1'-a_2'-b_1-a_2''-b_1''$, contradicting our assumption. Similarly, if the number of neighbors $a_3'$ has in $B_2$ is 3 or more, then the number of neighbors $a_3''$ has in $B_2$ is 1 or 2. Taking the (at least) $c$ $3$-paths represented by the edge $(b_1,b_2)$, at most one of them has a vertex in $A_2$ which has 3 or more neighbors in $B_1$, and at most one of them has a vertex in $A_3$ which has 3 or more neighbors in $B_2$. So for the edge $(b_1,b_2)$, there are at least $(c-2)$ $3$-paths in $J'$, $b_1-a_2-a_3-b_2$, such that $a_3$ has at most 2 neighbors in $B_2$, and $a_2$ has at most 2 neighbors in $B_1$. Take the mapping $\rho$ of each edge $(b_1,b_2)\in E(H)$ to the edges in $(a_2,a_3)\in A_2\times A_3$ such that $(a_2,a_3)\in \rho(b_1,b_2)$ if $b_1-a_2-a_3-b_2$ is a path in $J'$, $a_3$ has at most 2 neighbors in $B_2$, and $a_2$ has at most 2 neighbors in $B_1$. Each edge $e$ in $H$ is mapped to at least $c-2$ edges in $A_2\times A_3$, and each such edge in $A_2\times A_3$ has at most 4 edges mapped to it ($2$ possible values for $b_1$ and $2$ possible values for $b_2$). So there are at least $\frac{50\lambda(k)(c-2)n}{4}>50\lambda(k)n$ such edges in $A_2\times A_3$, as needed.
\end{proof}
From here, the graph $J'[A_2\cup A_3]$ contains only edges in $A_2\times A_3$ which are mapped to by edges in $H$ in the mapping defined above. We've shown that there are $50\lambda(k)n$ such edges, so by pruning $J'[A_2\cup A_3]$ and iteratively removing all vertices with degree below $10\lambda(k)$, we get a graph with at least $40\lambda(k)n$ edges that has a minimal degree of $10\lambda(k)$.
Now we can greedily construct a $2k-5$ path in $J'[A_2\cup A_3]$ which starts at $A_2$ and ends at $A_3$. Then take an edge from the first vertex to $B_1$ and from the last vertex to $B_2$. Such vertices exist because all edges $(a_2,a_3)$ we use are mapped to by edges in $H'$, so there is an edge $(b_1,b_2)\in E(H')$ such that $b_1-a_2-a_3-b_2$ is a path in $J'$. This means that each vertex in the path has a neighbor in either $B_1$ or $B_2$, depending on whether the vertex is in $A_2$ or $A_3$.
Finally we get a $(2k-3)$-path between $B_1$ and $B_2$, completing the proof.
\end{proof}
\begin{corollary}
    There is an $O(n^2+t_{2k}(G))$ cycle listing algorithm for $k=7,8$. There is also an $O(n^2+t_{2k}(G)+t_{2k-1}(G))$ $(2k-1)$-cycles listing algorithm for $k=7,8$.
\end{corollary}
We would now want to prove the Boundary Layer Property for $s\ge 4$. Unfortunately, the property becomes false for $s=4$.

\subsection{Boundary Layer Property is False for $s=4$}\label{sec:s4_barrier}

The intuition behind this counterexample is taking a 5-layer graph with layers $A_1-A_2-A_3-A_4-A_5$, such that each of the vertices in $A_2$ and $A_4$ only have one neighbor in $A_3$ and $A_1/A_5$, but for each pair $(a_1,a_5)\in A_1\times A_5$ there are paths from $a_1$ to $a_5$ going through every vertex in $A_3$. So the number of edges in $G^{T_{1,5}(J)}$ will be high, but a path cannot go from $A_1$ to $A_2$ and back to $A_1$ because each vertex in $A_2$ has one neighbor in $A_1$. Similarly, a path cannot go from $A_3$ to $A_2$ and back to $A_3$, a path cannot go from $A_3$ to $A_4$ and back to $A_3$, and a path cannot go from $A_5$ to $A_4$ and back to $A_5$. These restrictions will create limitations on the lengths of paths in $J$, which will contradict the boundary layer property.
We construct a 5-layer graph counterexample to show that the Boundary Layer Property is wrong for $s = 4$. 
Let $J = (V, E)$, $V = A_1 \cup A_2 \cup A_3 \cup A_4 \cup A_5$ be a 5-layer graph. The layer sizes are as follows:
\[ |A_1| = |A_5| = n^{0.8}, \quad |A_3| = n^{0.2}, \quad \text{and} \quad |A_2| = |A_4| = n. \]
For each layer $j \in [5]$, let $a_j^i$ denote the $i$-th vertex in $A_j$. The edge set $E$ is defined by the following connection rules for all $i \in [n]$:
\begin{enumerate}
    \item \textbf{Layer 1 to Layer 2:} Each vertex $a_2^i \in A_2$ has a single neighbor in $A_1$. Let $r = \lceil i / n^{0.2} \rceil$. The pair $(a_1^r, a_2^i)$ forms an edge in $J$.
    \item \textbf{Layer 4 to Layer 5:} Each vertex $a_4^i \in A_4$ has a single neighbor in $A_5$. Let $r = \lceil i / n^{0.2} \rceil$. The pair $(a_4^i, a_5^r)$ forms an edge in $J$.
    \item \textbf{Internal Layer Routing:} Let $l = (i \bmod n^{0.2})$. The pairs $(a_2^i, a_3^l)$ and $(a_3^l, a_4^i)$ form edges in $J$.
\end{enumerate}

By this construction, given any pair $a_1^i \in A_1$ and $a_3^j \in A_3$, there exists a single vertex in $A_2$ completing a path $a_1^i - a_2 - a_3^j$, given by index $a_2^{(i-1)n^{0.2} + j}$. Analogously, for any pair $a_5^i \in A_5$ and $a_3^j \in A_3$, there is a single vertex in $A_4$ yielding the path $a_3^j - a_4 - a_5^i$, given by index $a_4^{(i-1)n^{0.2} + j}$.

\begin{lemma}
Every simple path connecting a vertex in $A_1$ to a vertex in $A_5$ is sparse.
\end{lemma}

\begin{proof}
Let $P = a_1 - a_2 - a_3 - a_4 - a_5$ be a path in $J$ from $A_1$ to $A_5$. By our edge definitions, there is exactly one colorful path between $a_1$ and $a_3$. Because $a_3$ is the only neighbor of $a_4$ within the layer $A_3$, there is no other colorful path between $a_1$ and $a_4$. Symmetrically, there is a single colorful path between $a_3$ and $a_5$, which implies there is no other colorful path between $a_2$ and $a_5$. This proves the path is sparse.
\end{proof}

Observe that for every pair $(a_1, a_3) \in A_1 \times A_3$, a single $A_1$-$A_3$ path exists; similarly, a single $A_3$-$A_5$ path exists for every pair $(a_3, a_5) \in A_3 \times A_5$. Consequently, for any pair of vertices $(a_1, a_5) \in A_1 \times A_5$, there exists a colorful $5$-vertex path passing through $a_3$ for each of the $|A_3| = n^{0.2}$ available vertices in $A_3$. 

This guarantees that the number of distinct paths between $a_1$ and $a_5$ is exactly $n^{0.2}$. This collection of mildly sparse paths yields an edge $(a_1, a_5)$ in the auxiliary graph $G^{T_{1,5}(J)}$. The total number of edges in $G^{T_{1,5}(J)}$ is therefore:
\[ |A_1| \cdot |A_5| = (n^{0.8})^2 = n^{1.6}. \]
For any sufficiently large $n$, this amount is larger than the threshold $100\lambda(k)n$ for any fixed constant $k$. If the Boundary Layer Property held for $s=4$, this would imply the existence of a simple path of length $2k-4$ between $A_1$ and $A_5$ for any constant $k \ge 4$. If we select $k = 5$, there exists a simple path of length $2(5)-4 = 6$ between $A_1$ and $A_5$. We show this is impossible.

\begin{lemma}
The length of any simple path beginning in $A_1$ or $A_5$ and terminating in $A_1 \cup A_5$ must be a multiple of $4$.
\end{lemma}

\begin{proof}
Let $P = v_1-v_2 -v_3 \dots$ be a simple path originating at a vertex $v_1 \in A_1$. We track the allowed sequence of layers across the first four edges of $P$:
\begin{itemize}
    \item The first edge must transition to layer $A_2$, so $v_2 \in A_2$.
    \item Because $v_2\in A_2$ has exactly one neighbor in $A_1$ (which is the starting vertex $v_1$), the second edge must transition to layer $A_3$, so $v_3 \in A_3$.
    \item From $A_3$, the third edge can lead either to layer $A_4$ or return to layer $A_2$.
    \begin{itemize}
        \item If $v_4 \in A_4$, the fourth edge must transition to $A_5$, since $v_4\in A_4$ has no neighbors in $A_3$ other than $v_2$. Thus $v_5 \in A_5$.
        \item If $v_4 \in A_2$, the fourth edge must transition back to $A_1$ because $v_4\in A_2$ has exactly one neighbor in $A_3$ (which is $v_3)$. Thus $v_5 \in A_1$.
    \end{itemize}
\end{itemize}
Thus, after exactly $4$ edges, any simple path originating in $A_1$ must land in either $A_1$ or $A_5$. We also notice it cannot visit those layers before passing 4 edges. By symmetry, an identical argument holds for any path starting in $A_5$. By induction, a simple path can end in the layers $A_1 \cup A_5$ if and only if its length is $4r$ for some natural $r$.
\end{proof}

Because any path between $A_1$ and $A_5$ must have a length divisible by $4$, there cannot exist a simple path of length $6$ between these boundary layers. This provides a direct contradiction, proving that the Boundary Layer Property is false for $s=4$.

\subsection{A graph for which the algorithm is too slow}
From here, we can create a $(k+1)$-layer graph $G$ for $k\ge 9$, $k=1 \mod 2$, such that the layer transitions from $A_1$ to $A_5$ are equivalent to the ones in the graph defined in the above section, and the layer transitions from $A_5$ to $A_9$ are the same. The auxiliary graph $H=G^{T_{1,5}(G^{T_{5,9}(G)})}$ satisfies $|T_{1,3}(H)|=n^{2.4}$, but it's untrue that there are $\Omega(n^{2.4})$ $2k$-cycles in $G$ for $2k$. Because $2k= 2\mod 4$, there are 0 such cycles. So the method of bounding all tables by $O(n^2+t_{2k})$ in $(k+1)$ layer graphs stops working for $k\ge 9$. If the other layer transitions are empty, then the algorithm does not run in $O(n^2+t_{2k}(G))$ as needed for fast cycle enumeration.

\bibliographystyle{plain}  
\bibliography{cycles}

\end{document}